\documentclass[12pt]{article}
\usepackage[margin=1in]{geometry}
\usepackage[small]{titlesec}
\usepackage{natbib,amsmath,amssymb,amsthm,setspace,paralist,graphicx,algorithm,booktabs,rotating,subfig,url}
\usepackage[noend]{algorithmic}
\AtBeginDocument{%
\abovedisplayskip=5pt plus 3pt minus 2pt
\belowdisplayskip=\abovedisplayskip
\abovedisplayshortskip=0pt plus 3pt
\belowdisplayshortskip=\belowdisplayskip}
\titlespacing*{\section}{0pt}{*2}{*1}
\titlespacing*{\subsection}{0pt}{*2}{*1}
\bibpunct{(}{)}{;}{a}{}{,}
\DeclareCaptionStyle{ruled}{labelfont=bf,labelsep=period}
\newtheorem{thm}{Theorem}
\newtheorem{lem}{Lemma}
\newtheorem{prop}{Proposition}

\DeclareMathOperator*\argmin{\arg\min}
\DeclareMathOperator\sgn{sgn}
\DeclareMathOperator\supp{supp}

\DeclareMathOperator\CV{CV}
\def\br{\mathbf{r}}
\def\bw{\mathbf{w}}
\def\bx{\mathbf{x}}
\def\by{\mathbf{y}}
\def\bz{\mathbf{z}}
\def\bA{\mathbf{A}}

\def\bC{\mathbf{C}}
\def\bE{\mathbf{E}}

\def\bX{\mathbf{X}}
\def\bZ{\mathbf{Z}}
\def\balpha{\boldsymbol{\alpha}}
\def\bbeta{\boldsymbol{\beta}}
\def\bgamma{\boldsymbol{\gamma}}
\def\bdelta{\boldsymbol{\delta}}
\def\bfeta{\boldsymbol{\eta}}
\def\ve{\varepsilon}
\def\bve{\boldsymbol{\varepsilon}}
\def\btheta{\boldsymbol{\theta}}
\def\bGamma{\boldsymbol{\Gamma}}

\def\bSigma{\boldsymbol{\Sigma}}
\def\bzero{\mathbf{0}}
\def\cB{\mathcal{B}}
\def\TP{\mathrm{TP}}
\def\TN{\mathrm{TN}}
\def\FP{\mathrm{FP}}
\def\FN{\mathrm{FN}}
\let\what\widehat
\def\mid{\,|\,}

\def\pheq{\mathrel{\phantom{=}}}

\allowdisplaybreaks
\begin{document}
{\setstretch{1.17}
\title{Regularization Methods for High-Dimensional Instrumental Variables Regression With an Application to Genetical Genomics}
\author{Wei Lin, Rui Feng, and Hongzhe Li}
\date{}
\maketitle
\footnotetext{Wei Lin is Postdoctoral Researcher (E-mail: \emph{weilin1@mail.med.upenn.edu}), Rui Feng is Assistant Professor (E-mail: \emph{ruifeng@mail.med.upenn.edu}), and Hongzhe Li (E-mail: \emph{hongzhe@upenn.edu}) is Professor, Department of Biostatistics and Epidemiology, Perelman School of Medicine, University of Pennsylvania, Philadelphia, PA 19104. This research was supported in part by NIH grants CA127334, GM097505, and GM088566. The authors thank the Co-Editor, an Associate Editor, and three referees for helpful comments that have led to substantial improvement of the article.}

\begin{abstract}
In genetical genomics studies, it is important to jointly analyze gene expression data and genetic variants in exploring their associations with complex traits, where the dimensionality of gene expressions and genetic variants can both be much larger than the sample size. Motivated by such modern applications, we consider the problem of variable selection and estimation in high-dimensional sparse instrumental variables models. To overcome the difficulty of high dimensionality and unknown optimal instruments, we propose a two-stage regularization framework for identifying and estimating important covariate effects while selecting and estimating optimal instruments. The methodology extends the classical two-stage least squares estimator to high dimensions by exploiting sparsity using sparsity-inducing penalty functions in both stages. The resulting procedure is efficiently implemented by coordinate descent optimization. For the representative $L_1$ regularization and a class of concave regularization methods, we establish estimation, prediction, and model selection properties of the two-stage regularized estimators in the high-dimensional setting where the dimensionality of covariates and instruments are both allowed to grow exponentially with the sample size. The practical performance of the proposed method is evaluated by simulation studies and its usefulness is illustrated by an analysis of mouse obesity data. Supplementary materials for this article are available online.\medskip

\noindent\emph{Running title:} Regularization Methods for High-Dimensional Instrumental Variables Regression\medskip

\end{abstract}}\clearpage

\section{Introduction}
Genome-wide studies have been widely conducted to search tens of thousands of gene expressions or hundreds of thousands of single nucleotide polymorphisms (SNPs) to detect associations with complex traits. By measuring and analyzing gene expressions and genetic variants on the same subjects, genetical genomics studies provide an integrative and powerful approach to addressing fundamental questions in genetics and genomics at the functional level. In these studies, gene expression levels are viewed as quantitative traits that are subject to genetic analysis for identifying expression quantitative trait loci (eQTLs), in order to understand the genetic architecture of gene expression variation. The increasing availability of high-throughput genetical genomics data sets opens up the possibility of jointly analyzing gene expression data and genetic variants in exploring their associations with complex traits, with the goal of identifying key genes and genetic markers that are potentially causal for complex human diseases such as obesity, heart disease, and cancer \citep{Emil:Thor:Gudm:Zhan:Leon:gene:2008}.

Although in the past decade gene expression profiling has led to the discovery of many gene signatures that are highly predictive for clinical outcomes, the effort of using these findings to dissect the genetics of complex traits and diseases is often compromised by the critical issue of confounding. It is well known that many factors, such as unmeasured variables, experimental conditions, and environmental perturbations, may exert pronounced influences on gene expression levels, which may in turn induce spurious associations and/or distort true associations of gene expressions with the response of interest \citep{Leek:Stor:capt:2007,Fusi:Steg:Lawr:join:2012}. Moreover, due to the difficulty of high dimensionality, empirical studies are mostly based on marginal models, which are especially prone to variability caused by pleiotropic effects and dependence among genes. Ignoring these confounding issues tends to produce results that are both biologically less interpretable and less reproducible across independent studies.

Instrumental variables (IV) methods provide a practical and promising approach to control for confounding in genetical genomics studies, with genetic variants playing the role of instruments. This approach exploits the reasonable assumption that the genotype is assigned randomly, given the parents' genes, at meiosis and independently of possible confounding factors, and affects a clinical phenotype only indirectly through some intermediate phenotypes. In observational epidemiology, Mendelian randomization has been proposed as a class of methods for using genetic variants as instruments to assess the causal effect of a modifiable phenotype or exposure on a disease outcome; see, for example, \citet{Lawl:Harb:Ster:Timp:Dave:mend:2008} for a review. The primary scenario considered in this context, however, involves only one exposure variable and requires the existence of a genetic variant whose relationship with the exposure has been well established. Thus, the methodology intended for Mendelian randomization is typically not applicable to genetical genomics studies, where the number of expression phenotypes is exceedingly large and the genetic architecture of each phenotype may be complex and unknown.

IV models and methods have been extensively studied in the econometrics literature, where the problem is often cast in the framework of simultaneous equation models \citep{Heck:dumm:1978}. It has been shown that classical IV estimators such as the two-stage least squares (2SLS) estimator and the limited information maximum likelihood (LIML) estimator are consistent only when the number of instruments increases slowly with the sample size \citep{Chao:Swan:cons:2005,Hans:Haus:Newe:esti:2008}. Recent developments have introduced regularization methods to mitigate the overfitting problem in high-dimensional feature space by exploiting the sparsity of important covariates, thereby improving the performance of IV estimators substantially. \citet{Cane:lass:2009} considered penalized generalized method of moments (GMM) with the bridge penalty for variable selection and estimation in the classical setting of fixed dimensionality. \citet{Gaut:Tsyb:high:2011} developed a Dantzig selector--type procedure to select important covariates and estimate the noise level simultaneously in high-dimensional IV models where the dimensionality may be much larger than the sample size. Under the assumption that the important covariates are uncorrelated with the regression error, \citet{Fan:Liao:endo:2012} proposed a penalized focused GMM method based on a nonsmooth loss function to perform variable selection and achieve oracle properties in high dimensions. All the aforementioned methods, however, do not exploit the sparsity of instruments and hence are still facing the dimensionality curse of many instruments. Another active line of research in the econometrics literature has been concerned with the use of regularization and shrinkage methods for estimating optimal instruments in the context of estimating a low-dimensional parameter; see, for example, \citet{Okui:inst:2011} and \citet{Carr:regu:2012}. Of particular interest is the recent work of \citet{Bell:Chen:Cher:Hans:spar:2012}, where Lasso-based methods were applied to form first-stage predictions and estimate optimal instruments in an IV model with many instruments but itself of fixed dimensionality.

In this article, we focus on the application of high-dimensional sparse IV models to genetical genomics, where we are interested in associating gene expression data with a complex trait to identify potentially causal genes by using genetic variants as instruments. Motivated by this important application, we propose a two-stage regularization (2SR) methodology for identifying and estimating important covariate effects while selecting and estimating optimal instruments. Our approach merges the two independent lines of research mentioned above and provides a regularization framework for IV models that accommodate covariates and instruments both of high dimensionality. Specifically, the proposed procedure consists of two stages: In the first stage the covariates are regressed on the instruments in a regularized multivariate regression model and predictions are obtained, and in the second stage the response of interest is regressed on the first-stage predictions in a regularized regression model to perform final variable selection and estimation. In each stage, a sparsity-inducing penalty function is employed to yield desirable statistical properties and practical performance. The methodology can be viewed a high-dimensional extension of the 2SLS method, allowing the use of regularization methods to address the high-dimensional challenge in both stages.

Several key features make the proposed methodology especially appealing for the kind of applications we consider in this article. First, unlike marginal regression models commonly used in empirical studies that analyze a few variables at a time, our method allows for the joint modeling and inference of high-dimensional genetical genomics data. In view of the fact that many genes interact with each other and contribute together to a complex trait or disease, joint modeling is crucial for correcting bias and controlling false positives due to possible confounding. Second, our method requires neither a specification of a small set of important instruments nor an importance ranking among the instruments; instead, we consider the estimation of optimal instruments as a variable selection problem and allow the procedure to choose important instruments based on the data. Third, the proposed implementation by coordinate descent optimization is computationally very efficient and has provable convergence properties, therefore bypassing the computational obstacles faced by traditional model selection methods. Finally, we rigorously justify our method for the representative $L_1$ regularization and a class of concave regularization methods in the high-dimensional setting where the dimensionality of covariates and instruments are both allowed to grow exponentially with the sample size. Through the theoretical analysis, we explicate the impact of dimensionality and the role of regularization, and provide strong performance guarantees for the proposed method.

The remainder of this article is organized as follows. Section 2 introduces the high-dimensional sparse IV model. The 2SR methodology and implementation are presented in Section 3. Theoretical properties of the regularized estimators are investigated in Section 4. We illustrate our method by simulation studies in Section 5 and an analysis of mouse obesity data in Section 6. We conclude with some discussion in Section 7. Proofs are relegated to the Appendix and Supplementary Material.

\section{Sparse Instrumental Variables Models}
Suppose we have a quantitative trait or clinical phenotype $y$, a $p$-vector of gene expression levels $\bx$, and a $q$-vector of numerically coded genotypes $\bz$. In reality, there may be a sufficient set of unobserved confounding phenotypes $\bw$ that act as proxies for the long-term effects of environmental exposures and/or the state of the microenvironment of the cells or tissues within which the biological processes occur. These phenotypes are likely to have strong influences on gene expression levels while contributing substantially to the clinical phenotype. Figure \ref{fig:diagram}(a) illustrates the confounding between $\bx$ and $y$ with an example of six variables. If an ordinary regression analysis is to be applied, the effects of $x_1$ and $x_2$ on $y$ would be seriously confounded by $w$, resulting in a spurious association or effect modification.

\begin{figure}\centering
\subfloat[]{\includegraphics[width=.32\textwidth]{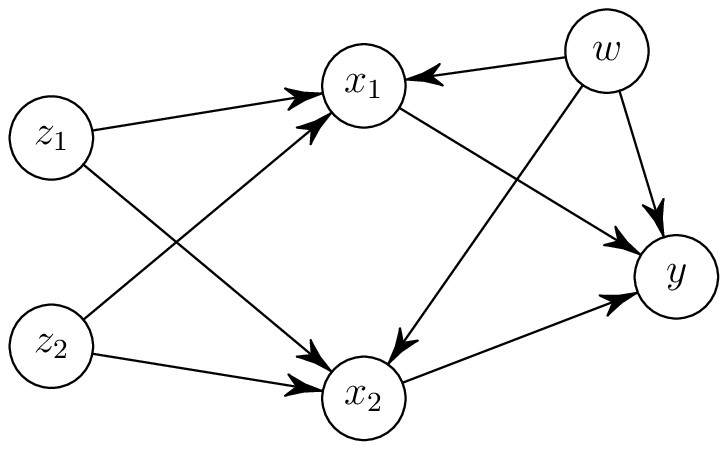}}\hfill
\subfloat[]{\includegraphics[width=.32\textwidth]{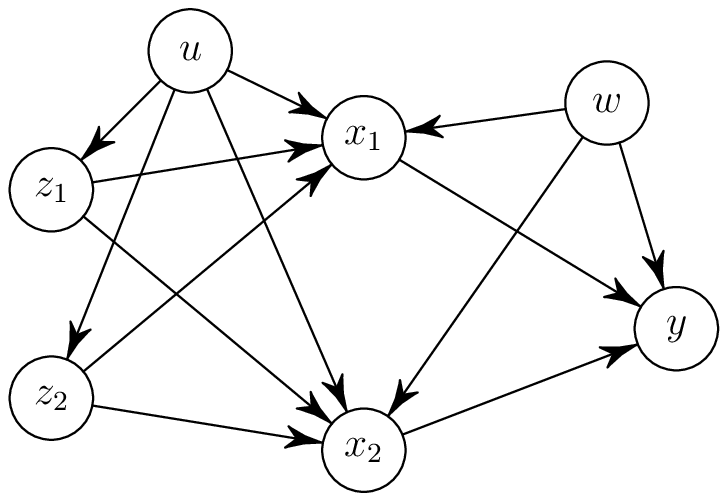}}\hfill
\subfloat[]{\includegraphics[width=.32\textwidth]{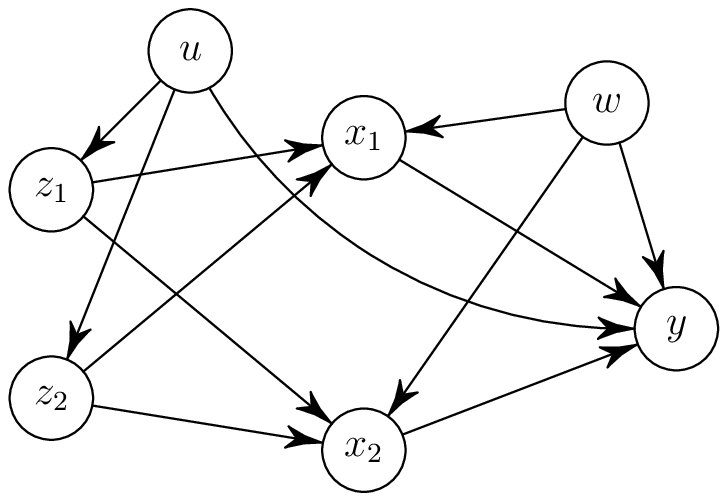}}
\caption{Causal diagrams showing the relationships between two genotypes $z_1$ and $z_2$, two gene expression levels $x_1$ and $x_2$, a clinical phenotype $y$, an unobserved phenotype $w$ that confounds the associations between gene expression levels and the clinical phenotype, and an unobserved variable $u$ representing possibly present population substructure. The population substructure (a) is not present, (b) affects genotypes and gene expression levels, or (c) affects genotypes and the clinical phenotype.}\label{fig:diagram}
\end{figure}

One way of controlling for the confounding due to $\bw$ is through the use of the genotype $\bz$ as instruments. In order for $\bz$ to be valid instruments, the following conditions must be satisfied \citep{Dide:Meng:Shee:assu:2010}:
\begin{compactenum}[\indent 1.]
\item The genotype $\bz$ is (marginally) independent of the confounding phenotype $\bw$ between $\bx$ and $y$;
\item The genotype $\bz$ is not (marginally) independent of the intermediate phenotype $\bx$;
\item Conditionally on $\bx$ and $\bw$, the genotype $\bz$ and the clinical phenotype $y$ are independent.
\end{compactenum}

The above conditions are not easily testable from the observed data, but can often be justified on the basis of plausible biological assumptions. Condition 1 is ensured by the usual assumption that the genotype is assigned at meiosis randomly, given the parents' genes, and independently of any confounding phenotype. Condition 2 requires that the genetic variants be reliably associated with the gene expression levels, which is often demonstrated by \emph{cis}-eQTLs with strong regulatory signals. Condition 3 requires that the genetic variants have no direct effects on the clinical phenotype and can affect the latter only indirectly through the gene expression levels. Owing to the large pool of gene expressions included in genetical genomics studies, the possibility of a strong indirect effect is greatly reduced and hence this condition is also mild and tends to be satisfied in practice.

We discuss here more on these assumptions for genetical genomics data and possible biological complications.  Population stratification is a major concern in genome-wide association studies, where the presence of subpopulations with different allele frequencies and different distributions of quantitative traits or risks of disease can lead to spurious associations \citep{Lin:Zeng:corr:2011}. Two typical scenarios for the impact of population stratification are illustrated in Figure \ref{fig:diagram}(b) and (c). In Figure \ref{fig:diagram}(b), all three conditions for valid IVs are still satisfied, although the population substructure, represented by an unobserved variable $u$, may strengthen or weaken the associations between the genotype $\bz$ and gene expression levels $\bx$ required by Condition 2. In Figure \ref{fig:diagram}(c), Condition 3 is violated because conditioning on $\bx$ and $w$ alone is insufficient to guarantee the independence of the genotype $\bz$ and the clinical phenotype $y$. To deal with possible population stratification, one can regress out the principal components calculated from the genotype data in clinical phenotype regression and gene expression regressions. We also require that the tissue where the gene expressions are measured be relevant to the clinical phenotype. Condition 3 assumes that the genetic variants have no direct effects on the clinical phenotype but manifest their effects through expressions in the relevant tissue. Using a phenotype-irrelevant tissue can potentially lead to violation of Condition 3. It is important, however, to note that strong instruments, a majority of which are most likely \emph{cis}-eQTLs, play a predominant role in our methodology. Recent studies have revealed that these \emph{cis}-eQTLs and their effect sizes are highly conserved across human tissues and populations \citep{Gori:tiss:2012,Stra:Mont:Dima:Part:Steg:patt:2012}. This fact helps to lessen the risks of potential assumption violations, although great care should be exercised in justifying the assumptions on a case-by-case basis. See, for example, \citet{Dide:Shee:mend:2007} and \citet{Lawl:Harb:Ster:Timp:Dave:mend:2008} for more discussion on the complications in Mendelian randomization studies.

Suppose we have $n$ independent observations of $(y,\bx,\bz)$. Denote by $\by$, $\bX$, and $\bZ$, respectively, the $n\times1$ response vector, the $n\times p$ covariate matrix, and the $n\times q$ genotype matrix. Using the genotypes as instruments, we consider the following linear IV model for the joint modeling of the data $(\by,\bX,\bZ)$:
\begin{align}
\begin{split}\label{eq:ivmodel}
\by&=\bX\bbeta_0+\bfeta,\\
\bX&=\bZ\bGamma_0+\bE,
\end{split}
\end{align}
where $\bbeta_0$ and $\bGamma_0$ are a $p\times1$ vector and a $q\times p$ matrix, respectively, of regression coefficients, and $\bfeta=(\eta_1,\dots,\eta_n)^T$ and $\bE=(\bve_1,\dots,\bve_n)^T$ are an $n\times1$ vector and an $n\times p$ matrix, respectively, of random errors such that the $(p+1)$-vector $(\bve_i^T,\eta_i)$ is multivariate normal conditional on $\bZ$ with mean zero and covariance matrix $\bSigma=(\sigma_{jk})$. We write $\sigma_{jj}=\sigma_j^2$. Without loss of generality, we assume that each variable is centered about zero so that no intercept terms appear in \eqref{eq:ivmodel}, and that each column of $\bZ$ is standardized to have $L_2$ norm $\sqrt{n}$. We emphasize that $\bve_i$ and $\eta_i$ may be correlated because of the arbitrary covariance structure. In contrast to the usual linear model regressing $\by$ on $\bX$, model \eqref{eq:ivmodel} does not require that the covariate $\bX$ and the error $\bfeta$ be uncorrelated, thus substantially relaxing the assumptions of ordinary regression models and being more appealing in data analysis.

We are interested in making inference for the IV model \eqref{eq:ivmodel} in the high-dimensional setting where the dimensions $p$ and $q$ can both be much larger than the sample size $n$. In addition to selecting and estimating important covariate effects, since the identities of optimal instruments are unknown, we also regard the identification and estimation of optimal instruments as a variable selection and estimation problem. As is typical in high-dimensional sparse modeling, we assume that model \eqref{eq:ivmodel} is sparse in the sense that only a small subset of the regression coefficients in $\bbeta_0$ and $\bGamma_0$ are nonzero. Our goal is, therefore, to identify and estimate the nonzero coefficients in both $\bbeta_0$ and $\bGamma_0$.

\section{Regularization Methods and Implementation}
In this section, we first study the suboptimality of penalized least squares (PLS) estimators for the causal parameter $\bbeta_0$. We then propose the 2SR methodology and present an efficient coordinate descent algorithm for implementation. Finally, strategies for tuning parameter selection are discussed.

\subsection{Suboptimality of Penalized Least Squares}\label{sec:subopt}
In the classical setting where no regularization is needed, it is well known that the ordinary least squares estimator is inconsistent in the presence of endogeneity, that is, when some of the covariates are correlated with the error term. In high dimensions, without using the instruments, a direct application of one-stage regularization leads to the PLS estimator
\[
\what\bbeta^*=\argmin_{\bbeta\in\mathbb{R}^p}\left\{\frac{1}{2n}\|\by-\bX\bbeta\|_2^2+\sum_{j=1}^p p_{\mu}(|\beta_j|)\right\},
\]
where $\beta_j$ is the $j$th component of $\bbeta$ and $p_{\mu}(\cdot)$ is a penalty function that depends on a tuning parameter $\mu>0$. With appropriately chosen penalty functions, the PLS estimator has been shown to enjoy superior performance and theoretical properties; see, for example, \citet{Fan:Lv:sele:2010} for a review. When the data are generated from the linear IV model \eqref{eq:ivmodel}, however, the usual linear model that assumes the covariates to be uncorrelated with the error term is misspecified, and the PLS estimator $\what\bbeta^*$ is no longer a reasonable estimator of $\bbeta_0$. In fact, theoretical results in \citet{Lu:Gold:Fine:on:2012} and \citet{Lv:Liu:mode:2013} on misspecified generalized linear models imply that, under some regularity conditions, the PLS estimator $\what\bbeta^*$ is consistent for the least false parameter $\bbeta^*$ that minimizes the Kullback--Leibler divergence from the true model, which satisfies the equation
\begin{equation}\label{eq:bet_star}
\bX^T(\bX\bbeta_0+\bfeta_0-\bX\bbeta^*)=\bzero,
\end{equation}
where $\bfeta_0=E(\bfeta\mid\bX)$. The following proposition shows that there is a nonnegligible gap between $\bbeta^*$ and the true parameter $\bbeta_0$.

\begin{prop}[Gap between $\bbeta^*$ and $\bbeta_0$]\label{prop:gap}
If $\|n^{-1}\bX^T\bfeta_0\|_{\infty}\ne o_P(1)$ and $\max_{1\le j\le p}\|\bx_j\|_2=O_P(\sqrt{n})$, where $\bx_j$ is the $j$th column of $\bX$, then $\|\bbeta^*-\bbeta_0\|_1\ne o_P(1)$.
\end{prop}

It is interesting to compare the gap $\|\bbeta^*-\bbeta_0\|_1\ne o_P(1)$ with the minimax optimal rate $O(s\sqrt{(\log p)/n})$ for high-dimensional linear regression in $L_1$ loss over the $L_0$ ball $\cB_0(s)=\{\bbeta\in\mathbb{R}^p\colon\|\bbeta\|_0\le s\}$ \citep{Ye:Zhan:rate:2010,Rask:Wain:Yu:mini:2011}: the former dominates the latter if $s^2\log p=o(n)$. Thus, Proposition \ref{prop:gap} entails that, in the presence of endogeneity, any optimal procedure for estimating $\bbeta^*$, such as the PLS estimator with $L_1$ or other sparsity-inducing penalties, is suboptimal for estimating $\bbeta_0$ as long as $s^2\log p=o(n)$. Moreover, since by definition $\bX\bbeta^*$ is the orthogonal projection of $E(\by\mid\bX)$ onto the column space of $\bX$, the component $\bX\bbeta^*$ is generally closer to the expected response than $\bX\bbeta_0$. This will likely lead to a larger proportion of variance explained for the PLS method. Hence, to assess how well the fitted model predicts the response in model \eqref{eq:ivmodel}, it is more meaningful to compare the predicted values $\bX\what\bbeta$ to the causal component $\bX\bbeta_0$.

\subsection{Two-Stage Regularization}
One standard way of eliminating endogeneity is to replace the covariates by their expectations conditional on the instruments. This idea leads to the classical two-stage least squares (2SLS) method \citep{Ande:orig:2005}, in which the covariates are first regressed on the instruments and the response is then regressed on the first-stage predictions of the covariates. The performance of the 2SLS method deteriorates drastically or become inapplicable, however, as the dimensionality of covariates and instruments increase. We thus propose to apply regularization methods to cope with the high dimensionality in both stages of the 2SLS method, resulting in the following 2SR methodology.

\begin{asparaenum}[\itshape{Stage} 1.]
\item The goal of the first stage is to identify and estimate the nonzero effects of the instruments and obtain the predicted values of the covariates. Let $\|\cdot\|_F$ denote the Frobenius norm of a matrix. The first-stage regularized estimator is defined as
    \begin{equation}\label{eq:opt1}
    \what\bGamma=\argmin_{\bGamma\in\mathbb{R}^{q\times p}}\left\{\frac{1}{2n}\|\bX-\bZ\bGamma\|_F^2+\sum_{i=1}^q\sum_{j=1}^p p_{\lambda_j}(|\gamma_{ij}|)\right\},
    \end{equation}
    where $\gamma_{ij}$ is the $(i,j)$th entry of the matrix $\bGamma$, $p_{\lambda_j}(\cdot)$ is a sparsity-inducing penalty function to be discussed later, and $\lambda_j>0$ are tuning parameters that control the strength of the first-stage regularization. After the estimate $\what\bGamma$ is obtained, the predicted value of $\bX$ is formed by $\what\bX=\bZ\what\bGamma$.
\item Substituting the first-stage prediction $\what\bX$ for $\bX$, we proceed to identify and estimate the nonzero effects of the covariates. The second-stage regularized estimator is defined as
    \begin{equation}\label{eq:opt2}
    \what\bbeta=\argmin_{\bbeta\in\mathbb{R}^p}\left\{\frac{1}{2n}\|\by-\what\bX\bbeta\|_2^2+\sum_{j=1}^p p_{\mu}(|\beta_j|)\right\},
    \end{equation}
    where $\beta_j$ is the $j$th component of $\bbeta$, $p_{\mu}(\cdot)$ is a sparsity-inducing penalty function as before, and $\mu>0$ is a tuning parameter that controls the strength of the second-stage regularization. We thus obtain the pair $(\what\bbeta,\what\bGamma)$ as our final estimator for the regression parameter $(\bbeta_0,\bGamma_0)$ in model \eqref{eq:ivmodel}.
\end{asparaenum}

We consider the following three choices of the penalty function $p_{\lambda}(t)$ for $t\ge0$: (a) the $L_1$ penalty or Lasso \citep{Tibs:regr:1996}, $p_{\lambda}(t)=\lambda t$; (b) the smoothly clipped absolute deviation (SCAD) penalty \citep{Fan:Li:vari:2001},
\[
p_{\lambda}(t)=\lambda\int_0^t\left\{I(\theta\le\lambda)+\frac{(a\lambda-\theta)_+}{(a-1)\lambda}I(\theta>\lambda)\right\}\,d\theta,\qquad a>2;
\]
and (c) the minimax concave penalty (MCP) \citep{Zhan:near:2010},
\[
p_{\lambda}(t)=\int_0^t\frac{(a\lambda-\theta)_+}{a}\,d\theta,\qquad a>1.
\]
The SCAD and MCP penalties have an additional tuning parameter $a$ to control the shape of the function. These penalty functions have been widely used in high-dimensional sparse modeling and their properties are well understood in ordinary regression models \citep[e.g.,][]{Fan:Lv:sele:2010}. Moreover, the fact that these penalties belong to the class of quadratic spline functions on $[0,\infty)$ allows for a closed-form solution to the corresponding penalized least squares problem in each coordinate, leading to very efficient implementation via coordinate descent \citep[e.g.,][]{Mazu:Frie:Hast:spar:2011}.

\subsection{Implementation}
We now present an efficient coordinate descent algorithm for solving the optimization problems \eqref{eq:opt1} and \eqref{eq:opt2} with the Lasso, SCAD, and MCP penalties. We first note that the matrix optimization problem \eqref{eq:opt1} can be decomposed into $p$ penalized least squares problems,
\begin{equation}\label{eq:opt1'}
\what\bgamma_j=\argmin_{\bgamma_j\in\mathbb{R}^q}\left\{\frac{1}{2n}\|\bx_j-\bZ\bgamma_j\|_2^2+\sum_{i=1}^qp_{\lambda_j}(|\gamma_{ij}|)\right\},
\end{equation}
where $\bx_j$ is the $j$th column of the covariate matrix $\bX$ and $\bgamma_j=(\gamma_{1j},\dots,\gamma_{qj})^T$. The univariate solution to the unpenalized least squares problem \eqref{eq:opt1'} is given by $\tilde\gamma_{ij}=n^{-1}(\bx_j-\sum_{k\ne i}\gamma_{kj}\bz_k)^T\bz_i=n^{-1}\br_j^T\bz_i+\gamma_{ij}$, where $\bz_j$ is the $j$th column of the instrument matrix $\bZ$, $\br_j=\bx_j-\sum_{k=1}^q\gamma_{kj}\bz_k$ is the current residual, and we have used the fact $n^{-1}\bz_j^T\bz_j=1$ due to standardization. The penalized univariate solution, then, can be obtained by $\gamma_{ij}=S(\tilde\gamma_{ij};\lambda)$, where $S(\cdot;\lambda)$ is a thresholding operator defined for Lasso, SCAD, and MCP, respectively, as $S_{\mathrm{Lasso}}(t;\lambda)=\sgn(t)(|t|-\lambda)_+$,
\begin{align*}
S_{\mathrm{SCAD}}(t;\lambda)&=
\left\{\!\!\!\rule{0pt}{1.5\baselineskip}\right.\begin{array}{ll}
\sgn(t)(|t|-\lambda)_+,&\text{if }|t|\le2\lambda,\\
\displaystyle\sgn(t)\frac{|t|-\lambda a/(a-1)}{1-1/(a-1)},&\text{if }2\lambda<|t|\le a\lambda,\\
t,&\text{if }|t|>a\lambda,
\end{array}\\
\intertext{and}
S_{\mathrm{MCP}}(t;\lambda)&=
\left\{\!\!\!\rule{0pt}{1.1\baselineskip}\right.\begin{array}{ll}
\displaystyle\sgn(t)\frac{(|t|-\lambda)_+}{1-1/a},&\text{if }|t|\le a\lambda,\\
t,&\text{if }|t|>a\lambda.
\end{array}
\end{align*}
Similarly, if the $j$th column $\what\bx_j$ of the first-stage prediction matrix $\what\bX$ is standardized to have $L_2$ norm $\sqrt{n}$, the penalized univariate solution for the optimization problem \eqref{eq:opt2} is given by $\beta_j=S(\tilde\beta_j;\mu)$, where $\tilde\beta_j=n^{-1}\br^T\what\bx_j+\beta_j$ is the unpenalized univariate solution and $\br=\by-\sum_{k=1}^p\beta_k\what\bx_k$ is the current residual. We summarize the coordinate descent algorithm for computing the 2SR estimator $(\what\bbeta,\what\bGamma)$ in Algorithm \ref{alg:cd}.

\begin{algorithm}
\caption{Coordinate descent for the 2SR estimator}\label{alg:cd}
\begin{algorithmic}
\STATE Initialize: $\bbeta,\bGamma\gets\bzero$ or warm starts, $\lambda_1,\dots,\lambda_p,\mu>0$
\FOR{$j=1,\dots,p$}
\WHILE{$\bgamma_j$ not convergent}
\FOR{$i=1,\dots,q$}
\STATE $\gamma_{ij}\gets S(\tilde\gamma_{ij};\lambda_j)$
\ENDFOR
\ENDWHILE
\ENDFOR
\STATE $\what\bGamma\gets(\bgamma_1,\dots,\bgamma_p)$, $\what\bX\gets\bZ\what\bGamma$
\WHILE{$\bbeta$ not convergent}
\FOR{$j=1,\dots,p$}
\STATE $\beta_j\gets S(\tilde\beta_j,\mu)$
\ENDFOR
\ENDWHILE
\STATE $\what\bbeta\gets\bbeta$
\end{algorithmic}
\end{algorithm}

The convergence of Algorithm \ref{alg:cd} to a local minimum for $\what\bbeta$ and $\what\bGamma$ follows from the convergence properties of coordinate descent algorithms for penalized least squares; see, for example, \citet{Lin:Lv:high:2013}. Since the SCAD and MCP penalties are nonconvex, convergence to a global minimum is not guaranteed in general. In practice, coordinate descent algorithms are often used to produce a solution path over a grid of regularization parameter values, with warm starts from nearby solutions. In this case, the algorithm tends to find a sparse local solution with superior performance.

\subsection{Tuning parameter selection}
The 2SR method has $p+1$ regularization parameters $\lambda_1,\dots,\lambda_p$ and $\mu$ to be tuned. We propose to choose the optimal tuning parameters by $K$-fold cross-validation. Specifically, we define the cross-validation error for $\lambda_j$ and $\mu$ by
\begin{equation}\label{eq:cv_gam}
\CV(\lambda_j)=\frac{1}{K}\sum_{k=1}^K\|\bx_j^{(k)}-\bZ^{(k)}\what\bgamma_j^{(-k)}(\lambda_j)\|_2^2
\end{equation}
and
\begin{equation}\label{eq:cv_lam}
\CV(\mu)=\frac{1}{K}\sum_{k=1}^K\|\by^{(k)}-\what\bX^{(k)}\what\bbeta^{(-k)}(\mu)\|_2^2,
\end{equation}
respectively, where $\bx_j^{(k)}$, $\bZ^{(k)}$, $\by^{(k)}$, and $\what\bX^{(k)}$ are vectors/matrices for the $k$th part of the sample, and $\what\bGamma^{(-k)}(\lambda_j)$ and $\what\bbeta^{(-k)}(\mu)$ are the estimates obtained with the $k$th part removed. In view of the fact that in typical genetical genomics studies, both $p$ and $q$ can be in the thousands, it is necessary to reduce the search dimension of tuning parameters. To this end, we propose to first determine the optimal $\lambda_j$ that minimizes the criterion \eqref{eq:cv_gam}, for $j=1,\dots,p$, and then, with $\lambda_1,\dots,\lambda_p$ fixed, find the optimal $\mu$ that minimizes the criterion \eqref{eq:cv_lam}. The practical performance of this search strategy proves to be very satisfactory.

\section{Theoretical Properties}
In this section, we investigate the theoretical properties of the 2SR estimators. Through our theoretical analysis, we wish to understand (a) the impact of the dimensionality of covariates and instruments as well as other factors on the quality of the regularized estimators, and (b) the role of the two-stage regularization in providing performance guarantees for the regularized estimators, especially for the second-stage estimators. To address (a), we adopt a \emph{nonasymptotic} framework that allows the dimensionality of covariates and instruments to vary freely and thus can both be much larger than the sample size; to address (b), we impose conditions only on the instrument matrix $\bZ$, and treat the covariate matrix $\bX$ and the first-stage prediction $\what\bX$ as \emph{nondeterministic}. The major challenge arises in the characterization of the second-stage estimation, where the ``design matrix'' $\what\bX$ is neither fixed nor a random design sampled from a known distribution. Therefore, existing formulations for the high-dimensional analysis of ordinary regression models are inapplicable to our setting. We also stress that our theoretical analysis is essentially different from the recent developments in sparse IV models. The methods and results developed by \citet{Gaut:Tsyb:high:2011} and \citet{Fan:Liao:endo:2012} involve only one-stage estimation and regularization. The second-stage estimation considered by \citet{Bell:Chen:Cher:Hans:spar:2012} is of fixed dimensionality, which allows them to focus on estimation efficiency based on standard asymptotic analysis. Owing to the complications involved in the analysis of a general penalty, we first consider the representative case of $L_1$ regularization in Section \ref{sec:l1}, which allows us to obtain clean conditions providing important insights. We then present in Section \ref{sec:gen} a generalization of the theory, which is applicable to a much broader class of regularization methods.

\subsection{$L_1$ Regularization}\label{sec:l1}
We begin by introducing some notation. Let $\|\cdot\|_1$ and $\|\cdot\|_{\infty}$ denote the matrix 1-norm and $\infty$-norm, respectively, that is, $\|\bA\|_1=\max_j\sum_i|a_{ij}|$ and $\|\bA\|_{\infty}=\max_i\sum_j|a_{ij}|$ for any matrix $\bA=(a_{ij})$. For any vector $\balpha$, matrix $\bA$, and index sets $I$ and $J$, let $\balpha_J$ denote the subvector formed by the $j$th components of $\balpha$ with $j\in J$, and $\bA_{IJ}$ the submatrix formed with the $(i,j)$th entries of $\bA$ with $i\in I$ and $j\in J$. Also, denote by $J^c$ the complement of $J$ and $|J|$ the number of elements in $J$. Following \citet{Bick:Rito:Tsyb:simu:2009}, define the restricted eigenvalue for an $n\times m$ matrix $\bA$ and $1\le s\le m$ by
\[
\kappa(\bA,s)=\min_{|J|\le s}\min_{\substack{\bdelta\ne\bzero\\\|\bdelta_{J^c}\|_1\le3\|\bdelta_J\|_1}}\frac{\|\bA\bdelta\|_2}{\sqrt{n}\|\bdelta_J\|_2}.
\]
Let $\supp(\balpha)$ denote the support of a vector $\balpha=(\alpha_j)$, that is, $\supp(\balpha)=\{j\colon\alpha_j\ne0\}$. Define the sparsity levels $r=\max_{1\le j\le p}|{\supp(\bgamma_{0j})}|$ and $s=|{\supp(\bbeta_0)}|$, and the first-stage noise level $\sigma_{\max}=\max_{1\le j\le p}\sigma_j$, where $\bgamma_{0j}$ is the $j$th column of $\bGamma_0$. We consider the parameter space with $\|\bGamma_0\|_1\le L$ and $\|\bbeta_0\|_1\le M$ for some constants $L,M>0$.

To derive nonasymptotic bounds on the estimation and prediction loss of the regularized estimators $\what\bGamma$ and $\what\bbeta$, we impose the following conditions:\medskip

(C1) There exists $\kappa_1>0$ such that $\kappa(\bZ,r)\ge\kappa_1$.

(C2) There exists $\kappa_2>0$ such that $\kappa(\bZ\bGamma_0,s)\ge\kappa_2$.\medskip

We emphasize that dimensions $p$ and $q$, sparsity levels $r$ and $s$, and lower bounds $\kappa_1$ and $\kappa_2$ may all depend on the sample size $n$; we have suppressed the dependency for notational simplicity. Conditions (C1) and (C2) are analogous to those in \citet{Bick:Rito:Tsyb:simu:2009} for usual linear models, and require that the matrices $\bZ$ and $\bZ\bGamma_0$ be well behaved over some restricted sets of sparse vectors. One important difference, however, is that Condition (C2) is imposed on the conditional expectation matrix $\bZ\bGamma_0$ of $\bX$, rather than the first-stage prediction matrix, or the second-stage design matrix, $\what\bX$. This condition is more natural in our context, but poses new challenges for the analysis.

The estimation and prediction quality of the first-stage estimator $\what\bGamma$ is characterized by the following result.

\begin{thm}[Estimation and prediction loss of $\what\bGamma$]\label{thm:gam_est}
Under Condition (C1), if we choose
\begin{equation}\label{eq:lam}
\lambda_j=C\sigma_j\sqrt{\frac{\log p+\log q}{n}}
\end{equation}
with a constant $C\ge2\sqrt{2}$, then with probability at least $1-(pq)^{1-C^2/8}$, the regularized estimator $\what\bGamma$ defined by \eqref{eq:opt1} with the $L_1$ penalty satisfies
\[
\|\what\bGamma-\bGamma_0\|_1\le\frac{16C}{\kappa_1^2}\sigma_{\max}r\sqrt{\frac{\log p+\log q}{n}}
\]
and
\[
\|\bZ(\what\bGamma-\bGamma_0)\|_F^2\le\frac{16C^2}{\kappa_1^2}\sigma_{\max}^2pr(\log p+\log q).
\]
\end{thm}

Using the nonasymptotic bounds provided by Theorem \ref{thm:gam_est}, we can show that Condition (C2) also holds with high probability for the matrix $\what\bX=\bZ\what\bGamma$ with a smaller $\kappa_2$; see Lemma \ref{lem:kappa} in the Appendix. This allows us to establish the following result concerning the estimation and prediction quality of the second-stage estimator $\what\bbeta$.

\begin{thm}[Estimation and prediction loss of $\what\bbeta$]\label{thm:bet_est}
Under Conditions (C1) and (C2), if the regularization parameters $\lambda_j$ are chosen as in \eqref{eq:lam} and satisfy
\begin{equation}\label{eq:rate1}
\lambda_{\max}(2L+\lambda_{\max})\le\frac{\kappa_1^2\kappa_2^2}{32^2rs},
\end{equation}
where $\lambda_{\max}=\max_{1\le j\le p}\lambda_j$, then there exist constants $c_0,c_1,c_2>0$ such that, if we choose
\begin{equation}\label{eq:mu}
\mu=\frac{C_0}{\kappa_1}\sqrt{\frac{r(\log p+\log q)}{n}},
\end{equation}
where $C_0=c_0L\max(\sigma_{p+1},M\sigma_{\max})$, then with probability at least $1-c_1(pq)^{-c_2}$, the regularized estimator $\what\bbeta$ defined by \eqref{eq:opt2} with the $L_1$ penalty satisfies
\[
\|\what\bbeta-\bbeta_0\|_1\le\frac{64C_0}{\kappa_1\kappa_2^2}s\sqrt{\frac{r(\log p+\log q)}{n}}
\]
and
\[
\|\what\bX(\what\bbeta-\bbeta_0)\|_2^2\le\frac{64C_0^2}{\kappa_1^2\kappa_2^2}rs(\log p+\log q).
\]
\end{thm}

We now turn to the model selection consistency of $\what\bbeta$. Let $\bC=n^{-1}(\bZ\bGamma_0)^T\bZ\bGamma_0$, $S=\supp(\bbeta_0)$, and $\varphi=\|(\bC_{SS})^{-1}\|_{\infty}$. Define the minimum signal $b_0=\min_{j\in S}|\beta_{0j}|$, where $\beta_{0j}$ is the $j$th component of $\bbeta_0$. To study the model selection consistency, we replace Condition (C2) by the following condition:\medskip

(C3) There exists a constant $0<\alpha\le1$ such that $\|\bC_{S^cS}(\bC_{SS})^{-1}\|_{\infty}\le1-\alpha$.\medskip

Condition (C3) is in the same spirit as the irrepresentability condition in \citet{Zhao:Yu:on:2006} for the ordinary Lasso problem. Although Condition (C3) is placed on the covariance matrix of $\bZ\bGamma_0$, we can apply Theorem \ref{thm:gam_est} to show that it also holds with high probability for the covariance matrix of $\what\bX=\bZ\what\bGamma$ with a smaller $\alpha$; see Lemma \ref{lem:alpha} in the Appendix. The model selection consistency of $\what\bbeta$, along with a closely related $L_{\infty}$ bound, is established by the following result.

\begin{thm}[Model selection consistency of $\what\bbeta$]\label{thm:bet_sel}
Under Conditions (C1) and (C3), if the regularization parameter $\lambda_j$ are chosen as in \eqref{eq:lam} and satisfy
\begin{equation}\label{eq:rate2}
\frac{16\varphi}{\kappa_1^2}rs\lambda_{\max}(2L+\lambda_{\max})\le\frac{\alpha}{4-\alpha},
\end{equation}
then there exist constants $c_0,c_1,c_2>0$ such that, if the regularization parameter $\mu$ is chosen as in \eqref{eq:mu} and the minimal signal satisfies
\[
b_0>\frac{2}{2-\alpha}\varphi\mu,
\]
then with probability at least $1-c_1(pq)^{-c_2}$, there exists a regularized estimator $\what\bbeta$ defined by \eqref{eq:opt2} with the $L_1$ penalty that satisfies
\begin{compactenum}[\indent(a)]
\item (Sign consistency) $\sgn(\what\bbeta)=\sgn(\bbeta_0)$, and
\item ($L_\infty$ loss)
\[
\|\what\bbeta_S-\bbeta_{0S}\|_{\infty}\le\frac{2C_0\varphi}{(2-\alpha)\kappa_1}\sqrt{\frac{r(\log p+\log q)}{n}}.
\]
\end{compactenum}
\end{thm}

Theorem \ref{thm:bet_sel} shows that the second-stage estimator $\what\bbeta$ has the weak oracle property in the sense of \citet{Lv:Fan:unif:2009}. Two remarks are in order. First, the validity of our arguments for Theorems \ref{thm:bet_est} and \ref{thm:bet_sel} relies on the first-stage regularization only through the estimation and prediction bounds given in Theorem \ref{thm:gam_est}; this allows the arguments to be generalized to a generic class of regularization methods for the first stage, which will be explored in Section \ref{sec:gen}. Second, a key difference from the high-dimensional analysis of the usual linear model is that $\bX$ and $\bfeta$ may be correlated, and we have to make good use of the assumption that $\bE$ and $\bfeta$ are mean zero conditional on $\bZ$; see Lemma \ref{lem:ineq} in the Appendix.

Theorems \ref{thm:gam_est}--\ref{thm:bet_sel} deliver the important message that dimensions $p$ and $q$ contribute only a logarithmic factor to the estimation and prediction loss, and thus are both allowed to grow exponentially with the sample size $n$. Note that \eqref{eq:rate1} and \eqref{eq:rate2} are critical assumptions relating the first-stage regularization parameter $\lambda_{\max}$ to the key quantities in the second stage. To gain further insight into the dimension restrictions, suppose for simplicity that $\kappa_1$, $\kappa_2$, and $\varphi$ are constants; then \eqref{eq:rate1} and \eqref{eq:rate2} hold for sufficiently large $n$ provided that
\begin{equation}\label{eq:dim_rel}
r^2s^2(\log p+\log q)=o(n).
\end{equation}
This implies that dimensions $p$ and $q$ can grow at most as $e^{o(n)}$ and sparsity levels $r$ and $s$ can grow as $o(\sqrt{n})$, if the other quantities are fixed. Moreover, when $q$ and $r$ are also fixed, the relation \eqref{eq:dim_rel} reduces to $s^2\log p=o(n)$. In view of the remark following Proposition \ref{prop:gap} and the $L_1$ bound given in Theorem \ref{thm:bet_est}, we see that the 2SR estimator achieves the optimal rate for estimating $\bbeta_0$, which is asymptotically faster than that of the PLS estimator.

\subsection{General Regularization}\label{sec:gen}
We next present a theory for the second-stage estimator $\what\bbeta$ that generalizes the results in Section \ref{sec:l1} in two aspects. First, we allow the first-stage regularization to be arbitrary provided that certain nonasymptotic bounds are satisfied. Second, we allow the second-stage regularization to adopt a generic form of sparsity-inducing penalties, thus including the Lasso, SCAD, and MCP as special cases. Specifically, we impose the following conditions:\medskip

(C4) There exist $e_1$, $e_2$, and probability $\pi_0$, which may depend on $(n,p,q,r)$, such that the first-stage estimator $\what\bGamma$ satisfies $\|\what\bGamma-\bGamma_0\|_1\le e_1$ and $\max_{1\le j\le p}n^{-1}\|\bZ(\what\bgamma_j-\bgamma_{0j})\|_2^2\le e_2$ with probability $1-\pi_0$.

(C5) The penalty function $\rho_{\mu}(\cdot)\equiv p_{\mu}(\cdot)/\mu$ is increasing and concave on $[0,\infty)$, and has a continuous derivative $\rho_{\mu}'(\cdot)$ on $(0,\infty)$. In addition, $\rho_{\mu}'(\cdot)$ is increasing in $\mu$, and $\rho_{\mu}'(0+)\equiv\rho'(0+)\in(0,\infty)$ is independent of $\mu$.\medskip

Moreover, we replace Condition (C3) by the weaker assumption:\medskip

(C6) There exist constants $0<\alpha\le1$, $0\le\nu\le1/2$, and $c\ge1$ such that
\[
\|\bC_{S^cS}(\bC_{SS})^{-1}\|_{\infty}\le\left\{(1-\alpha)\frac{\rho'(0+)}{\rho_{\mu}'(b_0/2)}\right\}\wedge(cn^{\nu}).
\]\medskip

The family of penalty functions in Condition (C5) and a similar condition to (C6) were studied by, for example, \citet{Fan:Lv:nonc:2011} for generalized linear models; see the discussion therein for the motivation of these conditions. In particular, Condition (C5) captures several desirable properties of commonly used sparsity-inducing penalties, and allows us to establish a unified theory for these penalties. Condition (C6) is generally weaker than Condition (C3), since concavity implies that $\rho'(0+)\ge\rho_{\mu}'(b_0/2)$ and the right-hand side can be much larger than $1-\alpha$. Note that for the $L_1$ penalty, $\rho_{\mu}'(\cdot)\equiv1$ and this condition reduces to Condition (C3). For SCAD and MCP, when the signals are sufficiently strong such that $b_0/2\ge a\mu$, we have $\rho_{\mu}'(b_0/2)=0$ and the right-hand side can grow at most as $O(\sqrt{n})$.

Following \citet{Lv:Fan:unif:2009}, for any vector $\btheta=(\theta_j)$ with $\theta_j\ne0$ for all $j$, define the local concavity of the penalty function $\rho_{\mu}(\cdot)$ at point $\btheta$ by
\[
\tau(\rho_{\mu};\btheta)=\lim_{\ve\to0+}\max_j\sup_{|\theta_j|-\ve<t_1<t_2<|\theta_j|+\ve}\left\{-\frac{\rho_{\mu}'(t_2)-\rho_{\mu}'(t_1)}{t_2-t_1}\right\}.
\]
Further, define
\[
\tau_0=\sup\{\tau(\rho_{\mu};\btheta)\colon\btheta\in\mathbb{R}^s,\|\btheta-\bbeta_{0S}\|_{\infty}\le b_0/2\}
\]
and
\[
\mu_0=\Lambda_{\min}(\bC_{SS})-\mu\tau_0.
\]
The following result generalizes Theorem \ref{thm:bet_sel} and establishes the model selection consistency and weak oracle property of $\what\bbeta$.

\begin{thm}[Weak oracle property of $\what\bbeta$]\label{thm:gen_weak}
Under Conditions (C4)--(C6), if $\mu_0>0$ and the first-stage error bounds $e_1$ and $e_2$ satisfy
\begin{equation}\label{eq:rate_gen}
s(2Le_1+e_2)\le\frac{\alpha}{(4-\alpha)\varphi}\wedge\frac{(\mu_0/2)^2}{s},
\end{equation}
then there exist constants $c_0,c_1,c_2>0$ such that, if we choose
\[
\mu\ge C_0n^{\nu}\sqrt{\frac{\log p+\log q}{n}\vee e_2},
\]
where $C_0=c_0L\max(\sigma_{p+1},M\sigma_{\max},M)$, and the minimum signal satisfies
\begin{equation}\label{eq:thr}
b_0\ge\frac{7}{2}\varphi\mu\rho'(0+),
\end{equation}
then with probability at least $1-\pi_0-c_1(pq)^{-c_2}$, there exists a regularized estimator $\what\bbeta$ defined by \eqref{eq:opt2} that satisfies
\begin{compactenum}[\indent(a)]
\item (Sign consistency) $\sgn(\what\bbeta)=\sgn(\bbeta_0)$, and
\item ($L_\infty$ loss)
\[
\|\what\bbeta_S-\bbeta_{0S}\|_{\infty}\le\frac{7}{4}\varphi\mu\rho'(0+).
\]
\end{compactenum}
\end{thm}

Compared with Theorem \ref{thm:bet_sel}, Theorem \ref{thm:gen_weak} justifies the advantages of concave penalties such as SCAD and MCP in that model selection consistency and weak oracle property are established under substantially relaxed conditions. To understand the implications of the assumption \eqref{eq:rate_gen}, note that, for the $L_1$ penalty, Theorem \ref{thm:gam_est} gives $e_1=O(r\sqrt{(\log p+\log q)/n})$ and $e_2=O(r(\log p+\log q)/n)$, and the term involving $\mu_0$ is not needed. Assuming for simplicity that $\varphi$ is constant, \eqref{eq:rate_gen} reduces to the dimension restriction \eqref{eq:dim_rel}. Therefore, \eqref{eq:rate_gen} plays essentially the same role as the assumption \eqref{eq:rate2}, but applies to a generic first-stage estimator. Moreover, taking $e_2$ as above and $\nu=0$, we obtain the same rate of convergence for the $L_{\infty}$ loss as in Theorem \ref{thm:bet_sel}.

\section{Simulation Studies}
In this section, we report on simulation studies to evaluate the performance of the proposed 2SR method with the Lasso, SCAD, and MCP penalties. We compare the proposed method with the PLS estimators with the same penalties that do not utilize the instruments, as well as the PLS and 2SR oracle estimators that knew the relevant covariates and instruments in advance. We are particularly interested in investigating how the PLS and 2SR methods perform differently in relation to the sample size and how the dimensionality and instrument strength affect the performance of the 2SR method.

We first consider the case where the dimensions $p$ and $q$ are moderately high and smaller than the sample size $n$. Four models were examined, with $(n,p,q)=(200,100,100)$ in Model 1 and $(400,200,200)$ in Models 2--4. We first generated the coefficient matrix $\bGamma_0$ by sampling $r=5$ nonzero entries of each column from the uniform distribution $U([-b,-a]\cup[a,b])$. To represent different levels of instrument strength, we took $(a,b)=(0.75,1)$ for strong instruments in Models 1 and 2, and $(a,b)=(0.5,0.75)$ for weak instruments in Model 3. In Model 4, which reflects a more realistic setting, we sampled $r=50$ nonzero entries, consisting of 5 strong/weak instruments with $(a,b)=(0.5,1)$ and 45 very weak instruments with $(a,b)=(0.05,0.1)$. Similarly, we generated the coefficient vector $\bbeta_0$ by sampling $s=5$ nonzero components from $U([-1,-0.5]\cup[0.5,1])$. The covariance matrix $\bSigma=(\sigma_{ij})$ was specified as follows: We first set $\sigma_{ij}=(0.2)^{|i-j|}$ for $i,j=1,\dots,p$, and $\sigma_{p+1,\,p+1}=1$; in addition to the five $\sigma_{j,\,p+1}$'s corresponding to the nonzero components of $\bbeta_0$, we sampled another five entries from the last column of $\bSigma$; we then set these ten entries to 0.3 and let $\sigma_{p+1,\,j}=\sigma_{j,\,p+1}$ for $j=1,\dots,p$. Note that the nonzero $\sigma_{j,\,p+1}$'s were intended to cause both effect modifications and spurious associations for the PLS method. Finally, the instrument matrix $\bZ$ was generated by sampling each entry independently from $\mathrm{Bernoulli}(p_0)$, where $p_0=0.5$ in Models 1--3 and $p_0\sim U([0,0.5])$ in Model 4, and the covariate matrix $\bX$ and the response vector $\by$ were generated accordingly.

Since the PLS method provides no estimates for the coefficient matrix $\bGamma_0$ and our main interest is in how the estimation of $\bbeta_0$ can be improved by the 2SR method, we focus our comparisons on the second-stage estimation. Five measures on estimation, prediction, and model selection qualities were used to assess the performance of each method. The $L_1$ estimation loss $\|\what\bbeta-\bbeta_0\|_1$ and the prediction loss $n^{-1/2}\|\bX(\what\bbeta-\bbeta_0)\|_2$ quantify the estimation and prediction performance, respectively. The model selection performance is characterized by the number of true positives (TP), the model size, and the Matthews correlation coefficient (MCC). Here, positives refer to nonzero estimates. The MCC is a measure on the correlation between the observed and predicted binary classifications and is defined as
\[
\mathrm{MCC}=\frac{\TP\times\TN-\FP\times\FN}{\sqrt{(\TP+\FP)(\TP+\FN)(\TN+\FP)(\TN+\FN)}},
\]
where TN, FP, and FN denote the number of true negatives, false positives, and false negatives, respectively; a larger MCC indicates a better variable selection performance. In all simulations, we applied ten-fold cross-validation to choose the optimal tuning parameters and averaged each performance measure over 50 replicates.

The simulation results for Models 1--4 are summarized in Table \ref{tab:sim_low}. From the table we see that the 2SR method improved on the performance of the PLS method substantially in all cases. The improvement on model selection performance was most remarkable. The PLS method selected an exceedingly large model with many false positives because of its failure in distinguishing between the true and confounding effects, whereas the 2SR method resulted in a much sparser model and controlled the number of false positives at a reasonable level. As a result, the 2SR method had a much higher MCC than the PLS method, indicating a superior variable selection performance. The estimation and prediction performance of the PLS method was also greatly compromised by the confounding effects, and the 2SR method achieved a dramatic improvement on the $L_1$ estimation loss and a considerable improvement on the prediction loss. The comparisons between Model 2 and the weaker instrument settings, Models 3 and 4, suggest that a weaker instrument strength tends to decrease the performance of the 2SR method, as expected, especially on the estimation and prediction quality. We observe, however, that the model selection quality was only slightly affected and the overall performance of the 2SR method was still very satisfactory.

\begin{sidewaystable}
\def\~{\phantom{0}}
\caption{Simulation results for Models 1--4. Each performance measure was averaged over 50 replicates with standard deviation shown in parentheses}\label{tab:sim_low}\vskip1ex
\begin{tabular*}{\textwidth}{@{}l*{10}{@{\extracolsep{\fill}}c}@{}}
\toprule\toprule
& \multicolumn{5}{c}{PLS} & \multicolumn{5}{c@{}}{2SR}\\
\cmidrule{2-6}\cmidrule{7-11}
Method  & $L_1$ est.\ loss & Pred.\ loss & TP & Model size & MCC & $L_1$ est.\ loss & Pred.\ loss & TP & Model size & MCC\\
\midrule
\multicolumn{11}{@{}c@{}}{Model 1: $(n,p,q)=(200,100,100)$, $(a,b)=(0.75,1)$}\\
Lasso   & 2.49 (0.57) & 0.79 (0.16) & 5.0 (0.0) & 46.9 (11.2) & 0.25 (0.06) & 1.47 (0.69) & 0.78 (0.26) & 5.0 (0.2) & 14.5 (5.6) & 0.58 (0.12)\\
SCAD    & 2.12 (0.53) & 0.82 (0.17) & 5.0 (0.0) & 29.6 (6.0)\~& 0.36 (0.06) & 1.21 (0.55) & 0.74 (0.32) & 5.0 (0.2) & 12.9 (4.3) & 0.62 (0.12)\\
MCP     & 2.12 (0.59) & 0.82 (0.17) & 5.0 (0.0) & 24.3 (6.6)\~& 0.42 (0.08) & 1.26 (0.66) & 0.82 (0.34) & 4.9 (0.2) &\~9.5 (3.8) & 0.74 (0.16)\\
Oracle  & 0.75 (0.11) & 0.58 (0.12) &   5 (0)   &    5 (0)    &    1 (0)    & 0.51 (0.20) & 0.47 (0.22) &   5 (0)   &  \~5 (0)   &    1 (0)\\\addlinespace
\multicolumn{11}{@{}c@{}}{Model 2: $(n,p,q)=(400,200,200)$, $(a,b)=(0.75,1)$}\\
Lasso   & 2.71 (0.44) & 0.81 (0.14) & 5.0 (0.0) & 74.5 (14.3) & 0.21 (0.03) & 1.16 (0.52) & 0.57 (0.18) & 5.0 (0.0) & 18.1 (7.0) & 0.54 (0.11)\\
SCAD    & 2.32 (0.36) & 0.84 (0.14) & 5.0 (0.0) & 47.2 (10.2) & 0.29 (0.05) & 0.86 (0.43) & 0.49 (0.17) & 5.0 (0.0) & 14.0 (5.5) & 0.62 (0.13)\\
MCP     & 2.28 (0.45) & 0.84 (0.15) & 5.0 (0.0) & 36.2 (11.5) & 0.36 (0.08) & 0.76 (0.39) & 0.51 (0.19) & 5.0 (0.0) &\~9.3 (3.3) & 0.76 (0.14)\\
Oracle  & 0.76 (0.07) & 0.58 (0.11) &   5 (0)   &    5 (0)    &    1 (0)    & 0.41 (0.16) & 0.37 (0.15) &   5 (0)   &  \~5 (0)   &    1 (0)\\\addlinespace
\multicolumn{11}{@{}c@{}}{Model 3: $(n,p,q)=(400,200,200)$, $(a,b)=(0.5,0.75)$}\\
Lasso   & 3.04 (0.39) & 0.86 (0.11) & 5.0 (0.0) & 72.3 (12.5) & 0.22 (0.03) & 1.72 (0.73) & 0.73 (0.24) & 5.0 (0.0) & 18.3 (6.3) & 0.53 (0.10)\\
SCAD    & 2.64 (0.36) & 0.89 (0.11) & 5.0 (0.0) & 43.3 (11.9) & 0.32 (0.06) & 1.50 (0.65) & 0.69 (0.25) & 5.0 (0.1) & 16.8 (6.7) & 0.56 (0.12)\\
MCP     & 2.61 (0.42) & 0.89 (0.12) & 5.0 (0.0) & 33.5 (11.9) & 0.38 (0.09) & 1.36 (0.67) & 0.71 (0.26) & 5.0 (0.2) & 11.0 (4.4) & 0.69 (0.14)\\
Oracle  & 1.00 (0.08) & 0.63 (0.09) &   5 (0)   &    5 (0)    &    1 (0)    & 0.57 (0.23) & 0.43 (0.17) &   5 (0)   &  \~5 (0)   &    1 (0)\\\addlinespace
\multicolumn{11}{@{}c@{}}{Model 4: $(n,p,q)=(400,200,200)$, $(a,b)=(0.5,1)$ or $(0.05,0.1)$}\\
Lasso   & 2.88 (0.36) & 0.79 (0.08) & 5.0 (0.0) & 71.4 (13.4) & 0.22 (0.04) & 1.68 (0.72) & 0.72 (0.21) & 5.0 (0.0) & 18.8 (7.2) & 0.52 (0.10)\\
SCAD    & 2.49 (0.29) & 0.83 (0.08) & 5.0 (0.0) & 42.2 (10.2) & 0.32 (0.05) & 1.54 (0.85) & 0.68 (0.24) & 5.0 (0.0) & 17.4 (7.5) & 0.55 (0.11)\\
MCP     & 2.46 (0.40) & 0.82 (0.08) & 5.0 (0.0) & 32.6 (10.6) & 0.38 (0.08) & 1.52 (0.88) & 0.73 (0.25) & 5.0 (0.1) & 12.8 (5.7) & 0.65 (0.14)\\
Oracle  & 0.94 (0.07) & 0.58 (0.08) &   5 (0)   &    5 (0)    &    1 (0)    & 0.42 (0.17) & 0.30 (0.12) &   5 (0)   &    5 (0)   &    1 (0)\\
\bottomrule
\end{tabular*}
\end{sidewaystable}

To facilitate performance comparisons among different methods with varying sample size, Figure \ref{fig:sim_low} depicts the trends in three performance measures with the dimensions $p=q=100$ fixed and the sample size $n$ varying from 200 to 1500. It is clear from Figure \ref{fig:sim_low} that the performance of the 2SR method improves consistently as the sample size increases, whereas the PLS method does not in general see performance gain and may even deteriorate. Moreover, a closer look at the tails of the curves for the 2SR method with different penalties reveals certain advantages of SCAD and MCP over the Lasso. There seems to be a nonvanishing gap between the Lasso and oracle estimators, which agrees with the existing theory in the context of linear regression that the Lasso does not possess the oracle property \citep{Zou:adap:2006}.

\begin{figure}\centering
\includegraphics[width=.5\textwidth]{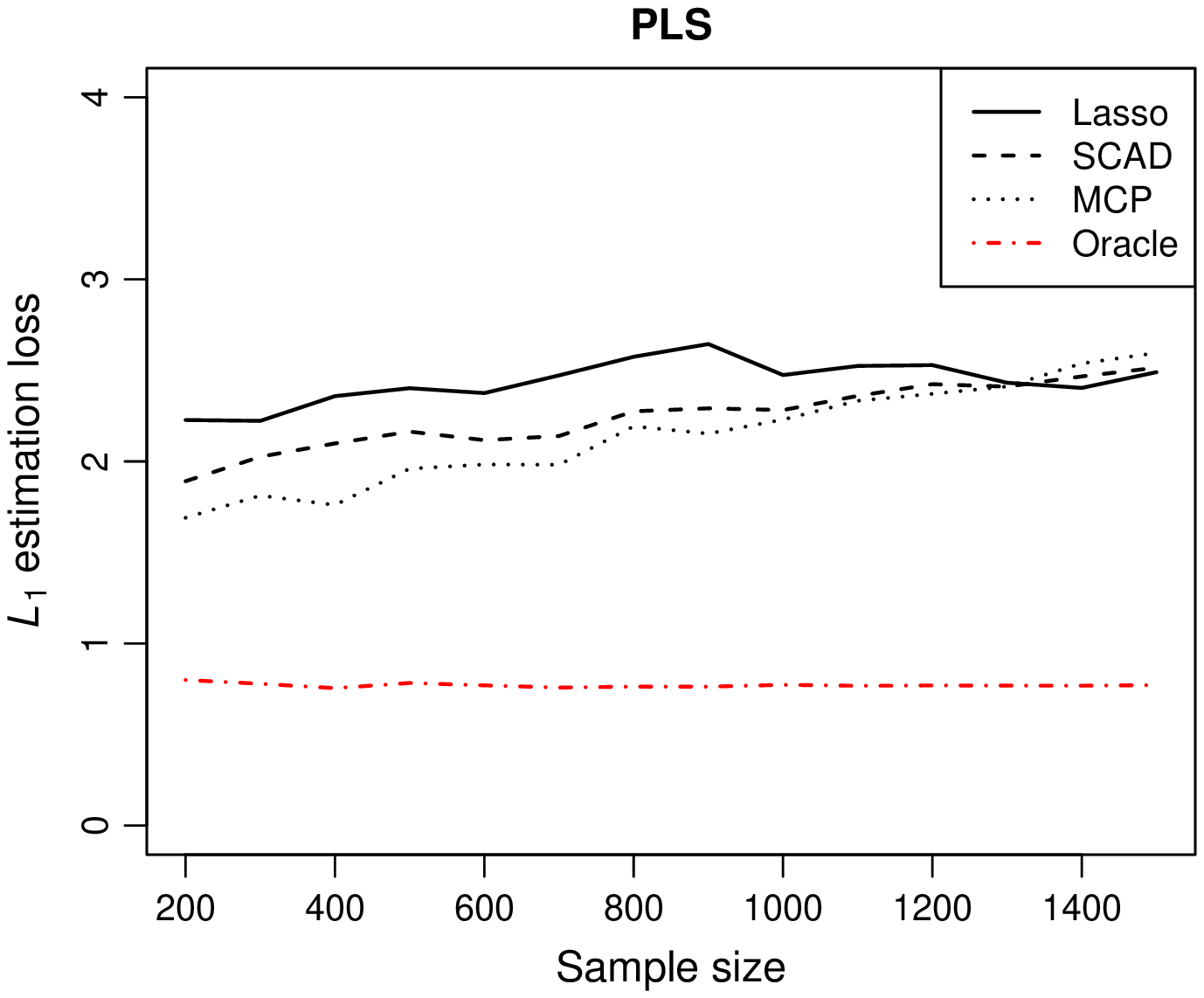}\includegraphics[width=.5\textwidth]{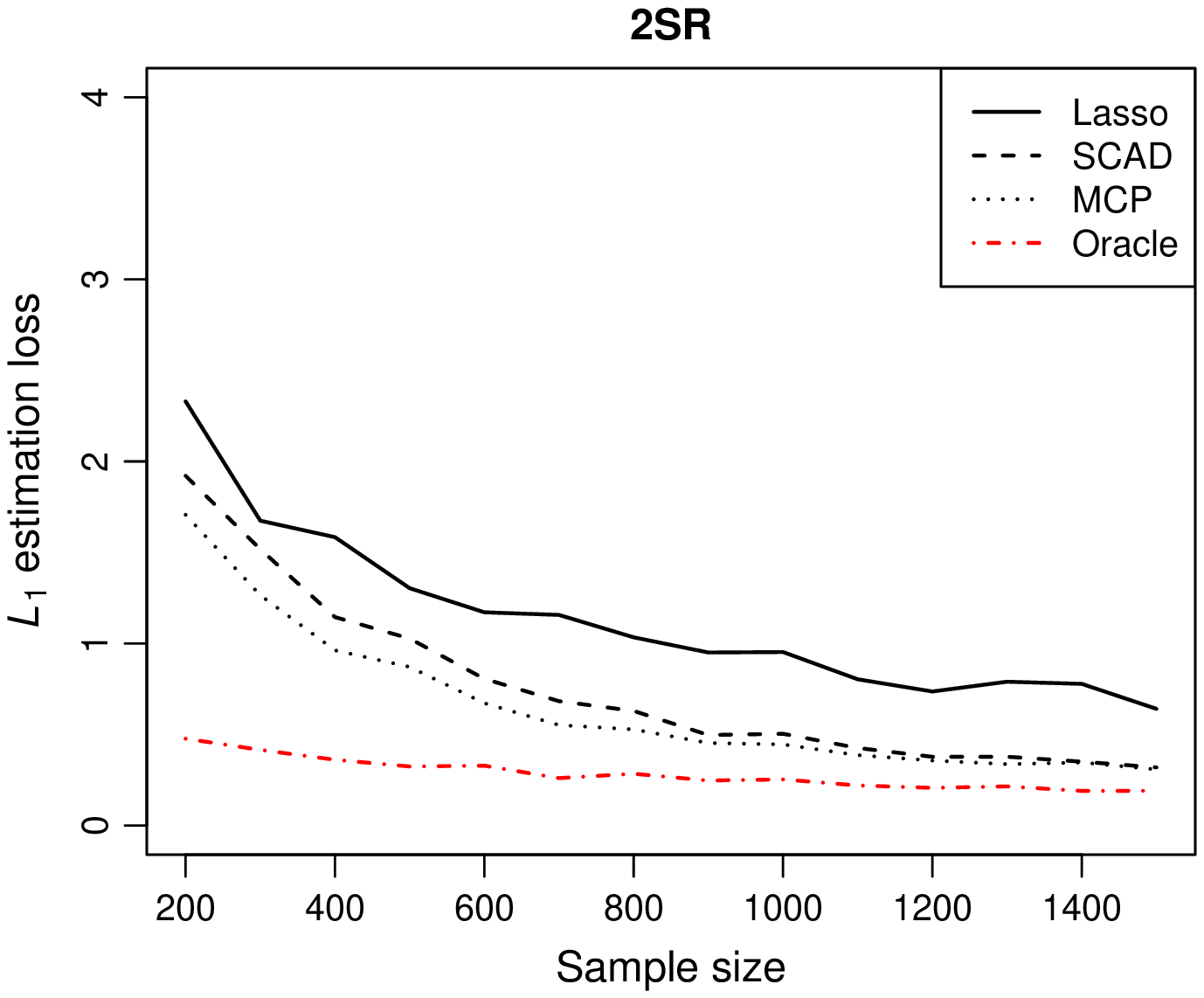}\\
\includegraphics[width=.5\textwidth]{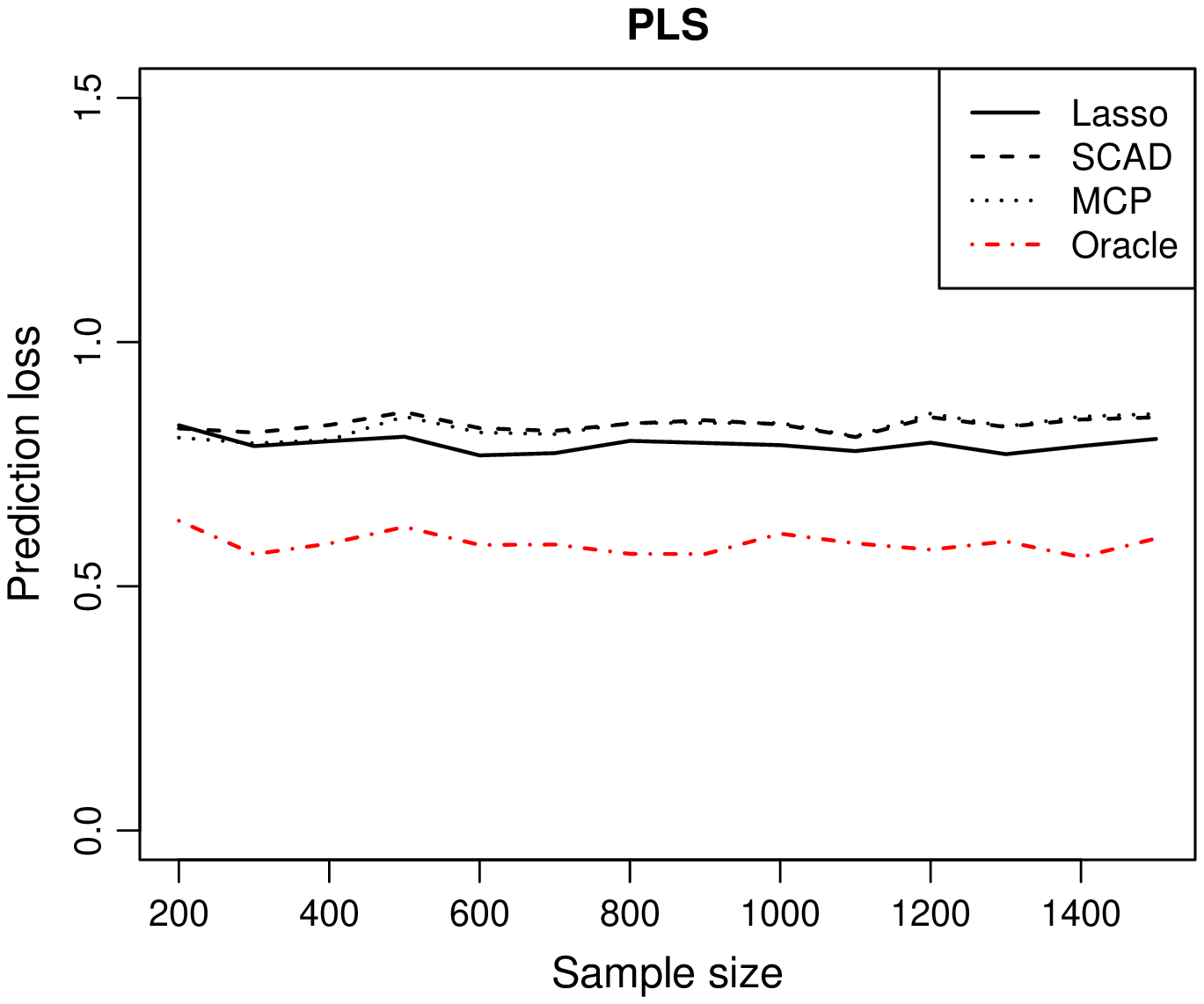}\includegraphics[width=.5\textwidth]{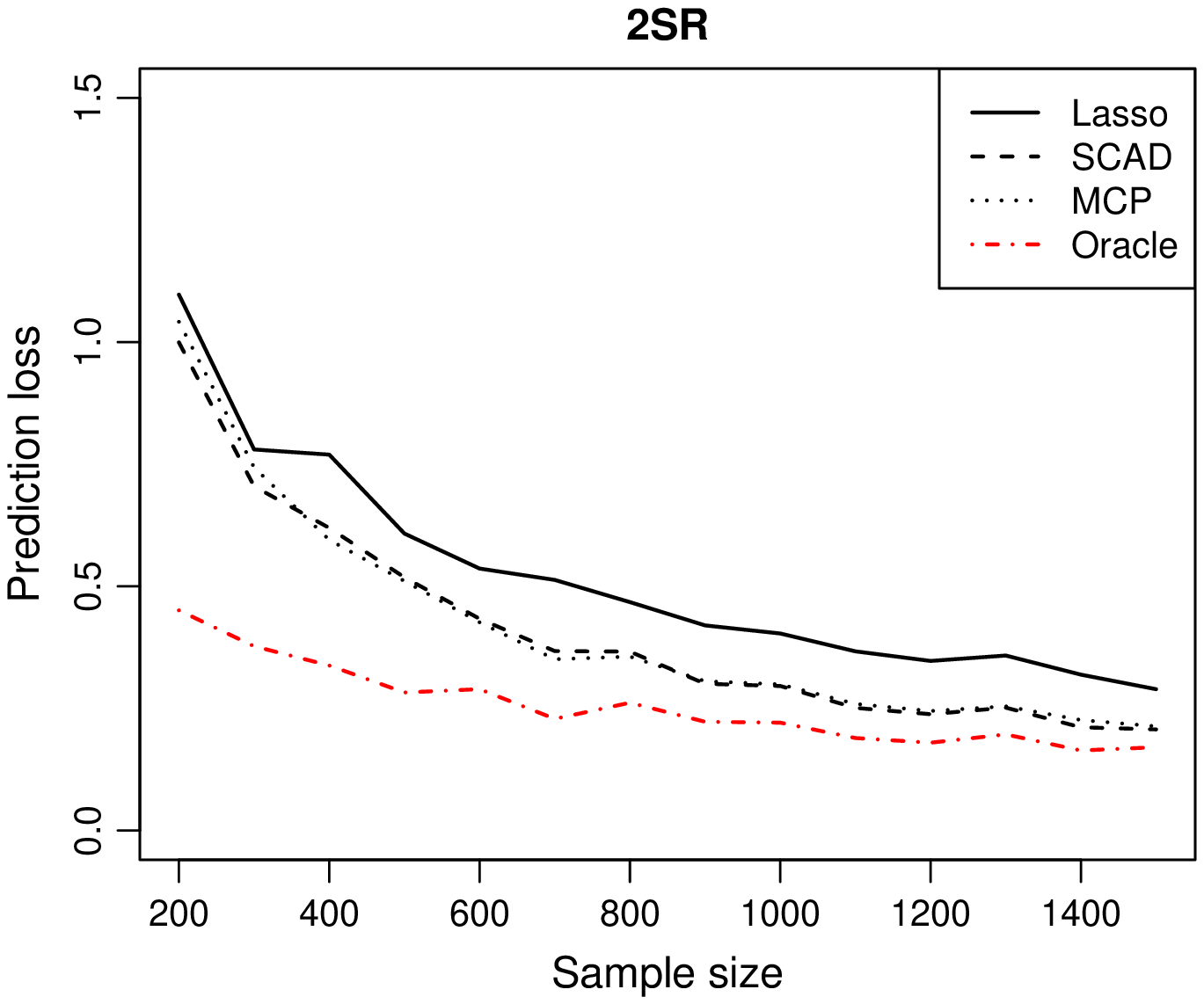}\\
\includegraphics[width=.5\textwidth]{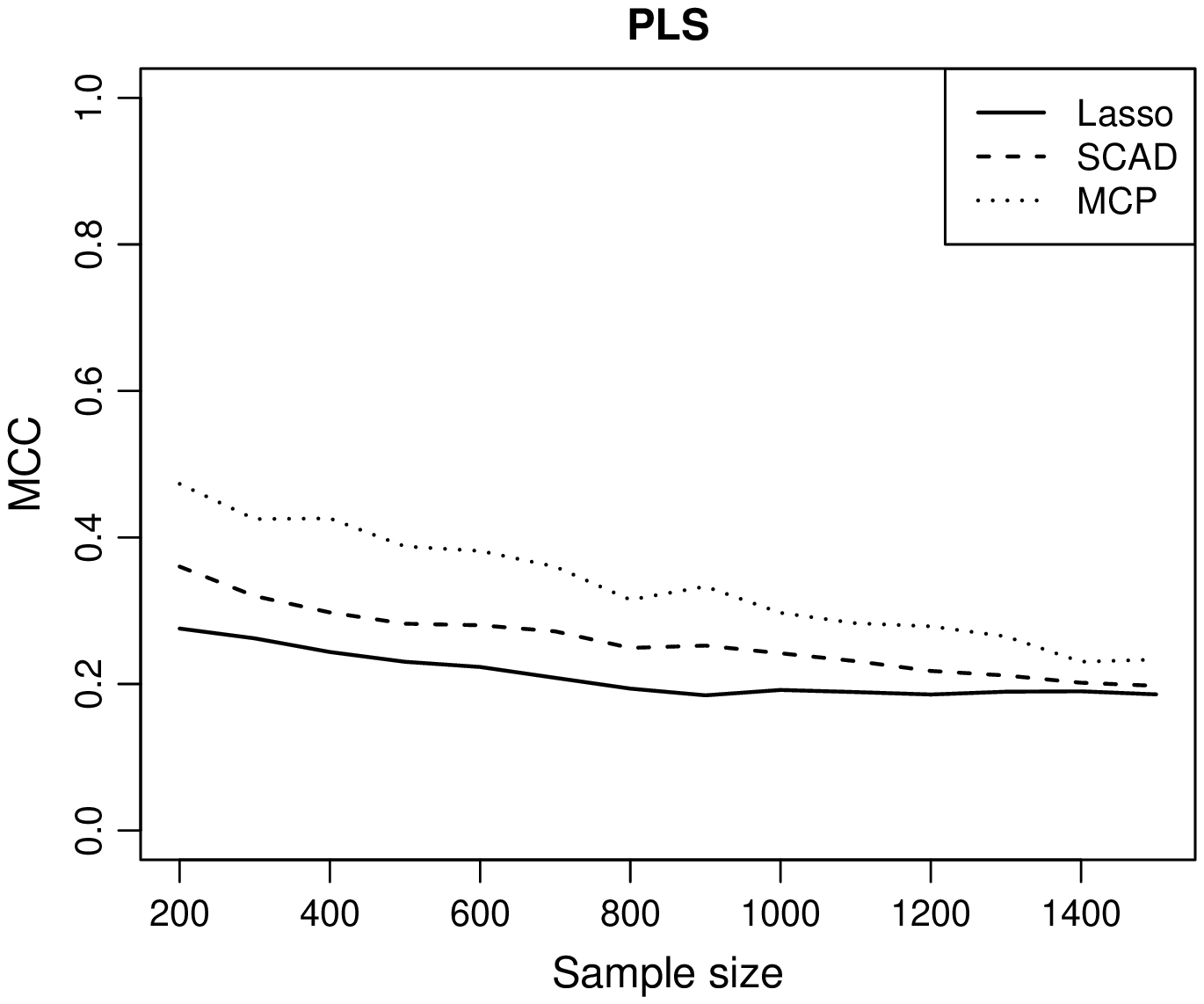}\includegraphics[width=.5\textwidth]{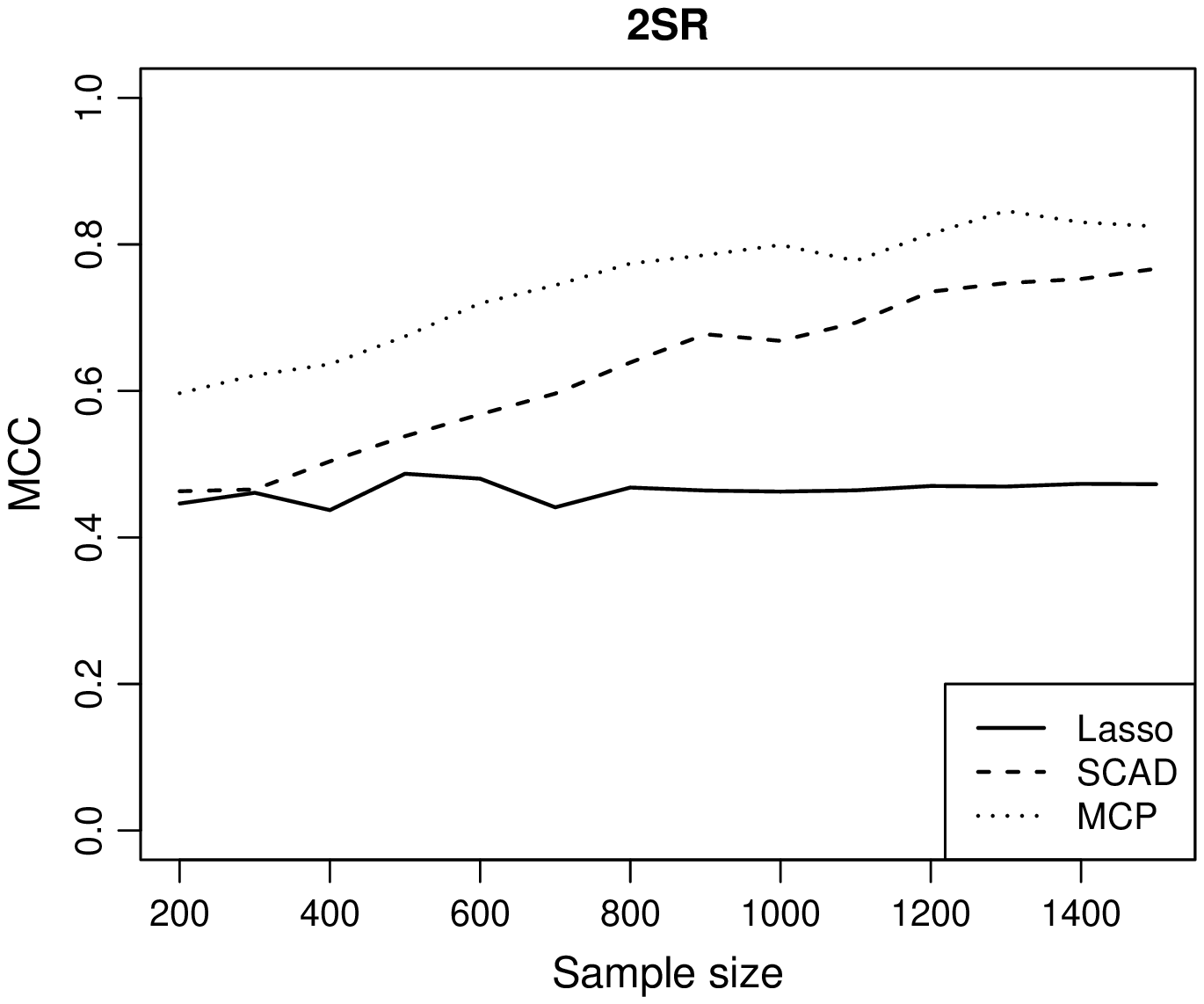}
\caption{Performance curves for different methods with the dimensions $p=q=100$ fixed and the sample size $n$ varying from 200 to 1500.}\label{fig:sim_low}
\end{figure}

We further study the case where the dimensions $p$ and $q$ are ultrahigh and larger than the sample size $n$. We considered four models with the same settings as in Models 1--4, except that $(n,p,q)=(300,600,600)$ in Model 5 and $(500,1000,1000)$ in Models 6--8. Table \ref{tab:sim_high} summarizes the simulation results for Models 5--8, and Figure \ref{fig:sim_high} shows the performance curves with $p=q=600$ fixed and $n$ varying from 200 to 1500. Trends in performance comparisons among different methods are similar to those in Table \ref{tab:sim_low} and Figure \ref{fig:sim_low}, demonstrating the advantages of the 2SR method over the PLS method. We observe that, although the ultrahigh dimensionality caused the 2SR method to select a larger model and resulted in a slightly lower MCC than in the previous settings, the performance of the 2SR method still compared favorably to the PLS method and the difference was pronounced for moderate sample sizes. These results suggest that the dimensionality has only mild impact on the performance of the 2SR method compared with the sample size, in agreement with our theoretical results in Section 4.

\begin{sidewaystable}
\def\~{\phantom{0}}
\caption{Simulation results for Models 5--8. Each performance measure was averaged over 50 replicates with standard deviation shown in parentheses}\label{tab:sim_high}\vskip1ex
\begin{tabular*}{\textwidth}{@{}l*{10}{@{\extracolsep{\fill}}c}@{}}
\toprule\toprule
& \multicolumn{5}{c}{PLS} & \multicolumn{5}{c@{}}{2SR}\\
\cmidrule{2-6}\cmidrule{7-11}
Method  & $L_1$ est.\ loss & Pred.\ loss & TP & Model size & MCC & $L_1$ est.\ loss & Pred.\ loss & TP & Model size & MCC\\
\midrule
\multicolumn{11}{@{}c@{}}{Model 5: $(n,p,q)=(300,600,600)$, $(a,b)=(0.75,1)$}\\
Lasso   & 2.22 (0.44) & 0.79 (0.15) & 5.0 (0.0) & 69.1 (20.0) & 0.26 (0.04) & 1.67 (0.81) & 0.78 (0.29) & 5.0 (0.0) & 25.5 (10.5) & 0.46 (0.10)\\
SCAD    & 2.03 (0.38) & 0.81 (0.17) & 5.0 (0.0) & 48.6 (14.1) & 0.32 (0.05) & 1.52 (0.74) & 0.70 (0.25) & 5.0 (0.0) & 26.3 (12.2) & 0.47 (0.12)\\
MCP     & 1.81 (0.34) & 0.79 (0.17) & 5.0 (0.0) & 28.7 (9.3)\~& 0.43 (0.07) & 1.26 (0.68) & 0.74 (0.30) & 5.0 (0.1) & 14.7 (7.1)\~& 0.62 (0.14)\\
Oracle  & 0.78 (0.08) & 0.57 (0.10) &   5 (0)   &    5 (0)    &    1 (0)    & 0.42 (0.15) & 0.38 (0.15) &   5 (0)   &    5 (0)    &    1 (0)\\\addlinespace
\multicolumn{11}{@{}c@{}}{Model 6: $(n,p,q)=(500,1000,1000)$, $(a,b)=(0.75,1)$}\\
Lasso   & 2.21 (0.46) & 0.80 (0.19) & 5.0 (0.0) & 87.1 (29.3) & 0.24 (0.04) & 1.28 (0.78) & 0.56 (0.23) & 5.0 (0.0) & 26.8 (13.3) & 0.46 (0.11)\\
SCAD    & 2.05 (0.36) & 0.83 (0.20) & 5.0 (0.0) & 61.2 (21.0) & 0.29 (0.06) & 0.93 (0.56) & 0.48 (0.23) & 5.0 (0.0) & 21.9 (12.3) & 0.54 (0.16)\\
MCP     & 1.82 (0.36) & 0.80 (0.20) & 5.0 (0.0) & 33.7 (15.0) & 0.41 (0.09) & 0.84 (0.63) & 0.49 (0.27) & 5.0 (0.0) & 14.1 (8.9)\~& 0.66 (0.16)\\
Oracle  & 0.76 (0.07) & 0.55 (0.10) &   5 (0)   &    5 (0)    &    1 (0)    & 0.29 (0.11) & 0.26 (0.11) &   5 (0)   &    5 (0)    &    1 (0)\\\addlinespace
\multicolumn{11}{@{}c@{}}{Model 7: $(n,p,q)=(500,1000,1000)$, $(a,b)=(0.5,0.75)$}\\
Lasso   & 2.65 (0.42) & 0.86 (0.16) & 5.0 (0.0) & 86.3 (27.1) & 0.24 (0.04) & 2.06 (1.32) & 0.77 (0.29) & 5.0 (0.0) & 27.6 (15.1) & 0.47 (0.14)\\
SCAD    & 2.39 (0.26) & 0.90 (0.17) & 5.0 (0.0) & 47.2 (19.0) & 0.34 (0.07) & 1.79 (0.80) & 0.70 (0.23) & 5.0 (0.0) & 28.9 (13.7) & 0.46 (0.15)\\
MCP     & 2.26 (0.27) & 0.89 (0.17) & 5.0 (0.0) & 28.3 (13.1) & 0.45 (0.11) & 1.52 (1.02) & 0.68 (0.31) & 5.0 (0.2) & 16.5 (10.5) & 0.61 (0.16)\\
Oracle  & 1.00 (0.07) & 0.61 (0.09) &   5 (0)   &    5 (0)    &    1 (0)    & 0.40 (0.15) & 0.30 (0.12) &   5 (0)   &    5 (0)    &    1 (0)\\\addlinespace
\multicolumn{11}{@{}c@{}}{Model 8: $(n,p,q)=(500,1000,1000)$, $(a,b)=(0.5,1)$ or $(0.05,0.1)$}\\
Lasso   & 2.68 (0.58) & 0.78 (0.09) & 5.0 (0.0) & 95.4 (37.8) & 0.23 (0.05) & 2.01 (0.80) & 0.73 (0.21) & 5.0 (0.0) & 29.5 (13.0) & 0.44 (0.09)\\
SCAD    & 2.38 (0.31) & 0.81 (0.08) & 5.0 (0.0) & 56.8 (20.2) & 0.30 (0.06) & 1.58 (0.63) & 0.65 (0.23) & 5.0 (0.1) & 25.7 (10.6) & 0.46 (0.10)\\
MCP     & 2.16 (0.31) & 0.80 (0.08) & 5.0 (0.0) & 30.3 (12.9) & 0.43 (0.10) & 1.45 (0.83) & 0.66 (0.24) & 5.0 (0.2) & 16.7 (9.1)\~& 0.59 (0.15)\\
Oracle  & 0.95 (0.07) & 0.58 (0.06) &   5 (0)   &    5 (0)    &    1 (0)    & 0.41 (0.14) & 0.29 (0.11) &   5 (0)   &    5 (0)    &    1 (0)\\
\bottomrule
\end{tabular*}
\end{sidewaystable}

\begin{figure}\centering
\includegraphics[width=.5\textwidth]{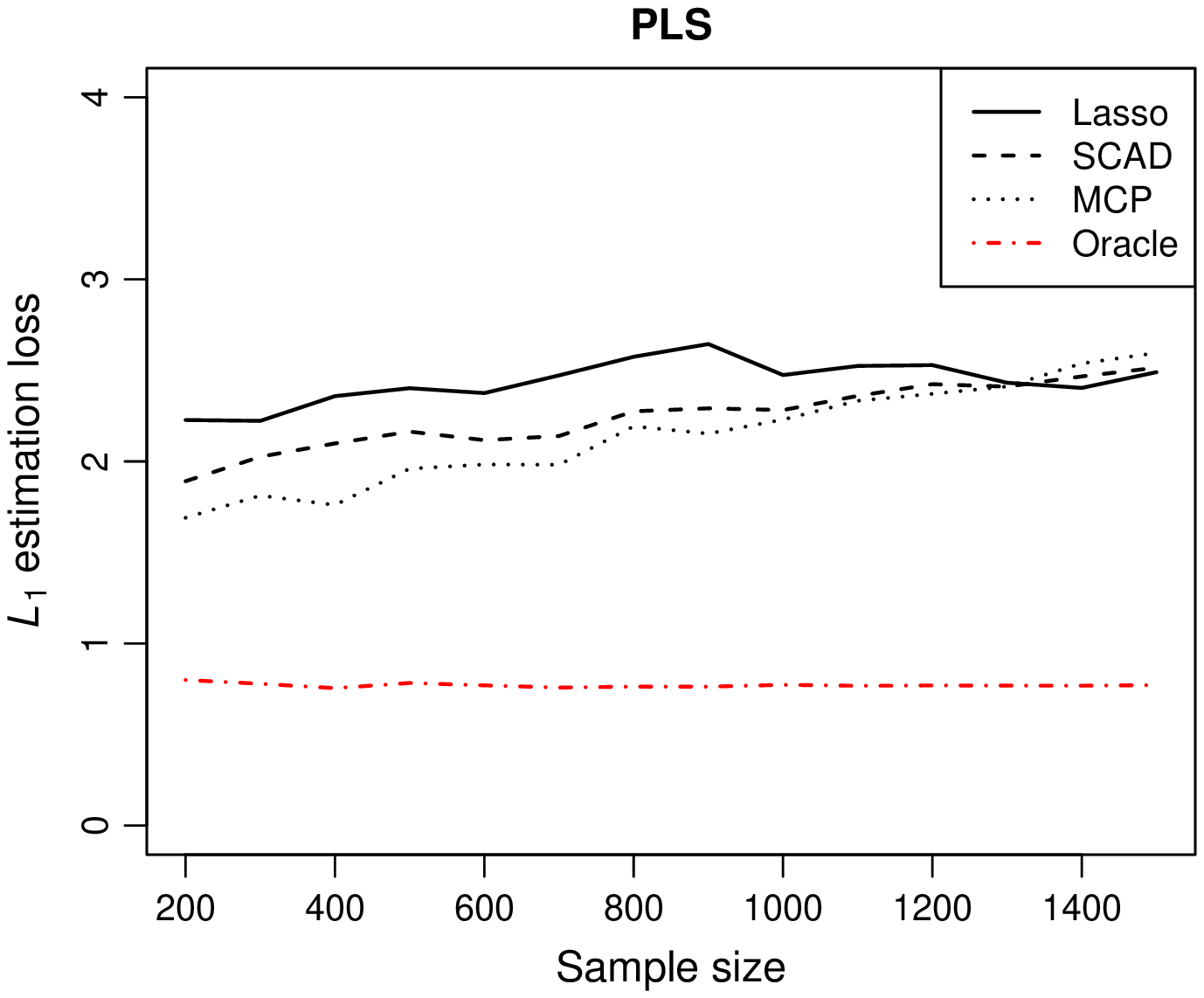}\includegraphics[width=.5\textwidth]{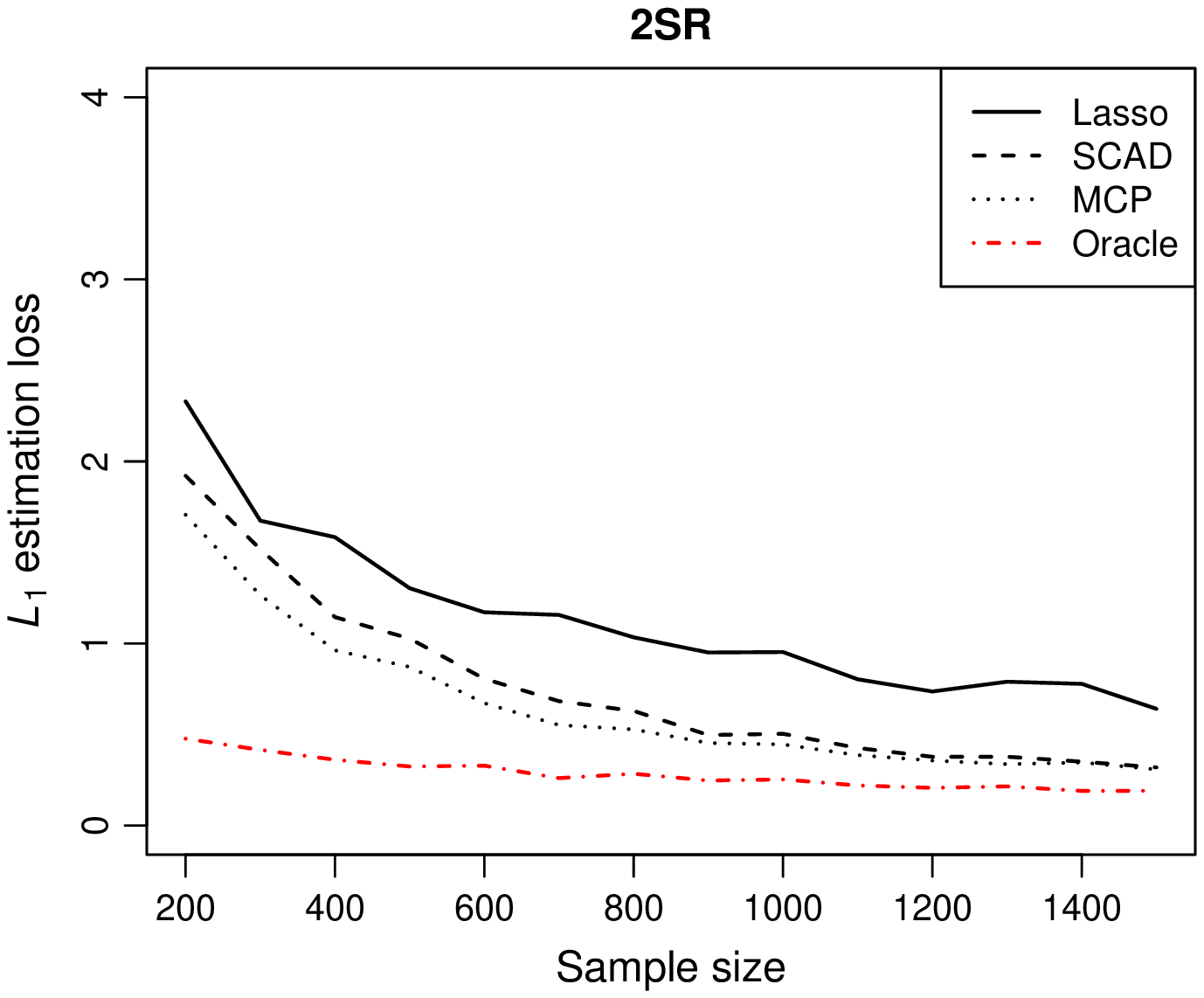}\\
\includegraphics[width=.5\textwidth]{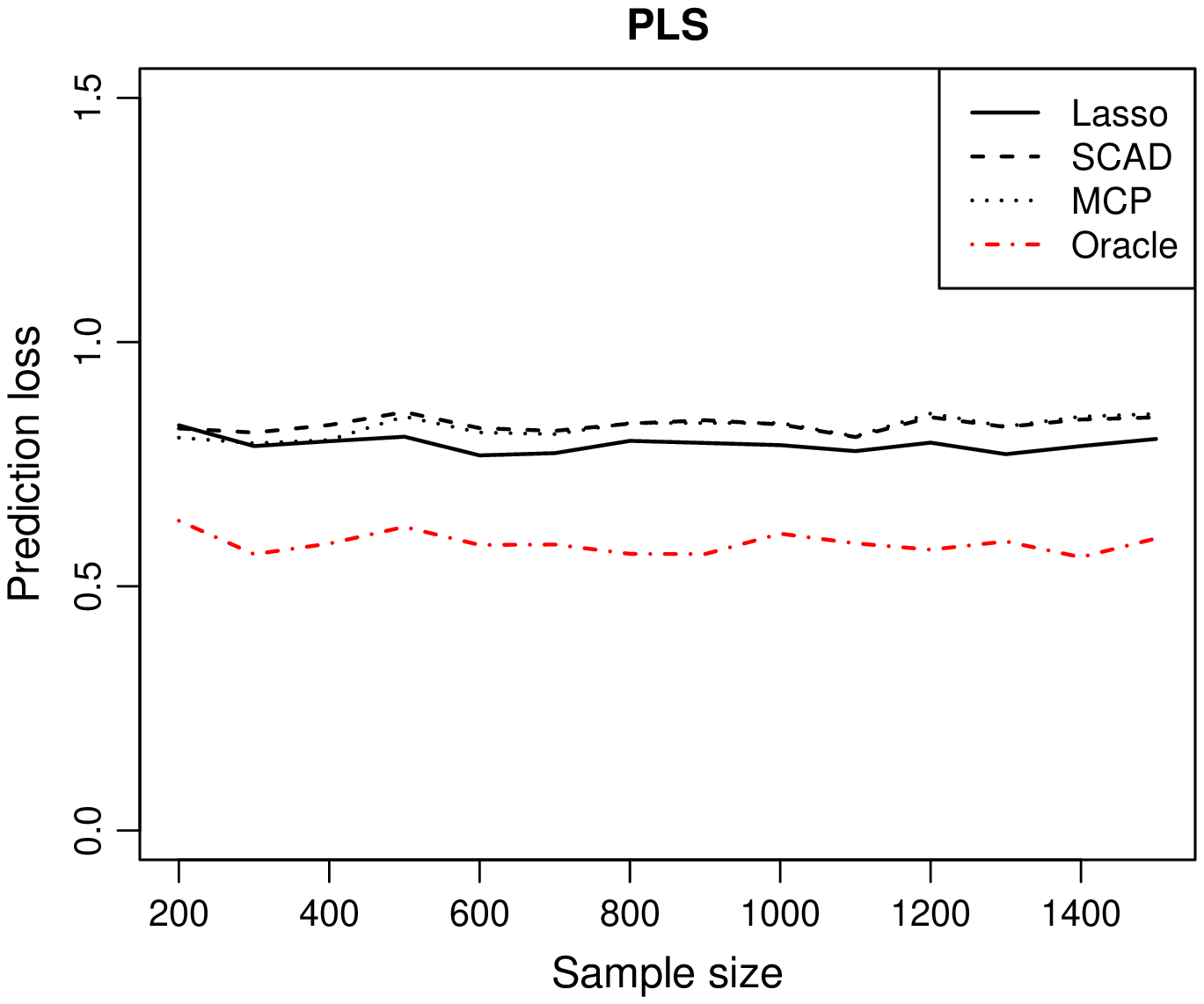}\includegraphics[width=.5\textwidth]{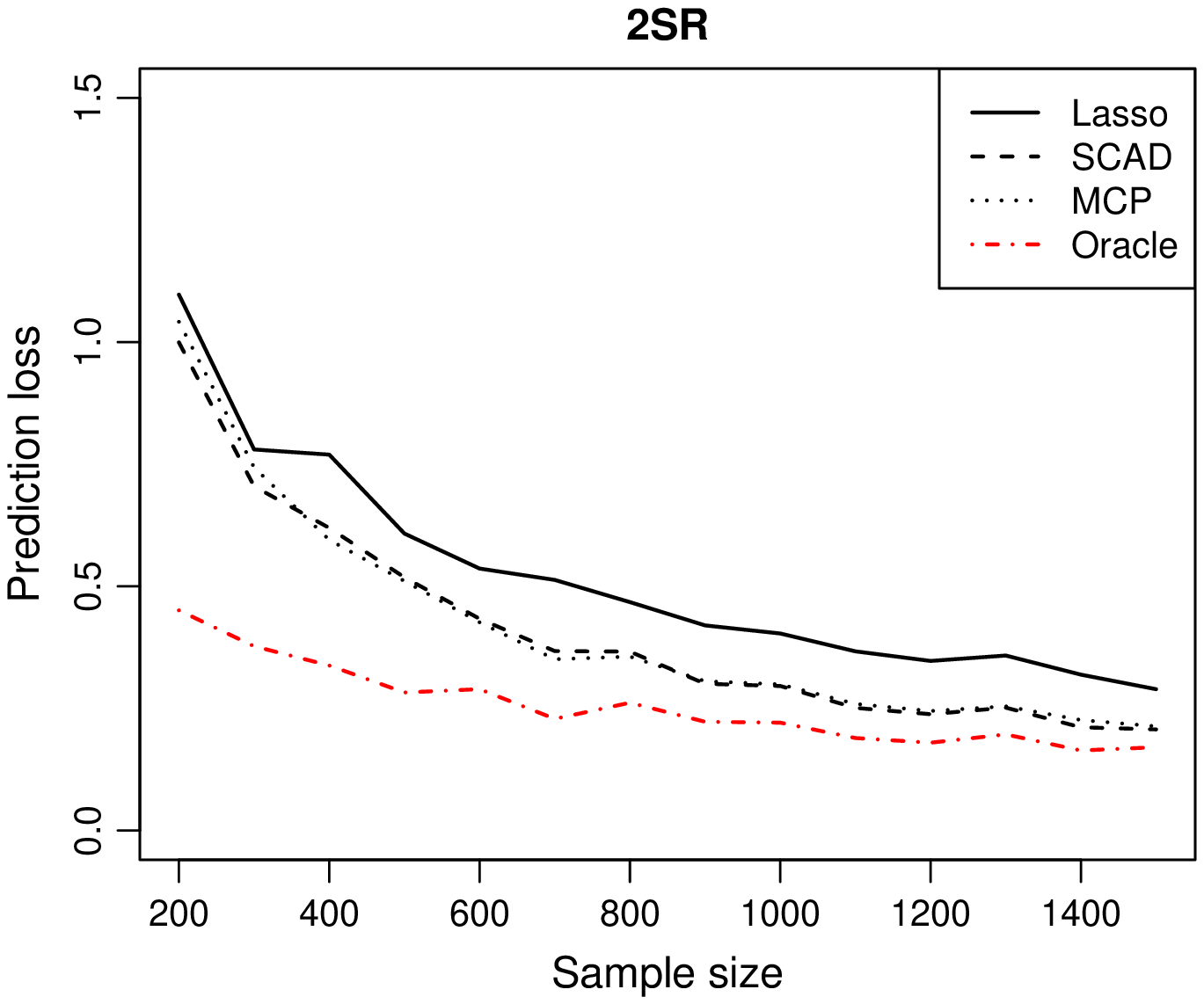}\\
\includegraphics[width=.5\textwidth]{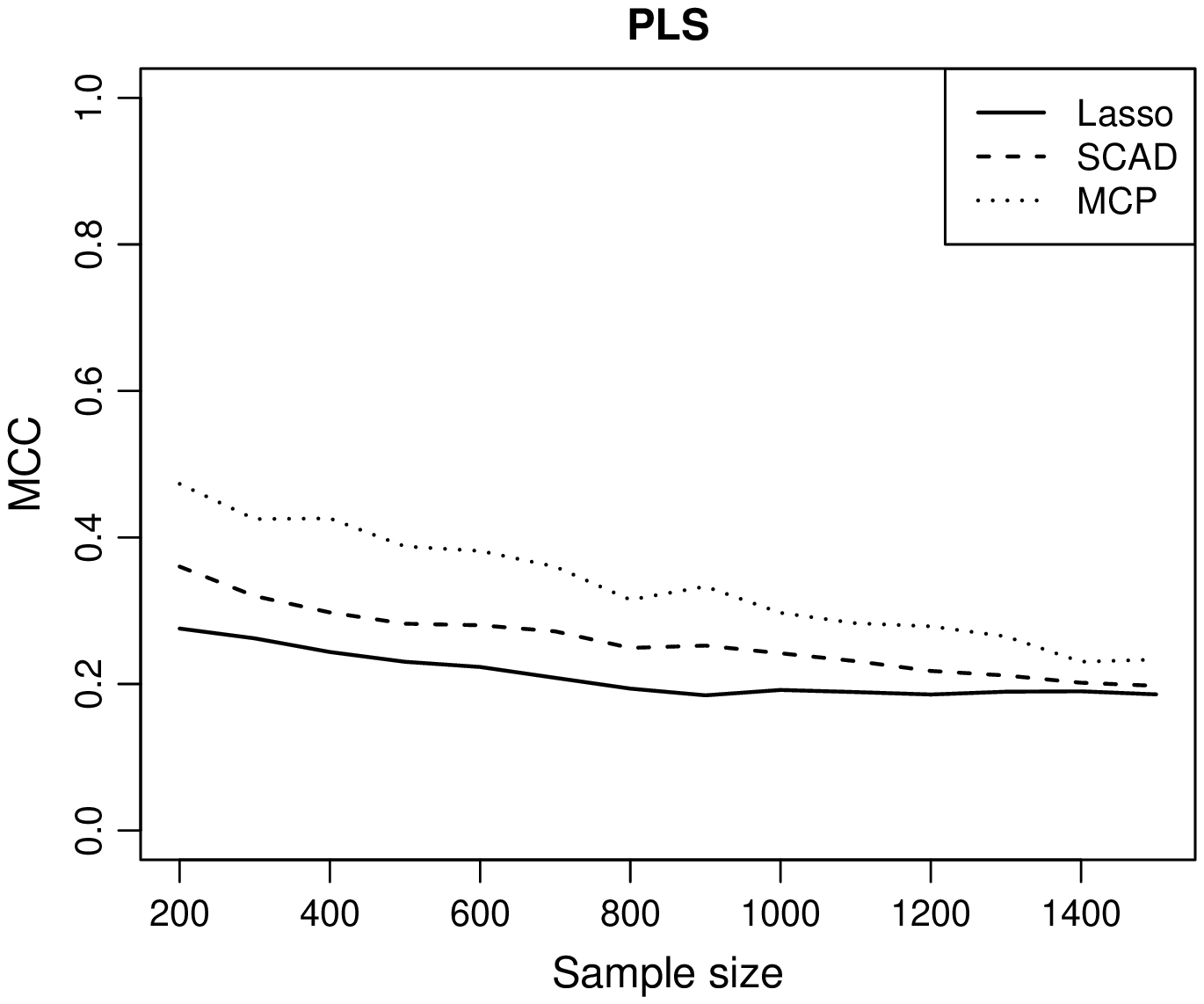}\includegraphics[width=.5\textwidth]{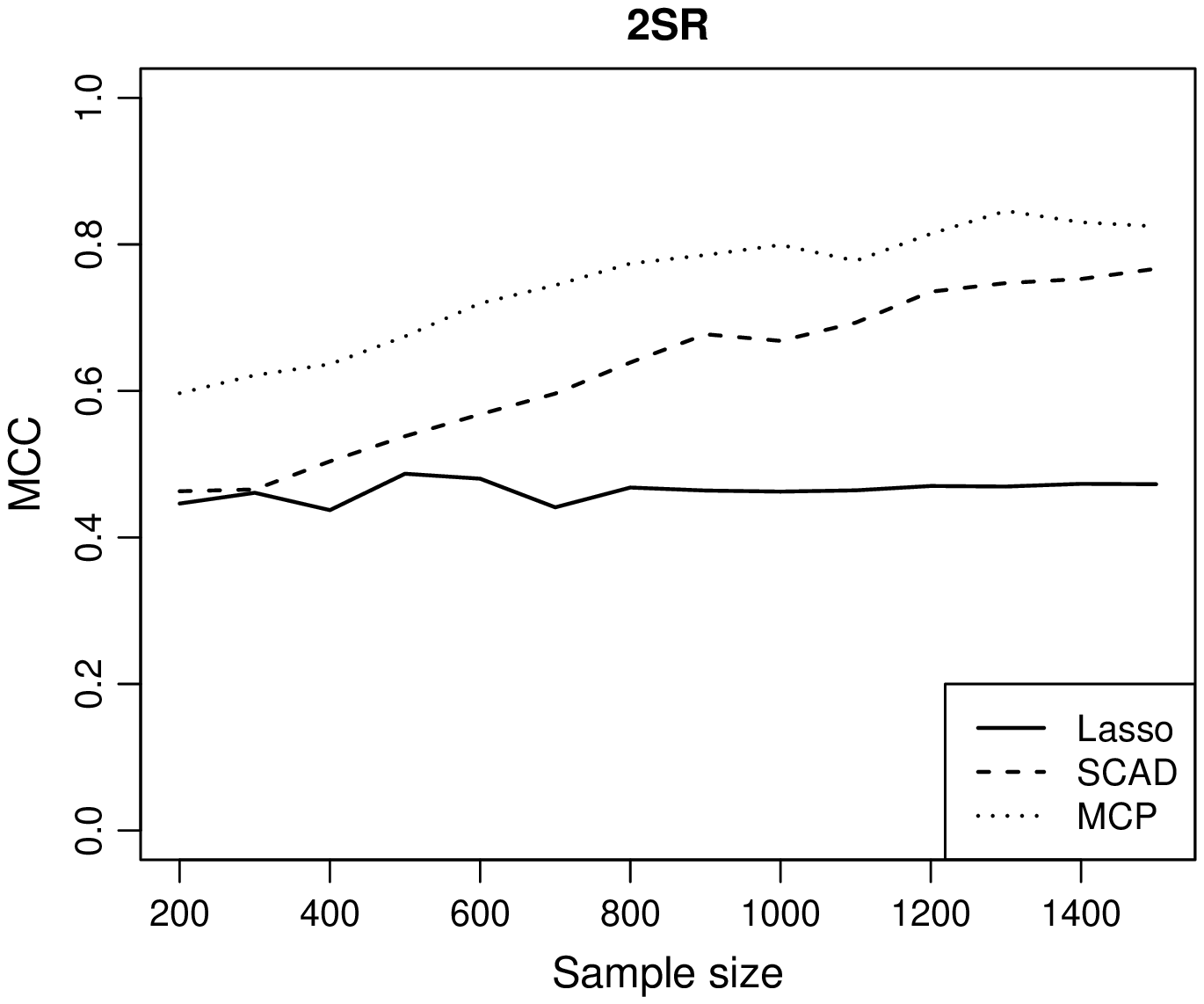}
\caption{Performance curves for different methods with the dimensions $p=q=600$ fixed and the sample size $n$ varying from 200 to 1500.}\label{fig:sim_high}
\end{figure}

\section{Analysis of Mouse Obesity Data}
To illustrate the application of the proposed method, in this section we present results from  the analysis of a mouse obesity data set described by \citet{Wang:Yehy:Scha:Wang:Drak:gene:2006}. The study includes an F2 intercross of 334 mice derived from the inbred strains C57BL/6J and C3H/HeJ on an apolipoprotein E (ApoE) null background, which were fed a high-fat Western diet from 8 to 24 weeks of age. The mice were genotyped using 1327 SNPs at an average density of 1.5 cM across the whole genome, and the gene expressions of the liver tissues of these mice were profiled on microarrays that include probes for 23,388 genes. Data on several obesity-related clinical traits were also collected on the animals. The genotype, gene expression, and clinical data are available for download, respectively, at \url{http://www.genetics.org/cgi/content/full/genetics.110.116087/DC1}, \url{ftp://ftp.ncbi.nlm.nih.gov/pub/geo/DATA/SeriesMatrix/GSE2814/}, and \url{http://labs.genetics.ucla.edu/horvath/CoexpressionNetwork/MouseWeight/}. Since the mice came from the same genetic cross, population stratification is unlikely an issue. Also, a study using a superset of these data demonstrated that most \emph{cis}-eQTLs were highly replicable across mouse crosses, tissues, and sexes \citep{Nas:Ingr:Sins:Wang:Scha:expr:2010}. Therefore, the three assumptions for valid IVs seem to be plausible.

After the individuals, SNPs, and genes with a missing rate greater than 0.1 were removed, the remaining missing genotype and gene expression data were imputed using the Beagle approach \citep{Brow:Brow:rapi:2007} and nearest neighbor averaging \citep{Troy:Cant:Sher:Brow:Hast:miss:2001}, respectively. Merging the genotype, gene expression, and clinical data yielded a complete data set with $q=1250$ SNPs and 23,184 genes on $n=287$ (144 female and 143 male) mice. To enhance the interpretability and stability of the results, we focus on the $p=2825$ genes that can be mapped to the Mouse Genome Database (MGD) \citep{Eppi:Blak:Bult:Kadi:Rich:mous:2012} and have standard deviation of gene expression levels greater than 0.1. The latter criterion is reasonable because gene expressions of too small variation are typically not of biological interest and suggest that the genetic perturbations may not be sufficiently strong for the genetic variants to be used as instruments.

Our goal is to jointly analyze the genotype, gene expression, and clinical data to identify important genes related to body weight. In order to utilize data from both sexes, we first adjusted the body weight for sex by fitting a marginal linear regression model with sex as the covariate and subtracting the estimated sex effect from the body weight. We then applied the proposed 2SR method with the Lasso, SCAD, and MCP penalties to the data set with adjusted body weight as the response. For comparison, we also applied the PLS method to the same data set, and used ten-fold cross-validation to choose the optimal tuning parameters for both methods. The models selected by cross-validation include 110 (Lasso), 49 (SCAD), and 16 (MCP) genes for the PLS method, and include 37 (Lasso), 15 (SCAD), and 9 (MCP) genes for the 2SR method. The selected models resulted in an adjusted $R^2$ of 0.894 (Lasso), 0.833 (SCAD), and 0.820 (MCP) for the PLS method, and 0.594 (Lasso), 0.581 (SCAD), and 0.579 (MCP) for the 2SR method, which is consistent with our remark following Proposition \ref{prop:gap}. Since we have no knowledge of the causal component $\bX\bbeta_0$ for the real data, a direct comparison between the PLS and 2SR methods in assessing the model fit is not possible. Nevertheless, we observe that the 2SR method produced a much sparser model with reasonably high proportion of variance explained.

To gain insight into the stability of the selection results, we followed the idea of stability selection \citep{Mein:Buhl:stab:2010} to compute the selection probability of each gene over 100 subsamples of size $\lfloor n/2\rfloor$ for each fixed value of the regularization parameter $\mu$. The resulting stability paths for different methods are displayed in Figure \ref{fig:path}. It is interesting to observe that, among the genes with maximum selection probability at least 0.4, only 5 (Lasso), 3 (SCAD), and 0 (MCP) genes are common to both the PLS and 2SR methods. As can be seen from Figure \ref{fig:path}, these few genes, which are reasonably conjectured to be among the truly relevant ones, stand out more clearly in the stability paths of the 2SR method. Moreover, the overall stability paths of the 2SR method seem less noisy and hence can be more useful for distinguishing the most important genes from the irrelevant ones.

\begin{figure}\centering
\includegraphics[width=.5\textwidth]{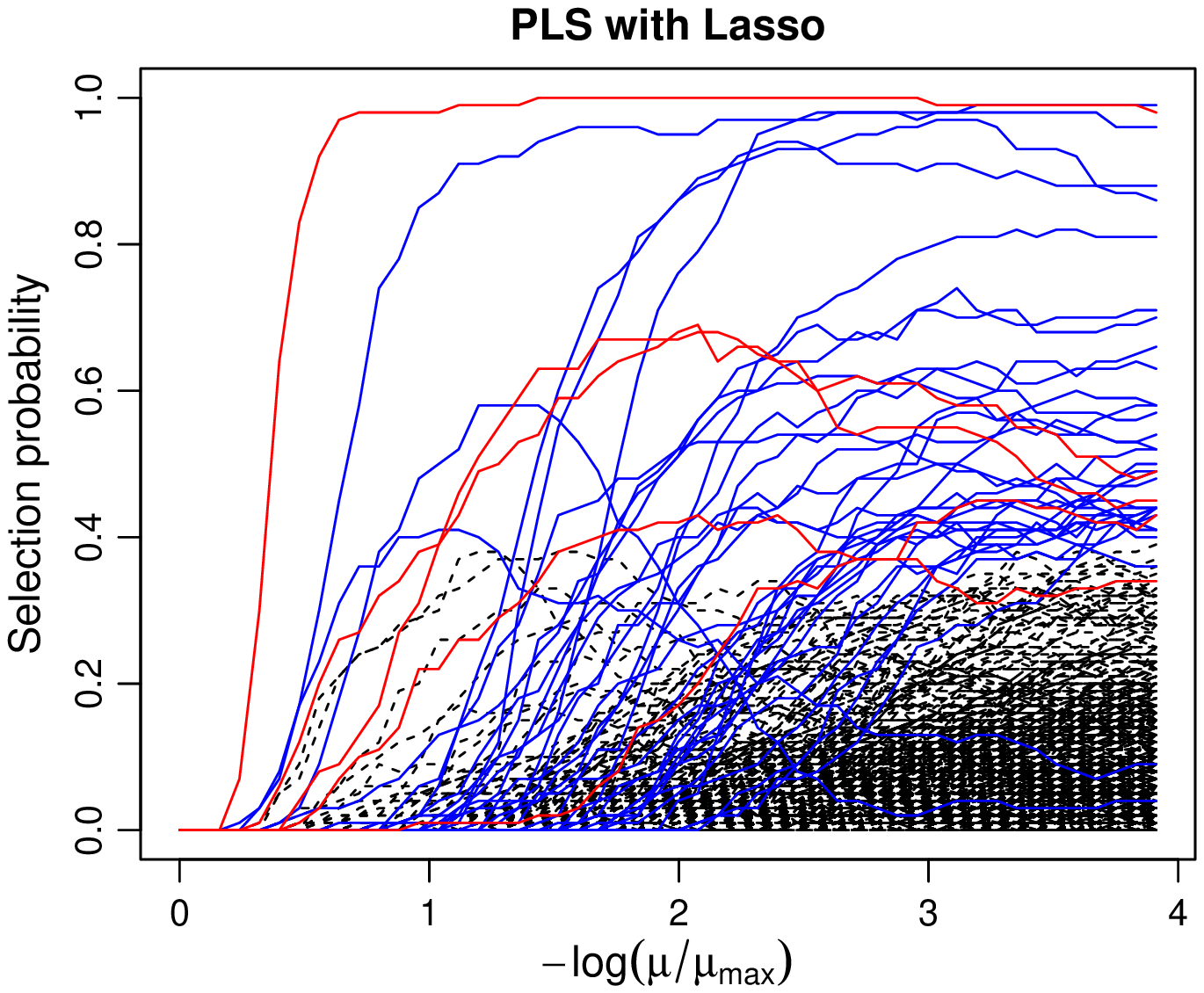}\includegraphics[width=.5\textwidth]{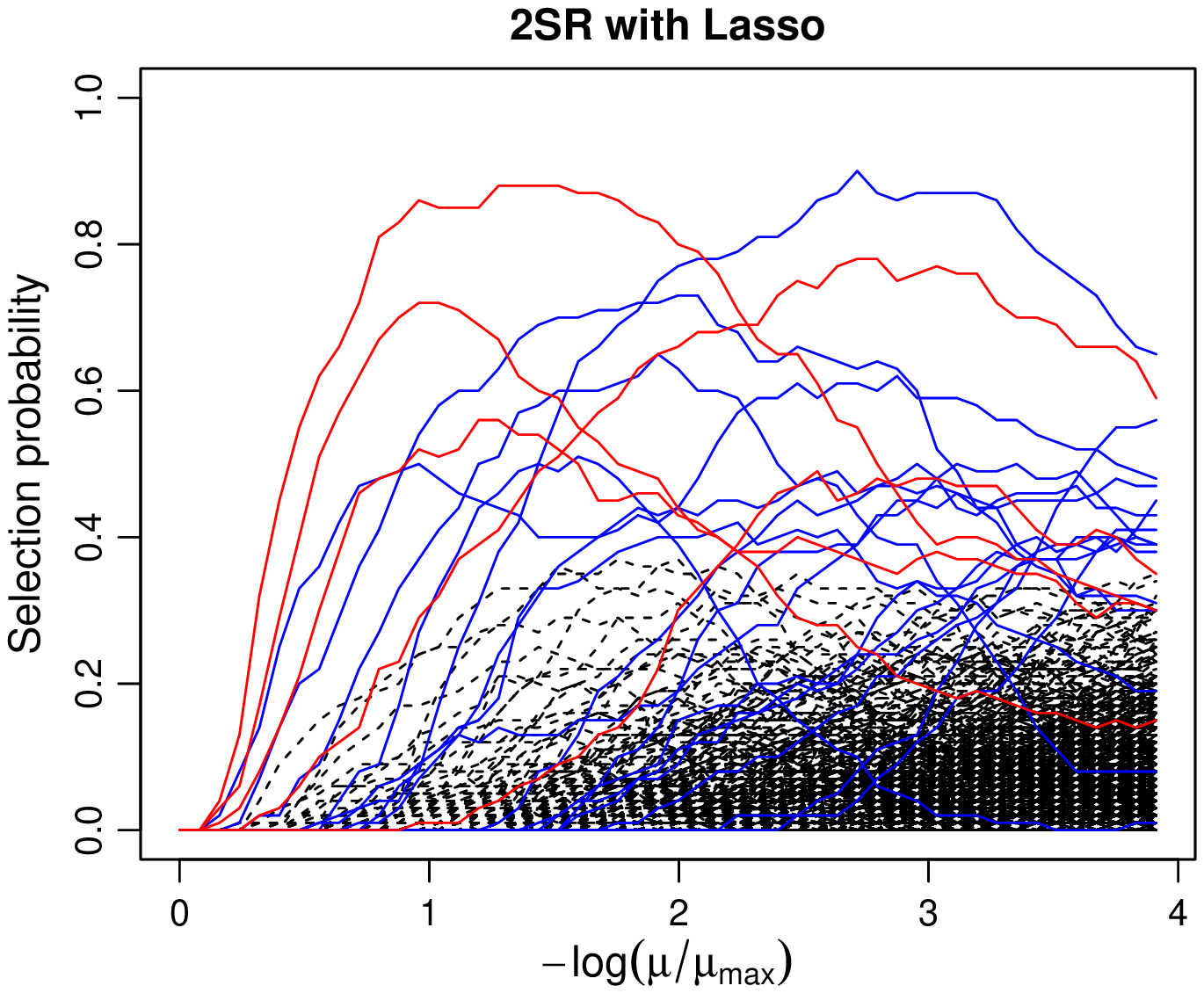}\\
\includegraphics[width=.5\textwidth]{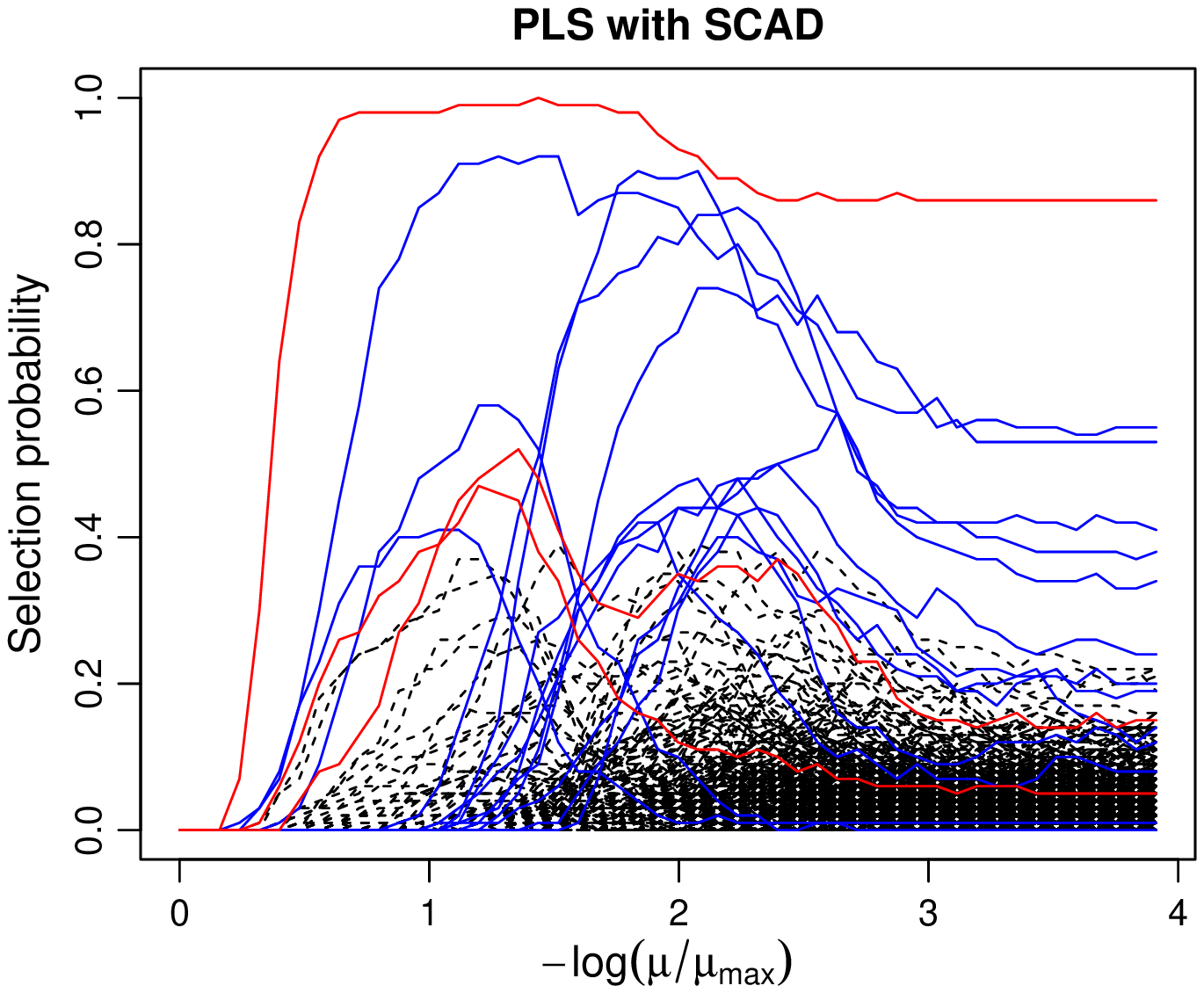}\includegraphics[width=.5\textwidth]{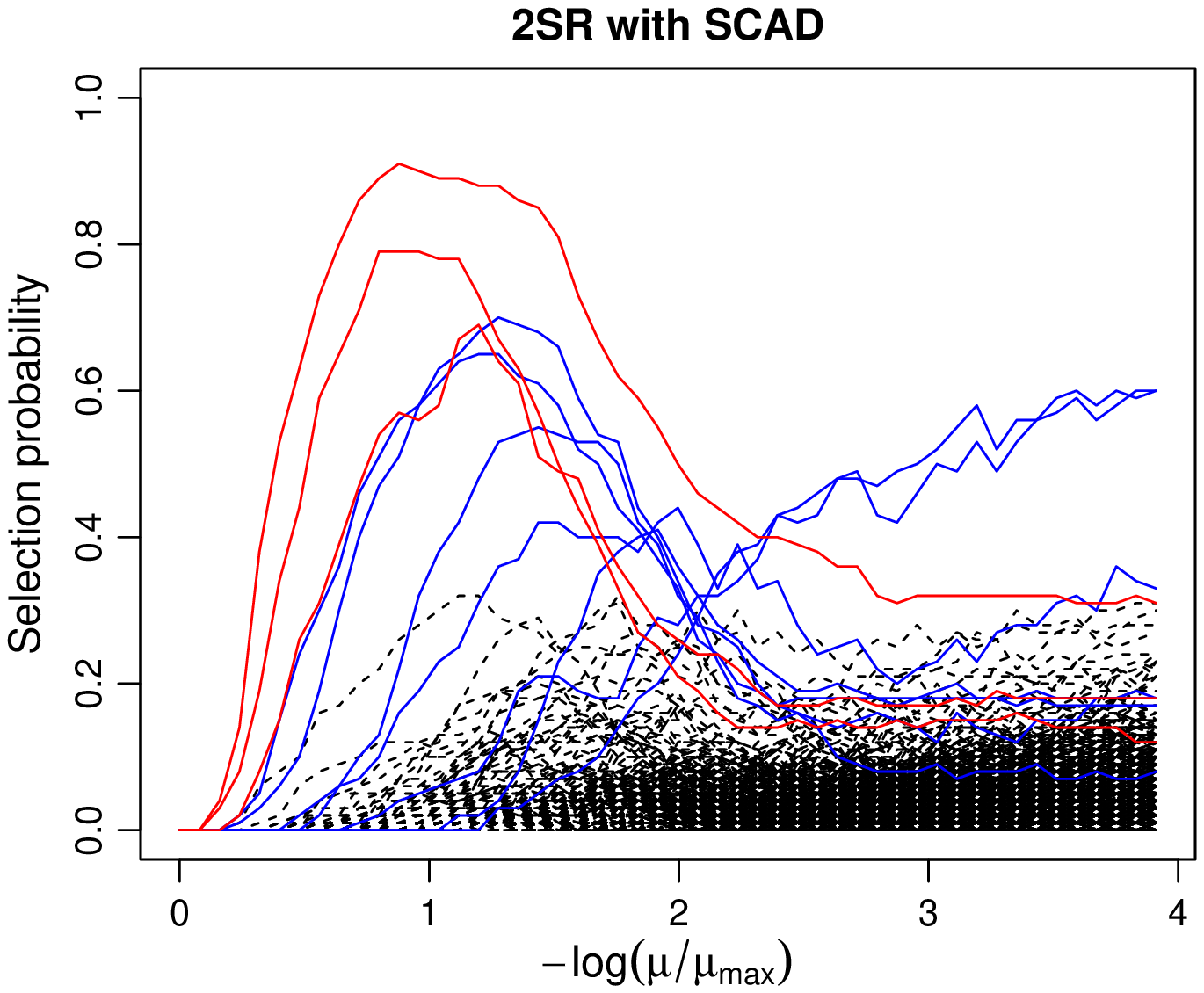}\\
\includegraphics[width=.5\textwidth]{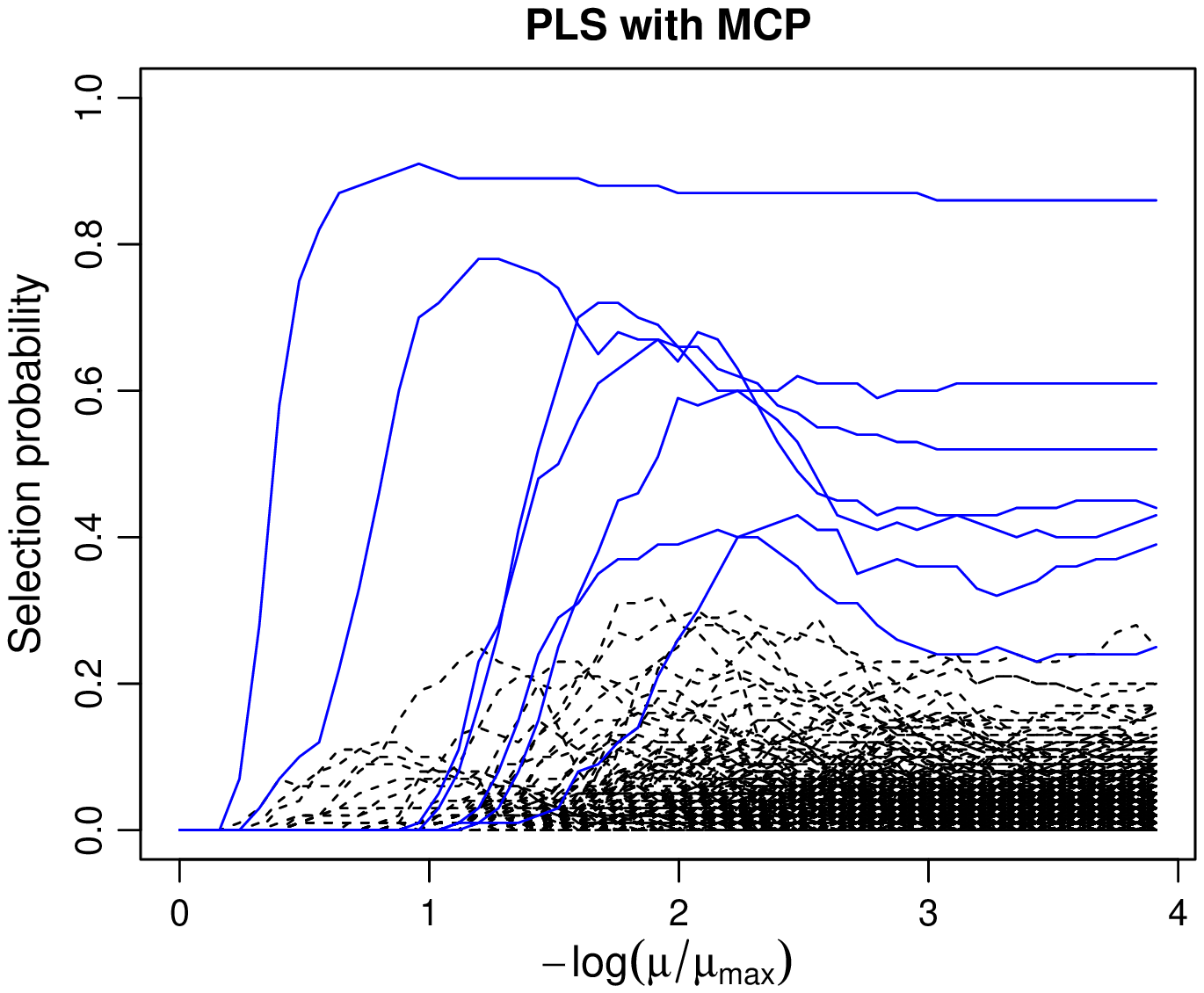}\includegraphics[width=.5\textwidth]{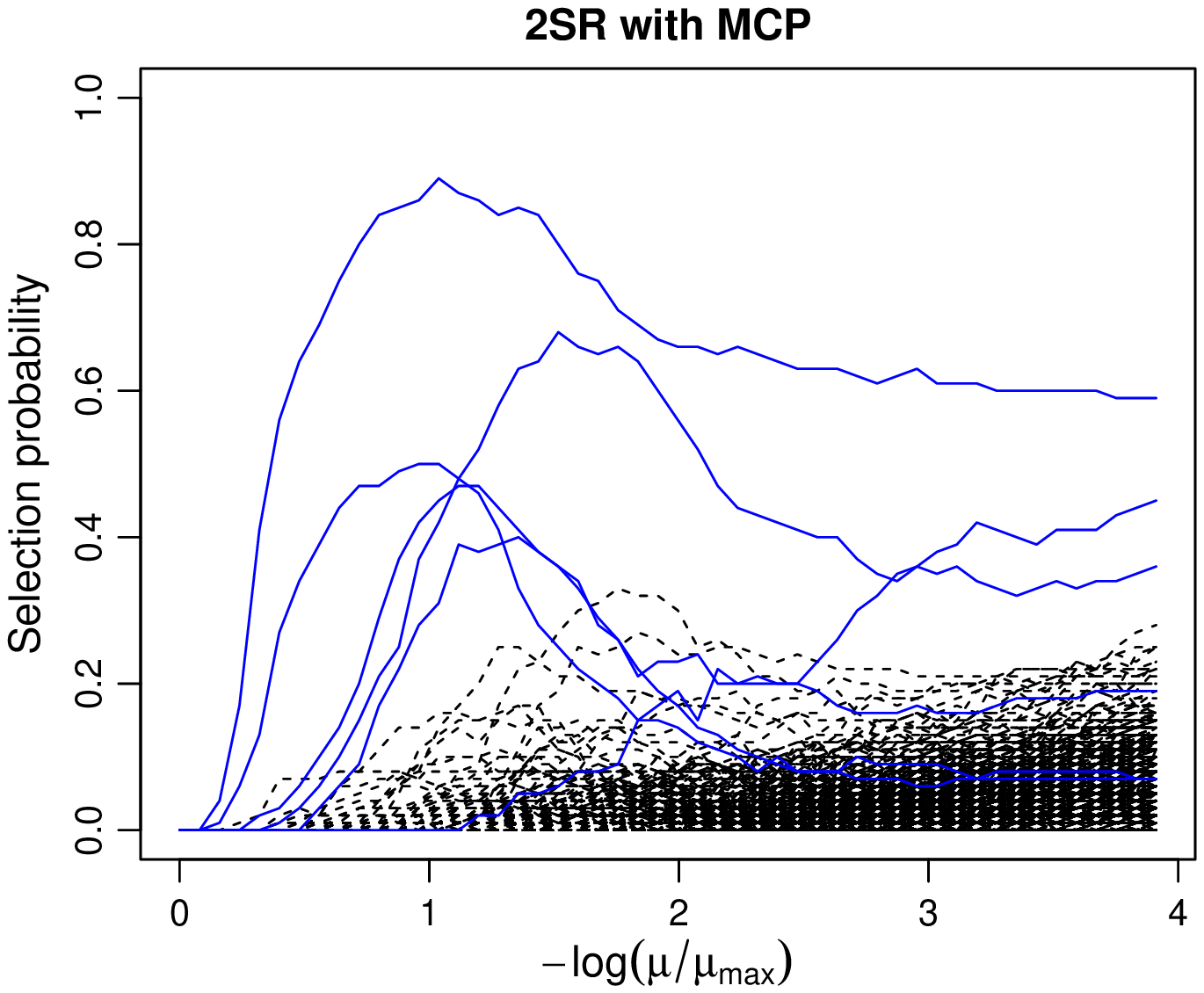}
\caption{Stability paths for different methods applied to the mouse obesity data based on 100 subsamples. Genes with maximum selection probability at least 0.4 are displayed in solid lines, among which genes common to both the PLS and 2SR methods are shown in red and the distinct ones in blue, and the remaining genes are displayed in dashed lines.}\label{fig:path}
\end{figure}

Table \ref{tab:mice_gene_stab} lists the genes that were chosen by stability selection with maximum selection probability at least 0.5 using the 2SR method with three different penalties. Among these 17 genes, only three were also selected by the PLS method. This includes insulin-like growth factor binding protein 2 (Igfbp2), which has been shown to protect against the development of obesity \citep{Whea:Kear:Shah:Ezza:Miel:igf-:2007}. Among the genes identified only by the 2SR method, apolipoprotein A-IV (Apoa4) plays an important role in lipoprotein metabolism and has been implicated in the control of food intake in rodents \citep{Tso:Sun:Liu:gast:2004}; it inhibits gastric emptying and serves as a satiety factor in response to ingestion of dietary fat. Apoa4 also acts as an enterogastrone, a humoral inhibitor of gastric acid secretion and motility \citep{Okum:Fuka:Tso:Tayl:Papp:intr:1994}, and is regulated by leptin, a major component of energy homeostasis \citep{Doi:Liu:Seel:Wood:Tso:effe:2001}. These previous findings suggest a potential role of Apoa4 in the regulation of food intake and, consequently, body weight. Suppressor of cytokine signaling 2 (Socs2) is a negative regulator in the growth hormone/insulin-like growth factor (IGF)-I signaling pathway \citep{Metc:Gree:Vine:Will:Star:giga:2000}, which is directly related to obesity. Slc22a3 is a downstream gene of the IGF signaling pathway. Recent studies have showed that the Gpld1 gene is associated with serum glycosylphosphatidylinositol-specific phospholipase D (GPI-PLD) levels, which predict changes in insulin sensitivity in response to a low-fat diet in obese women \citep{Gray:OBri:DAle:Breh:Deeg:plas:2008}. The IGF-binding protein also induces laminin gamma 1 (Lamc1) transcription \citep{Abra:Hans:insu:2010}. These identified genes clearly point out the importance of the IGF signaling pathway in the development of obesity in mice.

Table \ref{tab:mice_gene_stab} also presents the \emph{cis}-SNPs, which are defined to be the SNPs within a 10 cM distance of each gene, that are associated with the selected genes. These \emph{cis}-SNPs are likely to play a critical role in the regulation of the target genes and serve as strong instruments in statistical analysis. Not all selected genes have \emph{cis}-SNPs identified, partly due to the nonuniform, relatively sparse distribution of genotyped SNPs. If the criterion is relaxed to within 25 cM of each gene, we find that 13 of the 17 genes in the table have at least one \emph{cis}-SNP identified. Many of these \emph{cis}-SNPs coincide with QTLs detected for body weight traits in previous studies; for example, rs3663003 (Chr 1, 46.1 cM), rs4136518 (Chr 3, 54.6 cM), rs3694833 (Chr 10, 47.7 cM), rs4231406 (Chr 17, 12.0 cM), and rs3661189 (Chr 18, 27.5 cM) fall in previously detected QTL regions \citep{Roch:Eise:Vlec:Pomp:larg:2004}. Moreover, rs4231406 was previously identified as a QTL for atherosclerosis, which is strongly associated with body weight and adiposity \citep{Wang:Scha:Wang:Wang:Ingr:inde:2007}. These results demonstrate that our method can provide a more integrative, comprehensive understanding of the genetic architecture of complex traits than classical QTL analysis and gene expression studies, and would be useful for prioritizing candidate genes for complex diseases.

\begin{table}
\caption{Genes chosen by stability selection with maximum selection probability (values shown) at least 0.5 and \emph{cis}-SNPs (SNPs within 10 cM of each gene) identified by applying the 2SR method with different penalties to the mouse obesity data. Asterisks indicate genes that overlap those selected by the PLS method.}\label{tab:mice_gene_stab}\vskip1ex
\begin{tabular*}{\textwidth}{@{}l*{3}{@{\extracolsep{\fill}}c}@{\extracolsep{\fill}}l@{}}
\toprule\toprule
Gene          & Lasso & SCAD & MCP & \emph{cis}-SNPs\\
\midrule
Igfbp2$^*$    & 0.56 & 0.69 &      & rs3663003\\
Lamc1         &      & 0.70 &      & \\
Sirpa         & 0.51 & 0.55 &      & \\
Gstm2$^*$     & 0.88 & 0.91 & 0.89 & rs4136518\\
Ccnl2         & 0.50 &      &      & rs3720634\\
Glcci1        & 0.56 &      &      & \\
Vwf$^*$       & 0.72 & 0.79 & 0.50 & \\
Irx3          & 0.62 &      &      & \\
Apoa4         & 0.65 &      &      & \\
Socs2         &      &      & 0.68 & rs3694833\\
Avpr1a        & 0.78 &      &      & \\
Abca8a        & 0.50 &      &      & \\
Gpld1         & 0.50 &      &      & \\
Fam105a       &      & 0.60 &      & \\
Dscam         &      & 0.60 &      & \\
Slc22a3       & 0.90 &      &      & rs4137196, rs3722983, rs4231406\\
2010002N04Rik & 0.73 & 0.65 &      & rs3661189, rs3655324\\
\bottomrule
\end{tabular*}
\end{table}

\section{Discussion}
We have proposed a 2SR method for variable selection and estimation in sparse IV models where the dimensionality of covariates and instruments can both be much larger than the sample size. We have developed a high-dimensional theory that supports the theoretical advantages of our method and sheds light on the impact of dimensionality in the resulting procedure. We have applied our method to genetical genomics studies for jointly analyzing gene expression data and genetic variants to explore their associations with complex traits. The proposed method provides a powerful approach to effectively integrating and utilizing genotype, gene expression, and clinical data, which is of great importance for large-scale genomic inference. We have demonstrated on simulated and real data that our method is less affected by confounding and can lead to more reliable and biologically more interpretable results. Although we are primarily motivated by genetical genomics applications, the methodology is in fact very general and likely to find a wide range of applications in epidemiology, econometrics, and many other fields.

In our analysis of genetical genomics data, only genetic variants are used as instruments, and gene expressions that fail to be associated with any genetic variants in the first stage of the 2SR method have to be excluded at the second stage, which may comprise the inferences for genes with weak genetic effects. Epigenetic processes, such as DNA methylation, histone modification, and various RNA-mediated processes, are also known to play an essential role in the regulation of gene expression, and their influences on the gene expression levels may be profound \citep{Jaen:Bird:epig:2003}. Thus, when epigenetic data are also collected on the same subjects, they can be similarly treated as potential instruments in the sparse IV model. The joint consideration of genetic and epigenetic variants as instruments is likely to yield stronger instruments than using the genetic variants alone, which may lead to more reliable genomic inference.

Several extensions and improvements of the methodology are worthwhile to pursue. We have applied regularization methods to exploit the sparsity of individual coefficients, allowing the first stage to be decomposed into $p$ regression problems. While our general theory in Section \ref{sec:gen} applies to a generic first-stage estimator, the first-stage estimation and prediction could be improved by taking into account the correlations among the covariates and borrowing information across the $p$ subproblems. Two possibilities would be to exploit the structural sparsity of the coefficient matrix through certain matrix decompositions \citep[e.g.,][]{Chen:Chan:Sten:redu:2012,Chen:Huan:spar:2012}, and to jointly estimate the coefficient matrix and the covariance structure \citep[e.g.,][]{Roth:Levi:Zhu:spar:2010,Cai:Li:Liu:Xie:cova:2013}. Moreover, since the 2SR method is a high-dimensional extension of the classical 2SLS method, it would be natural to ask whether other IV estimators such as the LIML and GMM estimators can also be extended to our high-dimensional setting. Although asymptotic efficiency would not be a primary concern in high dimensions, certain advantages of these estimators in low dimensions may carry over and lead to performance improvement.

\appendix
\numberwithin{equation}{section}
\numberwithin{lem}{section}

\section*{Appendix: Proofs}
\addtocounter{section}{1}
\setcounter{equation}{0}

\subsection{Proof of Proposition \ref{prop:gap}}
We prove the result by contradiction. From \eqref{eq:bet_star} we have $\bX^T\bfeta_0=\bX^T\bX(\bbeta^*-\bbeta_0)$. If $\|\bbeta^*-\bbeta_0\|_1=o_P(1)$, then
\[
\left\|\frac{1}{n}\bX^T\bfeta_0\right\|_{\infty}\le\max_{1\le i,j\le p}\frac{1}{n}\bx_i^T\bx_j\|\bbeta^*-\bbeta_0\|_1\le\max_{1\le i,j\le p}\frac{1}{n}\|\bx_i\|_2\|\bx_j\|_2\|\bbeta^*-\bbeta_0\|_1=o_P(1).
\]
This yields a contradiction and completes the proof.

\subsection{Proof of Theorem \ref{thm:gam_est}}
Since the optimization problem \eqref{eq:opt1} can be decomposed into $p$ penalized least squares problems, the result is a straightforward extension of Theorem 7.2 of \citet{Bick:Rito:Tsyb:simu:2009} to the multivariate regression case. From Condition (C1) and the aforementioned result it follows that, with probability at least $1-q\exp(-n\lambda_j^2/(8\sigma_j^2))$, the regularized estimator $\what\bgamma_j$ defined by \eqref{eq:opt1'} with the $L_1$ penalty satisfies
\begin{equation}\label{eq:gam_j_est}
\|\what\bgamma_j-\bgamma_{0j}\|_1\le16r\lambda_j/\kappa_1^2
\end{equation}
and
\begin{equation}\label{eq:gam_j_pred}
\|\bZ(\what\bgamma_j-\bgamma_{0j})\|_2^2\le16nr\lambda_j^2/\kappa_1^2.
\end{equation}
Using the union bound, we have, with probability at least $1-\sum_{j=1}^pq\exp(-n\lambda_j^2/(8\sigma_j^2))$, the regularized estimator $\what\bGamma=(\what\bgamma_1,\dots,\what\bgamma_p)$ satisfies $\|\what\bGamma-\bGamma_0\|_1\le16r\lambda_{\max}/\kappa_1^2$ and $\|\bZ(\what\bGamma-\bGamma_0)\|_F^2\le16npr\lambda_{\max}^2/\kappa_1^2$. Now, if we choose $\lambda_j=C\sigma_j\sqrt{(\log p+\log q)/n}$ with a constant $C\ge2\sqrt{2}$, then with probability at least $1-(pq)^{1-C^2/8}$, the desired inequalities hold.

\subsection{Proof of Theorem \ref{thm:bet_est}}
The proof of Theorem \ref{thm:bet_est} relies on two key lemmas. Lemma \ref{lem:kappa} shows that Condition (C2), which is imposed on the matrix $\bZ\bGamma_0$, also holds with high probability for the matrix $\what\bX=\bZ\what\bGamma$. Lemma \ref{lem:ineq} establishes a fundamental inequality that is essential to the proof. To avoid repeatedly stating the probability bounds for certain inequalities to hold, we will occasionally condition on the events that these inequalities hold, and incorporate the probability bounds into the result by the union bound.

\begin{lem}\label{lem:kappa}
Under Conditions (C1) and (C2), if the regularization parameters $\lambda_j$ are chosen to satisfy
\begin{equation}\label{eq:rate1'}
\lambda_{\max}(2L+\lambda_{\max})\le\frac{\kappa_1^2\kappa_2^2}{32^2rs},
\end{equation}
then with probability at least $1-\sum_{j=1}^pq\exp(-n\lambda_j^2/(8\sigma_j^2))$, the matrix $\what\bX=\bZ\what\bGamma$, where $\what\bGamma$ is defined by \eqref{eq:opt1} with the $L_1$ penalty, satisfies
\[
\kappa(\what\bX,s)\ge\frac{\kappa_2}{2}.
\]
\end{lem}

\begin{proof}
For any subset $J\subset\{1,\dots,p\}$ with $|J|\le s$ and any $\bdelta\in\mathbb{R}^p$ with $\bdelta\ne\bzero$ and $\|\bdelta_{J^c}\|_1\le3\|\bdelta_J\|_1$, we have
\begin{align*}
\frac{\bdelta^T(\what\bX^T\what\bX-(\bZ\bGamma_0)^T\bZ\bGamma_0)\bdelta}{n\|\bdelta_J\|_2^2}
\le\frac{\|\bdelta\|_1^2\max_{1\le i,j\le p}|\what\bx_i^T\what\bx_j-(\bZ\bgamma_{0i})^T\bZ\bgamma_{0j}|}{n\|\bdelta_J\|_2^2}.
\end{align*}
Since $\|\bdelta_{J^c}\|_1\le3\|\bdelta_J\|_1$, we have $\|\bdelta\|_1^2=(\|\bdelta_J\|_1+\|\bdelta_{J^c}\|_1)^2\le16\|\bdelta_J\|_1^2\le16s\|\bdelta_J\|_2^2$. To bound the entrywise maximum, we write
\begin{align*}
&\what\bx_i^T\what\bx_j-(\bZ\bgamma_{0i})^T\bZ\bgamma_{0j}\\
&\quad=(\what\bx_i-\bZ\bgamma_{0i})^T(\what\bx_j-\bZ\bgamma_{0j}) +(\what\bx_i-\bZ\bgamma_{0i})^T\bZ\bgamma_{0j}+(\bZ\bgamma_{0i})^T(\what\bx_j-\bZ\bgamma_{0j})\\
&\quad=(\what\bgamma_i-\bgamma_{0i})^T\bZ^T\bZ(\what\bgamma_j-\bgamma_{0j})+(\what\bgamma_i-\bgamma_{0i})^T\bZ^T\bZ\bgamma_{0j} +(\bZ\bgamma_{0i})^T\bZ(\what\bgamma_j-\bgamma_{0j})\\
&\quad\equiv T_1+T_2+T_3.
\end{align*}
We now condition on the event that \eqref{eq:gam_j_est} and \eqref{eq:gam_j_pred} in the proof of Theorem 1 hold for $j=1,\dots,p$, which occurs with probability at least $1-\sum_{j=1}^pq\exp(-n\lambda_j^2/(8\sigma_j^2))$. Then, by the Cauchy--Schwarz inequality,
\[
|T_1|\le\|\bZ(\what\bgamma_i-\bgamma_{0i})\|_2\|\bZ(\what\bgamma_j-\bgamma_{0j})\|_2\le16nr\lambda_{\max}^2/\kappa_1^2.
\]
Also, noting that $\|\bz_j\|_2=\sqrt{n}$ by standardization and $\|\bGamma_0\|_1\le L$, we have
\begin{align*}
|T_2|&\le\|\what\bgamma_i-\bgamma_{0i}\|_1\max_{1\le k,l\le q}|\bz_k^T\bz_l|\|\bgamma_{0j}\|_1\\
&\le\|\what\bgamma_i-\bgamma_{0i}\|_1\max_{1\le k,l\le q}\|\bz_k\|_2\|\bz_l\|_2\|\bgamma_{0j}\|_1\le16Lnr\lambda_{\max}/\kappa_1^2.
\end{align*}
Similarly, $|T_3|\le16Lnr\lambda_{\max}/\kappa_1^2$. Combining these bounds and using the assumption \eqref{eq:rate1'}, we obtain
\[
\frac{\bdelta^T(\what\bX^T\what\bX-(\bZ\bGamma_0)^T\bZ\bGamma_0)\bdelta}{n\|\bdelta_J\|_2^2}\le\frac{16^2rs}{\kappa_1^2} \lambda_{\max}(2L+\lambda_{\max})\le\frac{16^2rs}{\kappa_1^2}\cdot\frac{\kappa_1^2\kappa_2^2}{32^2rs}=\left(\frac{\kappa_2}{2}\right)^2.
\]
This, together with Condition (C2), proves the lemma.
\end{proof}

\begin{lem}\label{lem:ineq}
Under Conditions (C1) and (C2), if $\kappa_1^{-2}r(\log p+\log q)=O(n)$ and the regularization parameters $\lambda_j$ are chosen as in \eqref{eq:lam}, then there exist constants $c_0,c_1,c_2>0$ such that, if we choose
\[
\mu=\frac{C_0}{\kappa_1}\sqrt{\frac{r(\log p+\log q)}{n}},
\]
where $C_0=c_0L\max(\sigma_{p+1},M\sigma_{\max})$, then with probability at least $1-c_1(pq)^{-c_2}$, the regularized estimator $\what\bbeta$ defined by \eqref{eq:opt2} with the $L_1$ penalty satisfies
\[
\frac{1}{2n}\|\what\bX(\what\bbeta-\bbeta_0)\|_2^2+\frac{\mu}{2}\|\what\bbeta-\bbeta_0\|_1\le2\mu\|\what\bbeta_S-\bbeta_{0S}\|.
\]
\end{lem}

\begin{proof}
By the optimality of $\what\bbeta$, we have
\[
\frac{1}{2n}\|\by-\what\bX\what\bbeta\|_2^2+\mu\|\what\bbeta\|_1\le\frac{1}{2n}\|\by-\what\bX\bbeta_0\|_2^2+\mu\|\bbeta_0\|_1.
\]
Substituting $\by=\bX\bbeta_0+\bfeta$, we write
\begin{align*}
\|\by-\what\bX\what\bbeta\|_2^2&=\|\bfeta-(\what\bX\what\bbeta-\bX\bbeta_0)\|_2^2\\
&=\|\bfeta\|_2^2+\|\what\bX\what\bbeta-\bX\bbeta_0\|_2^2-2\bfeta^T(\what\bX\what\bbeta-\bX\bbeta_0)\\
&=\|\bfeta\|_2^2+\|\what\bX(\what\bbeta-\bbeta_0)+(\what\bX-\bX)\bbeta_0\|_2^2-2\bfeta^T(\what\bX\what\bbeta-\bX\bbeta_0)\\
&=\|\bfeta\|_2^2+\|\what\bX(\what\bbeta-\bbeta_0)\|_2^2+\|(\what\bX-\bX)\bbeta_0\|_2^2-2\bfeta^T(\what\bX\what\bbeta-\bX\bbeta_0)\\ &\pheq{}+2\bbeta_0^T(\what\bX-\bX)^T\what\bX(\what\bbeta-\bbeta_0)
\end{align*}
and
\[
\|\by-\what\bX\bbeta_0\|_2^2=\|\bfeta-(\what\bX-\bX)\bbeta_0\|_2^2=\|\bfeta\|_2^2+\|(\what\bX-\bX)\bbeta_0\|_2^2-2\bfeta^T(\what\bX-\bX)\bbeta_0.
\]
Combining the last three displays yields
\begin{align}\label{eq:lem_ineq}
\begin{split}
\frac{1}{2n}\|\what\bX(\what\bbeta-\bbeta_0)\|_2^2 &\le\mu\|\bbeta_0\|_1-\mu\|\what\bbeta\|_1+\frac{1}{n}\bfeta^T\what\bX(\what\bbeta-\bbeta_0)-\frac{1}{n}\bbeta_0^T(\what\bX-\bX)^T\what\bX (\what\bbeta-\bbeta_0)\\
&\le\mu\|\bbeta_0\|_1-\mu\|\what\bbeta\|_1+\left\|\frac{1}{n}\what\bX^T\bfeta-\frac{1}{n}\what\bX^T(\what\bX-\bX)\bbeta_0\right\|_{\infty} \|\what\bbeta-\bbeta_0\|_1.
\end{split}
\end{align}
Next, we condition on the event that \eqref{eq:gam_j_est} and \eqref{eq:gam_j_pred} in the proof of Theorem 1 hold for $j=1,\dots,p$, which occurs with probability at least $1-(pq)^{1-C^2/8}$, and find a probability bound for the event that
\begin{equation}\label{eq:mu_inf}
\left\|\frac{1}{n}\what\bX^T\bfeta-\frac{1}{n}\what\bX^T(\what\bX-\bX)\bbeta_0\right\|_{\infty}\le\frac{\mu}{2}.
\end{equation}
Substituting $\what\bX=\bZ\what\bGamma$ and $\bX=\bZ\bGamma_0+\bE$, we write
\begin{align*}
&\frac{1}{n}\what\bX^T\bfeta-\frac{1}{n}\what\bX^T(\what\bX-\bX)\bbeta_0\\
&\quad=\frac{1}{n}\smash{\what\bGamma}^T\bZ^T\bfeta-\frac{1}{n}\smash{\what\bGamma}^T\bZ^T(\bZ\what\bGamma-\bZ\bGamma_0-\bE)\bbeta_0\\
&\quad=\frac{1}{n}(\what\bGamma-\bGamma_0)^T\bZ^T\bfeta+\frac{1}{n}\bGamma_0^T\bZ^T\bfeta+\frac{1}{n}(\what\bGamma-\bGamma_0)^T\bZ^T\bE\bbeta_0 +\frac{1}{n}\bGamma_0^T\bZ^T\bE\bbeta_0\\
&\quad\pheq{}-\frac{1}{n}(\what\bGamma-\bGamma_0)^T\bZ^T\bZ(\what\bGamma-\bGamma_0)\bbeta_0-\frac{1}{n}\bGamma_0^T\bZ^T\bZ(\what\bGamma-\bGamma_0)\bbeta_0\\
&\quad\equiv T_1+T_2+T_3+T_4+T_5+T_6.
\end{align*}
To bound term $T_1$, it follows from \eqref{eq:gam_j_est}, the union bound, and the classical Gaussian tail bound that
\[
P\left(\|T_1\|_{\infty}\ge\frac{\mu}{12}\right)
\le P\left(\left\|\frac{1}{n}\bZ^T\bfeta\right\|_{\infty}\ge\frac{\kappa_1^2}{16r\lambda_{\max}}\cdot\frac{\mu}{12}\right)
\le q\exp\left\{-\frac{n}{2\sigma_{p+1}^2}\left(\frac{\kappa_1^2}{16r\lambda_{\max}}\cdot\frac{\mu}{12}\right)^2\right\}.
\]
Noting that $\|\bGamma_0\|_1\le L$, we have
\[
P\left(\|T_2\|_{\infty}\ge\frac{\mu}{12}\right)\le P\left(\left\|\frac{1}{n}\bZ^T\bfeta\right\|_{\infty}\ge\frac{\mu}{12L}\right)\le q\exp\left\{-\frac{n} {2\sigma_{p+1}^2}\left(\frac{\mu}{12L}\right)^2\right\}.
\]
To bound term $T_3$, using \eqref{eq:gam_j_est} and $\|\bbeta_0\|_1\le M$, we obtain
\begin{align*}
P\left(\|T_3\|_{\infty}\ge\frac{\mu}{12}\right)&\le P\left(\max_{1\le i\le q,\,1\le j\le p}\left|\frac{1}{n}\bz_i^T\bve_j\right|\ge\frac{\kappa_1^2}{16r\lambda_{\max}}\cdot \frac{\mu}{12M}\right)\\
&\le pq\exp\left\{-\frac{n}{2\sigma_{\max}^2}\left(\frac{\kappa_1^2}{16r\lambda_{\max}}\cdot\frac{\mu}{12M}\right)^2\right\},
\end{align*}
where $\bve_j$ is the $j$th column of the matrix $\bE$. Similarly,
\[
P\left(\|T_4\|_{\infty}\ge\frac{\mu}{12}\right)\le P\left\{\max_{1\le i\le q,\,1\le j\le p}\left|\frac{1}{n}\bz_i^T\bve_j\right|\ge\frac{\mu}{12LM}\right\}
\le pq\exp\left\{-\frac{n}{2\sigma_{\max}^2}\left(\frac{\mu}{12LM}\right)^2\right\}.
\]
To bound term $T_5$, it follows from \eqref{eq:gam_j_pred}, $\|\bbeta_0\|_1\le M$, and the Cauchy--Schwarz inequality that
\begin{align*}
\|T_5\|_{\infty}&\le M\max_{1\le i,j\le p}\left|\frac{1}{n}(\what\bgamma_i-\bgamma_{0i})^T\bZ^T\bZ(\what\bgamma_j-\bgamma_{0j})\right|\\
&\le M\max_{1\le i,j\le p}\frac{1}{n}\|\bZ(\what\bgamma_i-\bgamma_{0i})\|_2\|\bZ(\what\bgamma_j-\bgamma_{0j})\|_2\le M\frac{16r\lambda_{\max}^2}{\kappa_1^2}.
\end{align*}
Noting that $\|\bz_j\|_2=\sqrt{n}$ by standardization, we have
\begin{align*}
\|T_6\|_{\infty}&\le LM\max_{1\le i\le q,\,1\le j\le p}\left|\frac{1}{n}\bz_i^T\bZ(\what\bgamma_j-\bgamma_{0j})\right|\\
&\le LM\max_{1\le j\le p}\frac{1}{\sqrt{n}}\|\bZ(\what\bgamma_j-\bgamma_{0j})\|_2\le LM\frac{4\sqrt{r}\lambda_{\max}}{\kappa_1}.
\end{align*}
Combining these bounds and in view of the assumption $\kappa_1^{-2}r(\log p+\log q)=O(n)$, there exist constants $c_0,c_1,c_2>0$ such that, if we choose
\[
\mu=\frac{C_0}{\kappa_1}\sqrt{\frac{r(\log p+\log q)}{n}},
\]
where $C_0=c_0L\max(\sigma_{p+1},M\sigma_{\max})$, then \eqref{eq:mu_inf} holds with probability at least $1-c_1(pq)^{-c_2}$. This, together with \eqref{eq:lem_ineq}, implies
\[
\frac{1}{2n}\|\what\bX(\what\bbeta-\bbeta_0)\|_2^2\le\mu\|\bbeta_0\|_1-\mu\|\what\bbeta\|_1+\frac{\mu}{2}\|\what\bbeta-\bbeta_0\|_1.
\]
Adding $\mu\|\what\bbeta-\bbeta_0\|_1/2$ to both sides yields
\begin{align*}
&\frac{1}{2n}\|\what\bX(\what\bbeta-\bbeta_0)\|_2^2+\frac{\mu}{2}\|\what\bbeta-\bbeta_0\|_1 \le\mu(\|\bbeta_0\|_1-\|\what\bbeta\|_1+\|\what\bbeta-\bbeta_0\|_1)\\
&\quad=\mu(\|\bbeta_{0S}\|_1-\|\what\bbeta_S\|_1+\|\what\bbeta_S-\bbeta_{0S}\|_1)\le2\mu\|\what\bbeta_S-\bbeta_{0S}\|_1.
\end{align*}
This completes the proof of the lemma.
\end{proof}

\begin{proof}[Proof of Theorem \ref{thm:bet_est}]
We first note that \eqref{eq:lam} and \eqref{eq:rate1} imply that the condition $\kappa_1^{-2}r(\log p+\log q)=O(n)$ is satisfied. Then it follows from Lemma \ref{lem:ineq} that, with probability at least $1-c_1(pq)^{-c_2}$, we have
\begin{equation}\label{eq:thm2_ineq1}
\frac{1}{2n}\|\what\bX(\what\bbeta-\bbeta_0)\|_2^2\le2\mu\|\what\bbeta_S-\bbeta_{0S}\|_1\le2\mu\sqrt{s}\|\what\bbeta_S-\bbeta_{0S}\|_2
\end{equation}
and
\[
\frac{\mu}{2}\|\what\bbeta-\bbeta_0\|_1\le2\mu\|\what\bbeta_S-\bbeta_{0S}\|_1,\quad\text{or }
\|\what\bbeta_{S^c}-\bbeta_{S^c}^0\|_1\le3\|\what\bbeta_S-\bbeta_{0S}\|_1.
\]
The last inequality, the definition of $\kappa(\what\bX,s)$, and Lemma 1 together imply
\begin{equation}\label{eq:thm2_ineq2}
\|\what\bbeta_S-\bbeta_{0S}\|_2\le\frac{\|\what\bX(\what\bbeta-\bbeta_0)\|_2}{\sqrt{n}\kappa(\what\bX,s)} \le\frac{2\|\what\bX(\what\bbeta-\bbeta_0)\|_2}{\sqrt{n}\kappa_2}.
\end{equation}
Combining \eqref{eq:thm2_ineq1} and \eqref{eq:thm2_ineq2}, we obtain
\[
\|\what\bX(\what\bbeta-\bbeta_0)\|_2^2\le\frac{64}{\kappa_2^2}ns\mu^2
\]
and
\[
\|\what\bbeta-\bbeta_0\|_1\le4\|\what\bbeta_S-\bbeta_{0S}\|_1\le4\sqrt{s}\|\what\bbeta_S-\bbeta_{0S}\|_2\le\frac{64}{\kappa_2^2}s\mu.
\]
Substituting \eqref{eq:mu} for $\mu$ concludes the proof.
\end{proof}

\subsection{Proof of Theorem \ref{thm:bet_sel}}
Central to the proof of Theorem 3 is the following lemma, which shows that Condition (C3) also holds with high probability for the matrix $\what\bX=\bZ\what\bGamma$ and gives a useful bound for the inverse matrix norm $\|(\what\bC_{SS})^{-1}\|_{\infty}$.

\begin{lem}\label{lem:alpha}
Under Condition (C3), if the regularization parameters $\lambda_j$ are chosen to satisfy
\begin{equation}\label{eq:rate2'}
\frac{16\varphi}{\kappa_1^2}rs\lambda_{\max}(2L+\lambda_{\max})\le\frac{\alpha}{4-\alpha},
\end{equation}
then with probability at least $1-\sum_{j=1}^pq\exp(-n\lambda_j^2/(8\sigma_j^2))$, the matrix $\what\bC=n^{-1}\what\bX^T\what\bX =n^{-1}(\bZ\what\bGamma)^T\bZ\what\bGamma$, where $\what\bGamma$ is defined by \eqref{eq:opt1} with the $L_1$ penalty, satisfies
\begin{equation}\label{eq:lem3ineq1}
\|(\what\bC_{SS})^{-1}\|_{\infty}\le\frac{4-\alpha}{2(2-\alpha)}\varphi
\end{equation}
and
\begin{equation}\label{eq:lem3ineq2}
\|\what\bC_{S^cS}(\what\bC_{SS})^{-1}\|_{\infty}\le1-\frac{\alpha}{2}.
\end{equation}
\end{lem}

\begin{proof}
We condition on the event that \eqref{eq:gam_j_est} and \eqref{eq:gam_j_pred} in the proof of Theorem 1 hold for $j=1,\dots,p$, which occurs with probability at least $1-\sum_{j=1}^pq\exp(-n\lambda_j^2/(8\sigma_j^2))$. From the proof of Lemma \ref{lem:kappa}, we have
\[
\max_{1\le i,j\le p}\frac{1}{n}|\what\bx_i^T\what\bx_j-(\bZ\bgamma_{0i})^T\bZ\bgamma_{0j}|\le\frac{16}{\kappa_1^2}r\lambda_{\max}(2L+\lambda_{\max}).
\]
This, along with the assumption \eqref{eq:rate2'}, gives
\begin{equation}\label{eq:CSS}
\varphi\|\what\bC_{SS}-\bC_{SS}\|_{\infty}\le\frac{16\varphi}{\kappa_1^2}rs\lambda_{\max}(2L+\lambda_{\max})\le\frac{\alpha}{4-\alpha},
\end{equation}
and similarly,
\begin{equation}\label{eq:CScS}
\varphi\|\what\bC_{S^cS}-\bC_{S^cS}\|_{\infty}\le\frac{\alpha}{4-\alpha}.
\end{equation}
Then, by an error bound for matrix inversion \citep[p.\,336]{Horn:John:matr:1985}, we have
\[
\|(\what\bC_{SS})^{-1}-(\bC_{SS})^{-1}\|_{\infty}\le\frac{\varphi\|\what\bC_{SS}-\bC_{SS}\|_{\infty}}{1-\varphi\|\what\bC_{SS}-\bC_{SS}\|_{\infty}} \varphi\le\frac{\alpha}{2(2-\alpha)}\varphi.
\]
The triangle inequality implies
\[
\|(\what\bC_{SS})^{-1}\|_{\infty}\le\|(\bC_{SS})^{-1}\|_{\infty}+\|(\what\bC_{SS})^{-1}-(\bC_{SS})^{-1}\|_{\infty}\le\varphi+\frac{\alpha}{2(2-\alpha)}\varphi =\frac{4-\alpha}{2(2-\alpha)}\varphi,
\]
which proves \eqref{eq:lem3ineq1}.

To show inequality \eqref{eq:lem3ineq2}, we write
\[
\what\bC_{S^cS}(\what\bC_{SS})^{-1}-\bC_{S^cS}(\bC_{SS})^{-1}=(\what\bC_{S^cS}-\bC_{S^cS})(\what\bC_{SS})^{-1}-\bC_{S^cS}(\bC_{SS})^{-1} (\what\bC_{SS}-\bC_{SS})(\what\bC_{SS})^{-1}.
\]
Then it follows from \eqref{eq:lem3ineq1}, \eqref{eq:CSS}, \eqref{eq:CScS}, and Condition (C3) that
\begin{align*}
&\|\what\bC_{S^cS}(\what\bC_{SS})^{-1}-\bC_{S^cS}(\bC_{SS})^{-1}\|_{\infty}\\
&\quad\le\|\what\bC_{S^cS}-\bC_{S^cS}\|_{\infty}\|(\what\bC_{SS})^{-1}\|_{\infty}+\|\bC_{S^cS}(\bC_{SS})^{-1}\|_{\infty} \|\what\bC_{SS}-\bC_{SS}\|_{\infty}\|(\what\bC_{SS})^{-1}\|_{\infty}\\
&\quad\le\frac{\alpha}{(4-\alpha)\varphi}\frac{4-\alpha}{2(2-\alpha)}\varphi+(1-\alpha)\frac{\alpha}{(4-\alpha)\varphi}\frac{4-\alpha}{2(2-\alpha)}\varphi =\frac{\alpha}{2},
\end{align*}
which, along with Condition (C3), proves \eqref{eq:lem3ineq2}. This completes the proof of the lemma.
\end{proof}

\begin{proof}[Proof of Theorem \ref{thm:bet_sel}]
For an index set $J$, let $\bX_J$ and $\what\bX_J$ denote the submatrices formed by the $j$th columns of $\bX$ and $\what\bX$ with $j\in J$, respectively. The optimality conditions for $\what\bbeta\in\mathbb{R}^p$ to be a solution to problem \eqref{eq:opt2} with the $L_1$ penalty can be written as
\begin{equation}\label{eq:kkt1}
\frac{1}{n}\what\bX_{\what{S}}^T(\by-\what\bX\what\bbeta)=\mu\sgn(\what\bbeta_{\what{S}})
\end{equation}
and
\begin{equation}\label{eq:kkt2}
\left\|\frac{1}{n}\what\bX_{{\what{S}}^c}^T(\by-\what\bX\what\bbeta)\right\|_{\infty}\le\mu.
\end{equation}
It suffices to find a $\what\bbeta\in\mathbb{R}^p$ with the desired properties such that \eqref{eq:kkt1} and \eqref{eq:kkt2} hold. Let $\what\bbeta_{S^c}=\bzero$. The idea of the proof is to first determine $\what\bbeta_S$ from \eqref{eq:kkt1}, and then show that thus obtained $\what\bbeta$ also satisfies \eqref{eq:kkt2}.

Using similar arguments to those in the proof of Lemma \ref{lem:ineq}, we can show that, there exist constants $c_0,c_1,c_2>0$ such that, if we can choose $\mu$ as before, then with probability at least $1-c_1(pq)^{-c_2}$, it holds that
\begin{equation}\label{eq:mu_inf_alp}
\left\|\frac{1}{n}\what\bX^T\bfeta-\frac{1}{n}\what\bX^T(\what\bX-\bX)\bbeta_0\right\|_{\infty}\le\frac{\alpha}{4-\alpha}\mu.
\end{equation}
From now on, we condition on the event that \eqref{eq:mu_inf_alp} holds and analyze conditions \eqref{eq:kkt1} and \eqref{eq:kkt2}.

We first determine $\what\bbeta_S$ from \eqref{eq:kkt1}. By substituting
\begin{equation}\label{eq:y_sub}
\by-\what\bX\what\bbeta=\bX_S\bbeta_{0S}+\bfeta-\what\bX_S\what\bbeta_S=\bfeta-(\what\bX_S-\bX_S)\bbeta_{0S}-\what\bX_S (\what\bbeta_S-\bbeta_{0S}),
\end{equation}
we write \eqref{eq:kkt1} with $\what{S}$ replaced by $S$ in the form
\[
\frac{1}{n}\what\bX_S^T\bfeta-\frac{1}{n}\what\bX_S^T(\what\bX_S-\bX_S)\bbeta_{0S}-\what\bC_{SS}(\what\bbeta_S-\bbeta_{0S}) =\mu\sgn(\what\bbeta_S),
\]
or
\begin{equation}\label{eq:bet_def}
\what\bbeta_S-\bbeta_{0S}=(\what\bC_{SS})^{-1}\left\{\frac{1}{n}\what\bX_S^T\bfeta-\frac{1}{n}\what\bX_S^T(\what\bX_S-\bX_S)\bbeta_{0S} -\mu\sgn(\what\bbeta_S)\right\}.
\end{equation}
This, along with \eqref{eq:lem3ineq1} and \eqref{eq:mu_inf_alp}, leads to
\begin{align*}
\|\what\bbeta_S-\bbeta_{0S}\|_{\infty}&\le\|(\what\bC_{SS})^{-1}\|_{\infty}\left\{\left\|\frac{1}{n}\what\bX_S^T\bfeta-\frac{1}{n}\what\bX_S^T (\what\bX_S-\bX_S)\bbeta_{0S}\right\|_{\infty}+\mu\right\}\\
&\le\frac{4-\alpha}{2(2-\alpha)}\varphi\left(\frac{\alpha}{4-\alpha}\mu+\mu\right)=\frac{2}{2-\alpha}\varphi\mu<b_0
\end{align*}
by assumption, which entails that $\sgn(\what\bbeta_S)=\sgn(\bbeta_{0S})$. Since $\what\bbeta_{S^c}=\bbeta_{S^c}^0=\bzero$ by definition, we have $\what{S}=S$. Let $\what\bbeta_S$ be defined by \eqref{eq:bet_def} with $\sgn(\what\bbeta_S)$ replaced by $\sgn(\bbeta_{0S})$. Clearly, thus defined $\what\bbeta$ satisfies the desired properties and \eqref{eq:kkt1}.

It remains to show that $\what\bbeta$ also satisfies \eqref{eq:kkt2}. By substituting \eqref{eq:y_sub} and \eqref{eq:bet_def}, we write
\begin{align*}
\frac{1}{n}\what\bX_{S^c}^T(\by-\what\bX\what\bbeta)&=\frac{1}{n}\what\bX_{S^c}^T\bfeta-\frac{1}{n}\what\bX_{S^c}^T(\what\bX_S-\bX_S)\bbeta_{0S}\\ &\pheq{}-\what\bC_{S^cS}(\what\bC_{SS})^{-1}\left\{\frac{1}{n}\what\bX_S^T\bfeta-\frac{1}{n}\what\bX_S^T(\what\bX_S-\bX_S)\bbeta_{0S} -\mu\sgn(\what\bbeta_S)\right\}.
\end{align*}
Then it follows from \eqref{eq:lem3ineq2} and \eqref{eq:mu_inf_alp} that
\begin{align*}
\left\|\frac{1}{n}\what\bX_{S^c}^T(\by-\what\bX\what\bbeta)\right\|_{\infty}
&\le\left\|\frac{1}{n}\what\bX_{S^c}^T\bfeta-\frac{1}{n}\what\bX_{S^c}^T(\what\bX_S-\bX_S)\bbeta_{0S}\right\|_{\infty}\\
&\pheq{}+\|\what\bC_{S^cS}(\what\bC_{SS})^{-1}\|_{\infty}\left\{\left\|\frac{1}{n}\what\bX_S^T\bfeta-\frac{1}{n}\what\bX_S^T(\what\bX_S-\bX_S) \bbeta_{0S}\right\|_{\infty}+\mu\right\}\\
&\le\frac{\alpha}{4-\alpha}\mu+\left(1-\frac{\alpha}{2}\right)\left(\frac{\alpha}{4-\alpha}\mu+\mu\right)=\mu.
\end{align*}
Since $\what{S}=S$, we see that $\what\bbeta$ also satisfies \eqref{eq:kkt2}, which concludes the proof.
\end{proof}

\section*{Supplementary Materials}
The supplementary material contains the proof of Theorem 4.

\bibliographystyle{jasa}
\bibliography{iv}

\begin{thebibliography}{55}

\bibitem[{Abrass and Hansen(2010)}]{Abra:Hans:insu:2010}
Abrass, C.~K., and Hansen, K.~M. (2010), ``Insulin-Like Growth Factor-Binding
  Protein-5-Induced Laminin $\gamma$1 Transcription Requires Filamin {A},''
  \textit{Journal of Biological Chemistry}, 285, 12925--12934.

\bibitem[{Anderson(2005)}]{Ande:orig:2005}
Anderson, T.~W. (2005), ``Origins of the Limited Information Maximum Likelihood
  and Two-Stage Least Squares Estimators,'' \textit{Journal of Econometrics},
  127, 1--16.

\bibitem[{Belloni et~al.(2012)Belloni, Chen, Chernozhukov, and
  Hansen}]{Bell:Chen:Cher:Hans:spar:2012}
Belloni, A., Chen, D., Chernozhukov, V., and Hansen, C. (2012), ``Sparse Models
  and Methods for Optimal Instruments With an Application to Eminent Domain,''
  \textit{Econometrica}, 80, 2369--2429.

\bibitem[{Bickel, Ritov,  and Tsybakov(2009)Bickel, Ritov, and
  Tsybakov}]{Bick:Rito:Tsyb:simu:2009}
Bickel, P.~J., Ritov, Y., and Tsybakov, A.~B. (2009), ``Simultaneous Analysis
  of {L}asso and {D}antzig Selector,'' \textit{The Annals of Statistics}, 37,
  1705--1732.

\bibitem[{Browning and Browning(2007)}]{Brow:Brow:rapi:2007}
Browning, S.~R., and Browning, B.~L. (2007), ``Rapid and Accurate Haplotype
  Phasing and Missing-Data Inference for Whole-Genome Association Studies by
  Use of Localized Haplotype Clustering,'' \textit{American Journal of Human
  Genetics}, 81, 1084--1097.

\bibitem[{Cai et~al.(2013)Cai, Li, Liu, and Xie}]{Cai:Li:Liu:Xie:cova:2013}
Cai, T.~T., Li, H., Liu, W., and Xie, J. (2013), ``Covariate-Adjusted Precision
  Matrix Estimation With an Application in Genetical Genomics,''
  \textit{Biometrika}, 100, 139--156.

\bibitem[{Caner(2009)}]{Cane:lass:2009}
Caner, M. (2009), ``Lasso-Type {GMM} Estimator,'' \textit{Econometric Theory},
  25, 270--290.

\bibitem[{Carrasco(2012)}]{Carr:regu:2012}
Carrasco, M. (2012), ``A Regularization Approach to the Many Instruments
  Problem,'' \textit{Journal of Econometrics}, 170, 383--398.

\bibitem[{Chao and Swanson(2005)}]{Chao:Swan:cons:2005}
Chao, J.~C., and Swanson, N.~R. (2005), ``Consistent Estimation With a Large
  Number of Weak Instruments,'' \textit{Econometrica}, 73, 1673--1692.

\bibitem[{Chen, Chan,  and Stenseth(2012)Chen, Chan, and
  Stenseth}]{Chen:Chan:Sten:redu:2012}
Chen, K., Chan, K.-S., and Stenseth, N.~C. (2012), ``Reduced Rank Stochastic
  Regression With a Sparse Singular Value Decomposition,'' \textit{Journal of
  the Royal Statistical Society, \emph{Series B}}, 74, 203--221.

\bibitem[{Chen and Huang(2012)}]{Chen:Huan:spar:2012}
Chen, L., and Huang, J.~Z. (2012), ``Sparse Reduced-Rank Regression for
  Simultaneous Dimension Reduction and Variable Selection,'' \textit{Journal of
  the American Statistical Association}, 107, 1533--1545.

\bibitem[{Didelez, Meng,  and Sheehan(2010)Didelez, Meng, and
  Sheehan}]{Dide:Meng:Shee:assu:2010}
Didelez, V., Meng, S., and Sheehan, N.~A. (2010), ``Assumptions of {IV} Methods
  for Observational Epidemiology,'' \textit{Statistical Science}, 25, 22--40.

\bibitem[{Didelez and Sheehan(2007)}]{Dide:Shee:mend:2007}
Didelez, V., and Sheehan, N. (2007), ``Mendelian Randomization as an
  Instrumental Variable Approach to Causal Inference,'' \textit{Statistical
  Methods in Medical Research}, 16, 309--330.

\bibitem[{Doi et~al.(2001)Doi, Liu, Seeley, Woods, and
  Tso}]{Doi:Liu:Seel:Wood:Tso:effe:2001}
Doi, T., Liu, M., Seeley, R.~J., Woods, S.~C., and Tso, P. (2001), ``Effect of
  Leptin on Intestinal Apolipoprotein {AIV} in Response to Lipid Feeding,''
  \textit{American Journal of Physiology Regulatory, Integrative and
  Comparative Physiology}, 281, R753--R759.

\bibitem[{Emilsson et~al.(2008)Emilsson, Thorleifsson, Zhang, Leonardson, Zink,
  Zhu, Carlson, Helgason, Walters, Gunnarsdottir, Mouy, Steinthorsdottir,
  Eiriksdottir, Bjornsdottir, Reynisdottir, Gudbjartsson, Helgadottir,
  Jonasdottir, Jonasdottir, Styrkarsdottir, Gretarsdottir, Magnusson,
  Stefansson, Fossdal, Kristjansson, Gislason, Stefansson, Leifsson,
  Thorsteinsdottir, Lamb, Gulcher, Reitman, Kong, Schadt, and
  Stefansson}]{Emil:Thor:Gudm:Zhan:Leon:gene:2008}
Emilsson, V., Thorleifsson, G., Zhang, B., Leonardson, A.~S., Zink, F., Zhu,
  J., Carlson, S., Helgason, A., Walters, G.~B., Gunnarsdottir, S., Mouy, M.,
  Steinthorsdottir, V., Eiriksdottir, G.~H., Bjornsdottir, G., Reynisdottir,
  I., Gudbjartsson, D., Helgadottir, A., Jonasdottir, A., Jonasdottir, A.,
  Styrkarsdottir, U., Gretarsdottir, S., Magnusson, K.~P., Stefansson, H.,
  Fossdal, R., Kristjansson, K., Gislason, H.~G., Stefansson, T., Leifsson,
  B.~G., Thorsteinsdottir, U., Lamb, J.~R., Gulcher, J.~R., Reitman, M.~L.,
  Kong, A., Schadt, E.~E., and Stefansson, K. (2008), ``Genetics of Gene
  Expression and Its Effect on Disease,'' \textit{Nature}, 452, 423--428.

\bibitem[{Eppig et~al.(2012)Eppig, Blake, Bult, Kadin, Richardson, and the
  Mouse Genome Database~Group}]{Eppi:Blak:Bult:Kadi:Rich:mous:2012}
Eppig, J.~T., Blake, J.~A., Bult, C.~J., Kadin, J.~A., Richardson, J.~E., and
  the Mouse Genome Database~Group (2012), ``The {M}ouse {G}enome {D}atabase
  ({MGD}): Comprehensive Resource for Genetics and Genomics of the Laboratory
  Mouse,'' \textit{Nucleic Acids Research}, 40, D881--D886.

\bibitem[{Fan and Li(2001)}]{Fan:Li:vari:2001}
Fan, J., and Li, R. (2001), ``Variable Selection via Nonconcave Penalized
  Likelihood and Its Oracle Properties,'' \textit{Journal of the American
  Statistical Association}, 96, 1348--1360.

\bibitem[{Fan and Liao(2012)}]{Fan:Liao:endo:2012}
Fan, J., and Liao, Y. (2012), ``Endogeneity in Ultrahigh Dimension,''
  unpublished manuscript. Available at arXiv:1204.5536.

\bibitem[{Fan and Lv(2010)}]{Fan:Lv:sele:2010}
Fan, J., and Lv, J. (2010), ``A Selective Overview of Variable Selection in
  High Dimensional Feature Space,'' \textit{Statistica Sinica}, 20, 101--148.

\bibitem[{Fan and Lv(2011)}]{Fan:Lv:nonc:2011}
---------\quad (2011), ``Nonconcave Penalized Likelihood With
  {NP}-Dimensionality,'' \textit{IEEE Transactions on Information Theory}, 57,
  5467--5484.

\bibitem[{Fusi, Stegle,  and Lawrence(2012)Fusi, Stegle, and
  Lawrence}]{Fusi:Steg:Lawr:join:2012}
Fusi, N., Stegle, O., and Lawrence, N.~D. (2012), ``Joint Modelling of
  Confounding Factors and Prominent Genetic Regulators Provides Increased
  Accuracy in Genetical Genomics Studies,'' \textit{PLoS Computational
  Biology}, 8, e1002330.

\bibitem[{Gautier and Tsybakov(2011)}]{Gaut:Tsyb:high:2011}
Gautier, E., and Tsybakov, A. (2011), ``High-Dimensional Instrumental Variables
  Regression and Confidence Sets,'' unpublished manuscript. Available at
  arXiv:1105.2454.

\bibitem[{G\"{o}ring(2012)}]{Gori:tiss:2012}
G\"{o}ring, H. H.~H. (2012), ``Tissue Specificity of Genetic Regulation of Gene
  Expression,'' \textit{Nature Genetics}, 44, 1077--1078.

\bibitem[{Gray et~al.(2008)Gray, O'Brien, D'Alessio, Brehm, and
  Deeg}]{Gray:OBri:DAle:Breh:Deeg:plas:2008}
Gray, D.~L., O'Brien, K.~D., D'Alessio, D.~A., Brehm, B.~J., and Deeg, M.~A.
  (2008), ``Plasma Glycosylphosphatidylinositol-Specific Phospholipase {D}
  Predicts the Change in Insulin Sensitivity in Response to a Low-Fat But Not a
  Low-Carbohydrate Diet in Obese Women,'' \textit{Metabolism Clinical and
  Experimental}, 57, 473--478.

\bibitem[{Hansen, Hausman,  and Newey(2008)Hansen, Hausman, and
  Newey}]{Hans:Haus:Newe:esti:2008}
Hansen, C., Hausman, J., and Newey, W. (2008), ``Estimation With Many
  Instrumental Variables,'' \textit{Journal of Business and Economic
  Statistics}, 26, 398--422.

\bibitem[{Heckman(1978)}]{Heck:dumm:1978}
Heckman, J.~J. (1978), ``Dummy Endogenous Variables in a Simultaneous Equation
  System,'' \textit{Econometrica}, 46, 931--959.

\bibitem[{Horn and Johnson(1985)}]{Horn:John:matr:1985}
Horn, R.~A., and Johnson, C.~R. (1985), \textit{Matrix Analysis}, New York:
  Cambridge University Press.

\bibitem[{Jaenisch and Bird(2003)}]{Jaen:Bird:epig:2003}
Jaenisch, R., and Bird, A. (2003), ``Epigenetic Regulation of Gene Expression:
  How the Genome Integrates Intrisic and Environmental Signals,''
  \textit{Nature Genetics}, 33 Suppl., 245--254.

\bibitem[{Lawlor et~al.(2008)Lawlor, Harbord, Sterne, Timpson, and
  Davey~Smith}]{Lawl:Harb:Ster:Timp:Dave:mend:2008}
Lawlor, D.~A., Harbord, R.~M., Sterne, J. A.~C., Timpson, N., and Davey~Smith,
  G. (2008), ``Mendelian Randomization: Using Genes as Instruments for Making
  Causal Inferences in Epidemiology,'' \textit{Statistics in Medicine}, 27,
  1133--1163.

\bibitem[{Leek and Storey(2007)}]{Leek:Stor:capt:2007}
Leek, J.~T., and Storey, J.~D. (2007), ``Capturing Heterogeneity in Gene
  Expression Studies by Surrogate Variable Analysis,'' \textit{PLoS Genetics},
  3, e161.

\bibitem[{Lin and Zeng(2011)}]{Lin:Zeng:corr:2011}
Lin, D.~Y., and Zeng, D. (2011), ``Correcting for Population Stratification in
  Genomewide Association Studies,'' \textit{Journal of the American Statistical
  Association}, 106, 997--1008.

\bibitem[{Lin and Lv(2013)}]{Lin:Lv:high:2013}
Lin, W., and Lv, J. (2013), ``High-Dimensional Sparse Additive Hazards
  Regression,'' \textit{Journal of the American Statistical Association}, 108,
  247--264.

\bibitem[{Lu, Goldberg,  and Fine(2012)Lu, Goldberg, and
  Fine}]{Lu:Gold:Fine:on:2012}
Lu, W., Goldberg, Y., and Fine, J.~P. (2012), ``On the Robustness of the
  Adaptive Lasso to Model Misspecification,'' \textit{Biometrika}, 99,
  717--731.

\bibitem[{Lv and Fan(2009)}]{Lv:Fan:unif:2009}
Lv, J., and Fan, Y. (2009), ``A Unified Approach to Model Selection and Sparse
  Recovery Using Regularized Least Squares,'' \textit{The Annals of
  Statistics}, 37, 3498--3528.

\bibitem[{Lv and Liu(2013)}]{Lv:Liu:mode:2013}
Lv, J., and Liu, J.~S. (2013), ``Model Selection Principles in Misspecified
  Models,'' \textit{Journal of the Royal Statistical Society, \emph{Series B}},
  to appear.

\bibitem[{Mazumder, Friedman,  and Hastie(2011)Mazumder, Friedman, and
  Hastie}]{Mazu:Frie:Hast:spar:2011}
Mazumder, R., Friedman, J.~H., and Hastie, T. (2011), ``\emph{SparseNet}:
  Coordinate Descent With Nonconvex Penalties,'' \textit{Journal of the
  American Statistical Association}, 106, 1125--1138.

\bibitem[{Meinshausen and B\"{u}hlmann(2010)}]{Mein:Buhl:stab:2010}
Meinshausen, N., and B\"{u}hlmann, P. (2010), ``Stability Selection''  (with
  discussion), \textit{Journal of the Royal Statistical Society, \emph{Series
  B}}, 72, 417--473.

\bibitem[{Metcalf et~al.(2000)Metcalf, Greenhalgh, Viney, Willson, Starr,
  Nicola, Hilton, and Alexander}]{Metc:Gree:Vine:Will:Star:giga:2000}
Metcalf, D., Greenhalgh, C.~J., Viney, E., Willson, T.~A., Starr, R., Nicola,
  N.~A., Hilton, D.~J., and Alexander, W.~S. (2000), ``Gigantism in Mice
  Lacking Suppressor of Cytokine Signalling-2,'' \textit{Nature}, 405,
  1069--1073.

\bibitem[{Okui(2011)}]{Okui:inst:2011}
Okui, R. (2011), ``Instrumental Variable Estimation in the Presence of Many
  Moment Conditions,'' \textit{Journal of Econometrics}, 165, 70--86.

\bibitem[{Okumura et~al.(1994)Okumura, Fukagawa, Tso, Taylor, and
  Pappas}]{Okum:Fuka:Tso:Tayl:Papp:intr:1994}
Okumura, T., Fukagawa, K., Tso, P., Taylor, I.~L., and Pappas, T. (1994),
  ``Intracisternal Injection of Apolipoprotein {A-IV} Inhibits Gastric
  Secretion in Pylorus-Ligated Conscious Rats,'' \textit{Gastroenterology},
  107, 1861--1864.

\bibitem[{Raskutti, Wainwright,  and Yu(2011)Raskutti, Wainwright, and
  Yu}]{Rask:Wain:Yu:mini:2011}
Raskutti, G., Wainwright, M.~J., and Yu, B. (2011), ``Minimax Rates of
  Estimation for High-Dimensional Linear Regression Over $\ell_q$-Balls,''
  \textit{IEEE Transactions on Information Theory}, 57, 6976--6994.

\bibitem[{Rocha et~al.(2004)Rocha, Eisen, Van~Vleck, and
  Pomp}]{Roch:Eise:Vlec:Pomp:larg:2004}
Rocha, J.~L., Eisen, E.~J., Van~Vleck, L.~D., and Pomp, D. (2004), ``A
  Large-Sample {QTL} Study in Mice: {I}. Growth,'' \textit{Mammalian Genome},
  15, 83--99.

\bibitem[{Rothman, Levina,  and Zhu(2010)Rothman, Levina, and
  Zhu}]{Roth:Levi:Zhu:spar:2010}
Rothman, A.~J., Levina, E., and Zhu, J. (2010), ``Sparse Multivariate
  Regression With Covariance Estimation,'' \textit{Journal of Computational and
  Graphical Statistics}, 19, 947--962.

\bibitem[{Stranger et~al.(2012)Stranger, Montgomery, Dimas, Parts, Stegle,
  Ingle, Sekowska, Davey~Smith, Evans, Gutierrez-Arcelus, Price, Raj, Nisbett,
  Nica, Beazley, Durbin, Deloukas, and
  Dermitzakis}]{Stra:Mont:Dima:Part:Steg:patt:2012}
Stranger, B.~E., Montgomery, S.~B., Dimas, A.~S., Parts, L., Stegle, O., Ingle,
  C.~E., Sekowska, M., Davey~Smith, G., Evans, D., Gutierrez-Arcelus, M.,
  Price, A., Raj, T., Nisbett, J., Nica, A.~C., Beazley, C., Durbin, R.,
  Deloukas, P., and Dermitzakis, E.~T. (2012), ``Patterns of \emph{Cis}
  Regulatory Variation in Diverse Human Populations,'' \textit{PLoS Genetics},
  8, e1002639.

\bibitem[{Tibshirani(1996)}]{Tibs:regr:1996}
Tibshirani, R. (1996), ``Regression Shrinkage and Selection via the Lasso,''
  \textit{Journal of the Royal Statistical Society, \emph{Series B}}, 58,
  267--288.

\bibitem[{Troyanskaya et~al.(2001)Troyanskaya, Cantor, Sherlock, Brown, Hastie,
  Tibshirani, Botstein, and Altman}]{Troy:Cant:Sher:Brow:Hast:miss:2001}
Troyanskaya, O., Cantor, M., Sherlock, G., Brown, P., Hastie, T., Tibshirani,
  R., Botstein, D., and Altman, R.~B. (2001), ``Missing Value Estimation
  Methods for {DNA} Microarrays,'' \textit{Bioinformatics}, 17, 520--525.

\bibitem[{Tso, Sun,  and Liu(2004)Tso, Sun, and Liu}]{Tso:Sun:Liu:gast:2004}
Tso, P., Sun, W., and Liu, M. (2004), ``Gastrointestinal Satiety Signals {IV}.
  Apolipoprotein {A-IV},'' \textit{American Journal of Physiology
  Gastrointestinal and Liver Physiology}, 286, G885--G890.

\bibitem[{van Nas et~al.(2010)van Nas, Ingram-Drake, Sinsheimer, Wang, Schadt,
  Drake, and Lusis}]{Nas:Ingr:Sins:Wang:Scha:expr:2010}
van Nas, A., Ingram-Drake, L., Sinsheimer, J.~S., Wang, S.~S., Schadt, E.~E.,
  Drake, T., and Lusis, A.~J. (2010), ``Expression Quantitative Trait Loci:
  Replication, Tissue- and Sex-Specificity in Mice,'' \textit{Genetics}, 185,
  1059--1068.

\bibitem[{Wang et~al.(2006)Wang, Yehya, Schadt, Wang, Drake, and
  Lusis}]{Wang:Yehy:Scha:Wang:Drak:gene:2006}
Wang, S., Yehya, N., Schadt, E.~E., Wang, H., Drake, T.~A., and Lusis, A.~J.
  (2006), ``Genetic and Genomic Analysis of a Fat Mass Trait With Complex
  Inheritance Reveals Marked Sex Specificity,'' \textit{PLoS Genetics}, 2, e15.

\bibitem[{Wang et~al.(2007)Wang, Schadt, Wang, Wang, Ingram-Drake, Shi, Drake,
  and Lusis}]{Wang:Scha:Wang:Wang:Ingr:inde:2007}
Wang, S.~S., Schadt, E.~E., Wang, H., Wang, X., Ingram-Drake, L., Shi, W.,
  Drake, T.~A., and Lusis, A.~J. (2007), ``Identification of Pathways for
  Atherosclerosis in Mice: Integration of Quantitative Trait Locus Analysis and
  Global Gene Expression Data,'' \textit{Circulation Research}, 101, e11--e30.

\bibitem[{Wheatcroft et~al.(2007)Wheatcroft, Kearney, Shah, Ezzat, Miell, Modo,
  Williams, Cawthorn, Medina-Gomez, Vidal-Puig, Sethi, and
  Crossey}]{Whea:Kear:Shah:Ezza:Miel:igf-:2007}
Wheatcroft, S.~B., Kearney, M.~T., Shah, A.~M., Ezzat, V.~A., Miell, J.~R.,
  Modo, M., Williams, S. C.~R., Cawthorn, W.~P., Medina-Gomez, G., Vidal-Puig,
  A., Sethi, J.~K., and Crossey, P.~A. (2007), ``{IGF}-Binding Protein-2
  Protects Against the Development of Obesity and Insulin Resistance,''
  \textit{Diabetes}, 56, 285--294.

\bibitem[{Ye and Zhang(2010)}]{Ye:Zhan:rate:2010}
Ye, F., and Zhang, C.-H. (2010), ``Rate Minimaxity of the {L}asso and {D}antzig
  Selector for the $\ell_q$ Loss in $\ell_r$ Balls,'' \textit{Journal of
  Machine Learning Research}, 11, 3519--3540.

\bibitem[{Zhang(2010)}]{Zhan:near:2010}
Zhang, C.-H. (2010), ``Nearly Unbiased Variable Selection Under Minimax Concave
  Penalty,'' \textit{The Annals of Statistics}, 38, 894--942.

\bibitem[{Zhao and Yu(2006)}]{Zhao:Yu:on:2006}
Zhao, P., and Yu, B. (2006), ``On Model Selection Consistency of {L}asso,''
  \textit{Journal of Machine Learning Research}, 7, 2541--2563.

\bibitem[{Zou(2006)}]{Zou:adap:2006}
Zou, H. (2006), ``The Adaptive Lasso and Its Oracle Properties,''
  \textit{Journal of the American Statistical Association}, 101, 1418--1429.

\end{thebibliography}

\end{document}


\title{Supplementary Material for ``Regularization Methods for High-Dimensional Instrumental Variables Regression With an Application to Genetical Genomics''}
\author{Wei Lin, Rui Feng, and Hongzhe Li}
\date{}
\maketitle

\appendix
\addtocounter{section}{19}
\setcounter{equation}{0}
\numberwithin{equation}{section}
\numberwithin{lem}{section}

\subsection*{Proof of Theorem 4}
We first present two lemmas that are essential to the proof of Theorem 4, which concern the concentration of the empirical covariance matrix $\what\bC$ around its population version $\bC$ and the score vector
\[
\frac{1}{n}\what\bX^T(\by-\what\bX\bbeta_0)=\frac{1}{n}\what\bX^T\bfeta-\frac{1}{n}\what\bX^T(\what\bX-\bX)\bbeta_0
\]
around zero. These lemmas can be viewed as generalizations of Lemma A.3 and inequality (A.15), respectively. For ease of presentation, we condition on the event of probability $1-\pi_0$ that the two error bounds in Condition (C4) hold, and incorporate the probability $\pi_0$ into the result by the union bound.

\begin{lem}\label{lem:conc_mat}
Under Conditions (C4)--(C6), if $\mu_0>0$ and the first-stage error bounds $e_1$ and $e_2$ satisfy
\begin{equation}\label{eq:rate_gen'}
s(2Le_1+e_2)\le\frac{\alpha}{(4-\alpha)\varphi}\wedge\frac{(\mu_0/2)^2}{s},
\end{equation}
then with probability at least $1-\pi_0$, the following inequalities holds:
\begin{align}
\|(\what\bC_{SS})^{-1}\|_{\infty}&\le\frac{4-\alpha}{2(2-\alpha)}\varphi,\label{eq:conc_mat1}\\
\|\what\bC_{S^cS}(\what\bC_{SS})^{-1}\|_{\infty}&\le\left\{\left(1-\frac{\alpha}{2}\right)\frac{\rho'(0+)}{\rho'_{\mu}(b_0/2)}\right\}\wedge(2cn^{\nu}), \label{eq:conc_mat2}
\end{align}
and
\begin{equation}\label{eq:conc_mat3}
\Lambda_{\min}(\what\bC_{SS})>\mu\tau_0.
\end{equation}
\end{lem}

\begin{proof}
It follows from the arguments in the proof of Lemma A.1 and Condition (C4) that
\[
\max_{1\le i,j\le p}\frac{1}{n}|\what\bx_i^T\what\bx_j-(\bZ\bgamma_{0i})^T\bZ\bgamma_{0j}|\le2Le_1+e_2.
\]
Consequently, by the assumption \eqref{eq:rate_gen'},
\begin{equation}\label{eq:CSSgen}
\varphi\|\what\bC_{SS}-\bC_{SS}\|_{\infty}\le\varphi s(2Le_1+e_2)\le\frac{\alpha}{4-\alpha}
\end{equation}
and
\begin{equation}\label{eq:CScSgen}
\varphi\|\what\bC_{S^cS}-\bC_{S^cS}\|_{\infty}\le\frac{\alpha}{4-\alpha}.
\end{equation}
Then inequality \eqref{eq:conc_mat1} follows as in the proof of Lemma A.3.

To show inequality \eqref{eq:conc_mat2}, by \eqref{eq:conc_mat1}, \eqref{eq:CSSgen}, \eqref{eq:CScSgen}, and Condition (C6), we have
\begin{align*}
&\|\what\bC_{S^cS}(\what\bC_{SS})^{-1}-\bC_{S^cS}(\bC_{SS})^{-1}\|_{\infty}\\
&\quad\le\|\what\bC_{S^cS}-\bC_{S^cS}\|_{\infty}\|(\what\bC_{SS})^{-1}\|_{\infty}+\|\bC_{S^cS}(\bC_{SS})^{-1}\|_{\infty} \|\what\bC_{SS}-\bC_{SS}\|_{\infty}\|(\what\bC_{SS})^{-1}\|_{\infty}\\
&\quad\le\frac{\alpha}{(4-\alpha)\varphi}\frac{4-\alpha}{2(2-\alpha)}\varphi +\left[\left\{(1-\alpha)\frac{\rho'(0+)}{\rho_{\mu}'(b_0/2)}\right\}\wedge(cn^{\nu})\right]\frac{\alpha}{(4-\alpha)\varphi}\frac{4-\alpha}{2(2-\alpha)}\varphi\\
&\quad\le\frac{\alpha}{2(2-\alpha)}+\left\{\frac{\alpha(1-\alpha)}{2(2-\alpha)}\frac{\rho'(0+)}{\rho_{\mu}'(b_0/2)}\right\}\wedge\left(\frac{c}{2}n^{\nu} \right)\\
&\quad\le\left\{\frac{\alpha}{2}\frac{\rho'(0+)}{\rho_{\mu}'(b_0/2)}\right\}\wedge(cn^{\nu}),
\end{align*}
where we have used the inequalities $\rho'(0+)/\rho_{\mu}'(b_0/2)\ge1$ and $\alpha/\{2(2-\alpha)\}\le1/2\le cn^{\nu}/2$. This, along with Condition (C6), implies \eqref{eq:conc_mat2}.

Finally, it follows from the Hoffman--Wielandt inequality (Horn and Johnson 1985) and the assumption \eqref{eq:rate_gen'} that
\[
|\Lambda_{\min}(\what\bC_{SS})-\Lambda_{\min}(\bC_{SS})|^2\le\|\what\bC_{SS}-\bC_{SS}\|_F^2\le s^2(2Le_1+e_2)\le\left(\frac{\mu_0}{2}\right)^2.
\]
In view of the definition of $\mu_0$, inequality \eqref{eq:conc_mat3} follows. This completes the proof of the lemma.
\end{proof}

\begin{lem}\label{lem:conc_vec}
Under Conditions (C4)--(C6), if the first-stage error bounds satisfy $e_1=O(1)$ and $e_2=O(1)$, then there exist constants $c_0,c_1,c_2>0$ such that, if we choose
\[
\mu\ge C_0n^{\nu}\sqrt{\frac{\log p+\log q}{n}\vee e_2},
\]
where $C_0=c_0L\max(\sigma_{p+1},M\sigma_{\max},M)$, then with probability at least $1-\pi_0-c_1(pq)^{-c_2}$, it holds that
\begin{equation}\label{eq:mu_inf_gen}
\left\|\frac{1}{n}\what\bX^T\bfeta-\frac{1}{n}\what\bX^T(\what\bX-\bX)\bbeta_0\right\|_{\infty}<\frac{\alpha}{6cn^{\nu}}\mu\rho'(0+).
\end{equation}
\end{lem}

\begin{proof}
As in the proof of Lemma A.2, we write $n^{-1}\what\bX^T\bfeta-n^{-1}\what\bX^T(\what\bX-\bX)\bbeta_0=T_1+\cdots+T_6$. Letting $t_0=\alpha\mu\rho'(0+)/(6cn^{\nu})$, we bound the six terms similarly as follows:
\begin{gather*}
P\left(\|T_1\|_{\infty}\ge\frac{t_0}{6}\right)\le P\left(\left\|\frac{1}{n}\bZ^T\bfeta\right\|_{\infty}\ge\frac{t_0}{6e_1}\right)\le q\exp\left\{-\frac{n}{2\sigma_{p+1}^2}\left(\frac{t_0}{6e_1}\right)^2\right\},\\
P\left(\|T_2\|_{\infty}\ge\frac{t_0}{6}\right)\le P\left(\left\|\frac{1}{n}\bZ^T\bfeta\right\|_{\infty}\ge\frac{t_0}{6L}\right)\le q\exp\left\{-\frac{n}{2\sigma_{p+1}^2}\left(\frac{t_0}{6L}\right)^2\right\},\\
P\left(\|T_3\|_{\infty}\ge\frac{t_0}{6}\right)\le P\left(\max_{1\le i\le q,\,1\le j\le p}\left|\frac{1}{n}\bz_i^T\bve_j\right|_{\infty} \ge\frac{t_0}{6Me_1}\right)\le pq\exp\left\{-\frac{n}{2\sigma_{\max}^2}\left(\frac{t_0}{6Me_1}\right)^2\right\},\\
P\left(\|T_4\|_{\infty}\ge\frac{t_0}{6}\right)\le P\left(\max_{1\le i\le q,\,1\le j\le p}\left|\frac{1}{n}\bz_i^T\bve_j\right|_{\infty} \ge\frac{t_0}{6LM}\right)\le pq\exp\left\{-\frac{n}{2\sigma_{\max}^2}\left(\frac{t_0}{6LM}\right)^2\right\},\\
\|T_5\|_{\infty}\le M\max_{1\le i,j\le p}\frac{1}{n}\|\bZ(\what\bgamma_i-\bgamma_{0i})\|_2\|\bZ(\what\bgamma_j-\bgamma_{0j})\|_2\le Me_2,
\end{gather*}
and
\[
\|T_6\|_{\infty}\le LM\max_{1\le j\le p}\frac{1}{\sqrt{n}}\|\bZ(\what\bgamma_j-\bgamma_{0j})\|_2\le LM\sqrt{e_2}.
\]
Combining these bounds and in view of the assumptions $e_1=O(1)$ and $e_2=O(1)$, there exist constants $c_0,c_1,c_2>0$ such that, if we choose
\[
\mu\ge C_0n^{\nu}\sqrt{\frac{\log p+\log q}{n}\vee e_2},
\]
where $C_0=c_0L\max(\sigma_{p+1},M\sigma_{\max},M)$, then with probability at least $1-\pi_0-c_1(pq)^{-c_2}$, the desired inequality holds. The completes the proof of the lemma.
\end{proof}

\begin{proof}[Proof of Theorem 4]
One can easily show that $\what\bbeta\in\mathbb{R}^p$ is a strict local minimizer of problem (4) if the following conditions hold:
\begin{gather}
\frac{1}{n}\what\bX_{\what{S}}^T(\by-\what\bX\what\bbeta)=\mu\rho'_{\mu}(|\what\bbeta_{\what{S}}|)\circ\sgn(\what\bbeta_{\what{S}}),\label{eq:opt_cond1}\\
\left\|\frac{1}{n}\what\bX_{\what{S}^c}^T(\by-\what\bX\what\bbeta)\right\|_{\infty}<\mu\rho'(0+),\label{eq:opt_cond2}
\end{gather}
and
\begin{equation}\label{eq:opt_cond3}
\Lambda_{\min}(\what\bC_{\what{S}\what{S}})>\mu\tau(\rho_{\mu};\what\bbeta_{\what{S}}),
\end{equation}
where $\circ$ denotes the Hadamard (entrywise) product, and $|\cdot|$, $\rho'_{\mu}(\cdot)$, and $\sgn(\cdot)$ are applied componentwise. It suffices to find a $\what\bbeta\in\mathbb{R}^p$ with the desired properties such that conditions \eqref{eq:opt_cond1}--\eqref{eq:opt_cond3} hold. Let $\what\bbeta_{S^c}=\bzero$. The idea of the proof is to first determine $\what\bbeta_S$ from \eqref{eq:opt_cond1}, and then show that thus obtained $\what\bbeta$ also satisfies \eqref{eq:opt_cond2} and \eqref{eq:opt_cond3}.

From now on, we condition on the event of probability at least $1-\pi_0-c_1(pq)^{-c_2}$ that the inequalities in Lemmas \ref{lem:conc_mat} and \ref{lem:conc_vec} hold. Using similar arguments to those in the proof of Theorem 3, \eqref{eq:opt_cond1} with $\what{S}$ replaced by $S$ can be written in the form
\begin{equation}\label{eq:fix_pt}
\what\bbeta_S-\bbeta_{0S}=(\what\bC_{SS})^{-1}\left\{\frac{1}{n}\what\bX_S^T\bfeta-\frac{1}{n}\what\bX_S^T(\what\bX_S-\bX_S)\bbeta_{0S} -\mu\rho'_{\mu}(|\what\bbeta_S|)\circ\sgn(\what\bbeta_S)\right\}.
\end{equation}
Define the function $f\colon\mathbb{R}^s\to\mathbb{R}^s$ by $f(\btheta)=\bbeta_{0S}+(\what\bC_{SS})^{-1}\{n^{-1}\what\bX_S^T\bfeta-n^{-1}\what\bX_S^T(\what\bX_S-\bX_S)\bbeta_{0S} -\mu\rho'_{\mu}(|\btheta|)\circ\sgn(\btheta)\}$, and let $\mathcal{K}$ denote the hypercube $\{\btheta\in\mathbb{R}^s\colon\|\btheta-\bbeta_{0S}\|_{\infty}\le 7\varphi\mu\rho'(0+)/4\}$. It follows from \eqref{eq:conc_mat1}, \eqref{eq:mu_inf_gen}, and Condition (C4) that, for $\btheta\in\mathcal{K}$,
\begin{align*}
\|f(\btheta)-\bbeta_{0S}\|_{\infty}&\le\|(\what\bC_{SS})^{-1}\|_{\infty}\left\{\left\|\frac{1}{n}\what\bX_S^T\bfeta-\frac{1}{n}\what\bX_S^T(\what\bX_S-\bX_S) \bbeta_{0S}\right\|_{\infty}+\mu\rho'(0+)\right\}\\
&\le\frac{4-\alpha}{2(2-\alpha)}\varphi\left\{\frac{\alpha}{6cn^{\nu}}\mu\rho'(0+)+\mu\rho'(0+)\right\}\\
&\le\frac{3}{2}\varphi\left\{\frac{1}{6}\mu\rho'(0+)+\mu\rho'(0+)\right\}=\frac{7}{4}\varphi\mu\rho'(0+),
\end{align*}
that is, $f(\mathcal{K})\subset\mathcal{K}$. Also, the last inequality and the assumption (14) imply that for $\btheta\in\mathcal{K}$, $\|\btheta-\bbeta_{0S}\|_{\infty}\le b_0/2$, and hence $\sgn(\btheta)=\sgn(\bbeta_{0S})$. Thus, in view of Condition (C4), $f$ is a continuous function on the convex, compact hypercube $\mathcal{K}$. An application of Brouwer's fixed point theorem yields that equation \eqref{eq:fix_pt} has a solution $\what\bbeta_S$ in $\mathcal{K}$. Moreover, $\sgn(\what\bbeta_S)=\sgn(\bbeta_{0S})$, so that $\what{S}=S$. Therefore, we have found a $\what\bbeta$ that satisfies the desired properties and \eqref{eq:opt_cond1}.

To verify that $\what\bbeta$ satisfies \eqref{eq:opt_cond2}, by substituting \eqref{eq:fix_pt}, we write
\begin{align*}
\frac{1}{n}\what\bX_{S^c}^T(\by-\what\bX\what\bbeta)&=\frac{1}{n}\what\bX_{S^c}^T\bfeta-\frac{1}{n}\what\bX_{S^c}^T(\what\bX_S-\bX_S)\bbeta_{0S}\\ &\pheq{}-\what\bC_{S^cS}(\what\bC_{SS})^{-1}\left\{\frac{1}{n}\what\bX_S^T\bfeta-\frac{1}{n}\what\bX_S^T(\what\bX_S-\bX_S)\bbeta_{0S} -\mu\rho'_{\mu}(|\what\bbeta_S|)\circ\sgn(\what\bbeta_S)\right\}.
\end{align*}
Also, we have $\|\what\bbeta_S\|_{\infty}=\|\what\bbeta_{0S}+(\what\bbeta_S-\bbeta_{0S})\|_{\infty}\ge\|\what\bbeta_{0S}\|_{\infty}-\|\what\bbeta_S- \bbeta_{0S}\|_{\infty}\ge b_0-b_0/2=b_0/2$. This, together with \eqref{eq:conc_mat2}, \eqref{eq:mu_inf_gen}, and Condition (C4), leads to
\begin{align*}
\left\|\frac{1}{n}\what\bX_{S^c}^T(\by-\what\bX\bbeta)\right\|_{\infty} &\le\left\|\frac{1}{n}\what\bX_{S^c}^T\bfeta-\frac{1}{n}\what\bX_{S^c}^T(\what\bX_S-\bX_S)\bbeta_{0S}\right\|_{\infty} +\|\what\bC_{S^cS}(\what\bC_{SS})^{-1}\|_{\infty}\\
&\pheq{}\times\left\{\left\|\frac{1}{n}\what\bX_S^T\bfeta-\frac{1}{n}\what\bX_S^T(\what\bX_S-\bX_S)\bbeta_{0S} \right\|_{\infty}+\mu\rho'_{\mu}(b_0/2)\right\}\\
&<\frac{\alpha}{6cn^{\nu}}\mu\rho'(0+)+2cn^{\nu}\cdot\frac{\alpha}{6cn^{\nu}}\mu\rho'(0+) +\left(1-\frac{\alpha}{2}\right)\frac{\rho'(0+)}{\rho'_{\mu}(b_0/2)}\cdot\mu\rho'_{\mu}(b_0/2)\\
&\le\frac{\alpha}{6}\mu\rho'(0+)+\frac{\alpha}{3}\mu\rho'(0+)+\left(1-\frac{\alpha}{2}\right)\mu\rho'(0+)=\mu\rho'(0+).
\end{align*}

Finally, it follows from \eqref{eq:conc_mat3} and the definition of $\tau_0$ that $\Lambda_{\min}(\what\bC_{SS})>\mu\tau_0\ge\mu\tau(\rho_{\mu};\what\bbeta_S)$, which verifies \eqref{eq:opt_cond3} and completes the proof.
\end{proof}